\documentclass[12pt, onecolumn]{IEEEtran} 

\usepackage[utf8]{inputenc} 
\usepackage[T1]{fontenc}
\usepackage{url}
\usepackage{ifthen}
\usepackage{cite}
\usepackage[cmex10]{amsmath} 
\usepackage{amssymb}
\usepackage{algorithm}
\usepackage{algorithmic}
\usepackage{subfigure}
\usepackage{tikz}
\usepackage{algorithm}
\usepackage{algorithmic}
\usepackage{amsmath}
\usepackage{amsthm}
\usepackage{amssymb}
\usepackage{ecltree}
\usepackage{enumerate}
\usepackage{mathrsfs}
\usepackage{xcolor}
\usepackage{comment}
\usepackage{tikz}
\usepackage{wrapfig}
\usepackage{mathtools}
\usetikzlibrary{trees}

\usepackage{here}
\usepackage{cases}
\usepackage{stackengine}
\usepackage{longtable}

\def\BibTeX{{\rm B\kern-.05em{\sc i\kern-.025em b}\kern-.08em
    T\kern-.1667em\lower.7ex\hbox{E}\kern-.125emX}}

\newtheorem{definition}{Definition}
\newtheorem{lemma}{Lemma}
\newtheorem{theorem}{Theorem}

\newtheorem{remark}{Remark}
\newtheorem{example}{Example}

\def\trans{\mathrm{\tau}}

\def\PREF{\mathcal{P}}

\def\hdec{\text{-}\mathrm{dec}}
\def\ext{\mathrm{ext}}

\def\reg{\mathrm{reg}}
\def\opt{\mathrm{opt}}

\def\AIFV{\mathrm{AIFV}}
\def\pref{\mathrm{pref}}
\def\suff{\mathrm{suff}}
\def\kernel{\mathcal{R}}
\def \assign{\mathcal{S}}
\def \prefset{\mathscr{P}}

\newcommand{\argmax}{\mathop{\rm arg~max}\limits}
\newcommand{\argmin}{\mathop{\rm arg~min}\limits}

\begin{document}

\title{The Optimality of AIFV Codes\\ in the Class of $2$-bit Delay Decodable Codes} 

\author{Kengo Hashimoto, Ken-ichi Iwata \thanks{University of Fukui, Japan. E-mail: \{khasimot, k-iwata\}@u-fukui.ac.jp}}


\maketitle

\begin{abstract}
AIFV (almost instantaneous fixed-to-variable length) codes are noiseless source codes that can attain a shorter average codeword length than Huffman codes by allowing a time-variant encoder with two code tables and a decoding delay of at most $2$ bits.
First, we consider a general class of noiseless source codes, called $k$-bit delay decodable codes, in which one allows a finite number of code tables and a decoding delay of at most $k$ bits for $k \geq 0$.
Then we prove that AIFV codes achieve the optimal average codeword length in the $2$-bit delay decodable codes class.
\end{abstract}


\section{Introduction}
\label{sec:introduction}

Huffman codes \cite{Huffman1952} achieve the optimal average codeword length in the class of instantaneous (i.e., uniquely decodable without decoding delay) codes.
McMillan's theorem \cite{McMillan1956} implies that Huffman codes achieve the optimal average codeword length also in the class of uniquely decodable codes.
However, McMillan's theorem implicitly assumes that a single code table is used for coding.
When multiple code tables and decoding delay of some bits are allowed, one can achieve a shorter average codeword length than Huffman codes. 
AIFV (almost instantaneous fixed-to-variable length) codes developed by Yamamoto, Tsuchihashi, and Honda \cite{Yamamoto2015} can attain a shorter average codeword length than Huffman codes by using a time-variant encoder with two code tables and allowing decoding delay of at most two bits.

AIFV codes are generalized to binary AIFV-$m$ codes \cite{Hu2017}, which can achieve a shorter average codeword length than AIFV codes for $m \geq 3$, allowing $m$ code tables and a decoding delay of at most $m$ bits.
The worst-case redundancy of AIFV and AIFV-$m$ codes are analyzed in \cite{Hu2017, Fujita2020} for $m=2,3,4,5$. 
The literature \cite{IY:ISITA16, IY:ITW17, Fujita2019, Fujita2018, ISIT2018, ISITA2018, Golin2019, Golin2020, ISIT2020, Golin2021, Golin2022, Sumigawa2017, Hashimoto2019, ITW2020} proposes the code construction and coding method of AIFV and AIFV-$m$ codes.
Extensions of AIFV-$m$ codes are proposed in \cite{Sugiura2018, Sugiura2022}.

The literature \cite{JSAIT2022} formalizes a binary encoder with a finite number of code tables as a \emph{code-tuple} and
introduces the class of code-tuples decodable with a delay of at most $k$ bits as the class of \emph{$k$-bit delay decodable codes},
which general properties are studied in \cite{IEICE2023}.
It is known that Huffman codes achieve the optimal average codeword length in the class of $1$-bit delay decodable code-tuples \cite{JSAIT2022}.
On the other hand, for the class of $2$-bit delay decodable code-tuples, only a partial result, limited to the case of two code tables, is known:
 AIFV codes achieve the optimal average codeword length in the class of $2$-bit delay decodable code-tuples with two code tables \cite{Hashimoto2021}.
This paper removes the constraint of two code tables and gives a complete result for the class of $2$-bit delay decodable code-tuples.
Namely, we prove that AIFV codes achieve the optimal average codeword length in the class of $2$-bit delay decodable codes with a finite number of code tables.
 
This paper is organized as follows.
\begin{itemize}
\item In Section \ref{sec:preliminary}, we prepare some notations, describe our data compression scheme, introduce some notions including $k$-bit delay decodable code-tuples, and show their basic properties.
\item In Section \ref{sec:optimality}, we prove the main result, the optimality of AIFV codes in the class of $2$-bit delay decodable code-tuples.
\item Lastly, we conclude this paper in Section \ref{sec:conclusion}.
\end{itemize}
To clarify the flow of the discussion, we relegate the proofs of most of the lemmas to the appendix.
The main notations are listed in Appendix \ref{subsec:notation}.

\section{Preliminaries}
\label{sec:preliminary}

This paper focuses on binary coding in which a source sequence over a finite alphabet $\mathcal{S}$ is encoded to a codeword sequence over $\mathcal{C} \coloneqq \{0, 1\}$.

We first define some notations based on \cite{JSAIT2022, IEICE2023}.
Let $|\mathcal{A}|$ denote the cardinality of a finite set $\mathcal{A}$.
Let $\mathcal{A}^k$ (resp. $\mathcal{A}^{\ast}$, $\mathcal{A}^{+}$) denote the set of all sequences of length $k$ (resp. finite length, finite positive length) over a set $\mathcal{A}$.
Namely, $\mathcal{A}^{+} = \mathcal{A}^{\ast} \setminus \{\lambda\}$, where $\lambda$ denotes the empty sequence.
The length of a sequence $\pmb{x}$ is denoted by $|\pmb{x}|$, in particular, $|\lambda| = 0$.
We say $\pmb{x} \preceq \pmb{y}$ if $\pmb{x}$ is a prefix of $\pmb{y}$, that is, there exists a sequence $\pmb{z}$, possibly $\pmb{z} = \lambda$, such that $\pmb{y} = \pmb{x}\pmb{z}$.
Also, we say $\pmb{x} \prec \pmb{y}$ if $\pmb{x} \preceq \pmb{y}$ and $\pmb{x} \neq \pmb{y}$.
For a non-empty sequence $ \pmb{x} = x_1x_2\ldots x_{n}$, we define
$\pref(\pmb{x}) = x_1x_2\ldots x_{n-1}$ and $\suff(\pmb{x}) = x_2\ldots x_{n-1}x_{n}$.
Namely, $\pref(\pmb{x})$ (resp. $\suff(\pmb{x})$) is the sequence obtained by deleting the last (resp. first) letter from $\pmb{x}$.
For $c \in \mathcal{C}$, the negation of $c$ is denoted by $\bar{c}$, that is, $\bar{0} \coloneqq 1$ and $\bar{1} \coloneqq 0$.
For $c \in \mathcal{C}$ and $\mathcal{A} \subseteq \mathcal{C}^{\ast}$, we define $c\mathcal{A} \coloneqq \{c\pmb{b} : \pmb{b} \in \mathcal{A}\} \subseteq \mathcal{C}^{\ast}$.
The main notations are listed in Appendix \ref{subsec:notation}.

In this paper, we consider a data compression system consisting of a source, an encoder, and a decoder, described as follows.
\begin{itemize}
\item Source: We consider an i.i.d.\ source, which outputs a sequence $\pmb{x} = x_1x_2\ldots x_n$ of symbols of the source alphabet $\mathcal{S} = \{s_1, s_2, \ldots, s_{\sigma}\}$, where $n$ and $\sigma$ denote the length of $\pmb{x}$ and the alphabet size, respectively.
In this paper, we assume $\sigma \geq 2$.
Each source output follows a fixed probability distribution $(\mu(s_1), \mu(s_2), \ldots, \mu(s_{\sigma}))$, where $\mu(s_i)$ is the probability of occurrence of $s_i$ for $i = 1, 2, \ldots, \sigma$.
More precisely, we fix a real-valued function $\mu : \mathcal{S} \rightarrow \mathbb{R}$ such that $\sum_{s \in \mathcal{S}} \mu(s) = 1$ and $0 < \mu(s) \leq 1$ for any $s \in \mathcal{S}$.
Note that we exclude the case where $\mu(s) = 0$ for some $s \in \mathcal{S}$ without loss of generality.

\item Encoder: The encoder has $m$ fixed code tables $f_0, f_1, \ldots, f_{m-1} : \mathcal{S} \rightarrow \mathcal{C}^{\ast}$.
The encoder reads the source sequence $\pmb{x} \in \mathcal{S}^{\ast}$ symbol by symbol from the beginning of $\pmb{x}$ and encodes them according to the code tables.
For the first symbol $x_1$, we use an arbitrarily chosen code table from $f_0, f_1, \ldots, f_{m-1}$.
For $x_2, x_3, \ldots, x_n$, we determine which code table to use to encode them according to $m$ fixed mappings $\trans_0, \trans_1, \ldots, \trans_{m-1} : \mathcal{S} \rightarrow [m] \coloneqq \{0, 1, 2, \ldots, m-1\}$.
More specifically, if the previous symbol $x_{i-1}$ is encoded by the code table $f_j$, then the current symbol $x_i$ is encoded by the code table $f_{\trans_j(x_{i-1})}$.
Hence, if we use the code table $f_i$ to encode $x_1$, then a source sequence $\pmb{x} = x_1x_2\ldots x_n$ is encoded to a codeword sequence $f(\pmb{x}) \coloneqq f_{i_1}(x_1)f_{i_2}(x_n)\ldots f_{i_n}(x_n)$, where
\begin{equation}
i_j \coloneqq
\begin{cases}
i &\,\,\text{if}\,\, j = 1,\\
\trans_{i_{j-1}}(x_{j-1})  &\,\,\text{if}\,\, j \geq 2
\end{cases}
\end{equation}
for $j = 1, 2, \ldots, n$.

\item Decoder: 
The decoder reads the codeword sequence $f(\pmb{x})$ bit by bit from the beginning of $f(\pmb{x})$.
Each time the decoder reads a bit, the decoder recovers as long prefix of $\pmb{x}$ as the decoder can uniquely identify from the prefix of $f(\pmb{x})$ already read.
We assume that the encoder and decoder share the index $i_1$ of the code table used to encode $x_1$ in advance.
\end{itemize}

\subsection{Code-tuples}
\label{subsec:treepair}

The behavior of the encoder and decoder for a given source sequence is completely determined by $m$ code tables $f_0, f_1, \ldots, f_{m-1}$, and $m$ mappings $\trans_0, \trans_1, \ldots, \trans_{m-1}$ if we fix the index of code table used to encode $x_1$.
Accordingly, we name a tuple $F(f_0, f_1, \ldots, f_{m-1}, \trans_0, \trans_1, \ldots, \trans_{m-1})$ as a \emph{code-tuple} $F$ and identify a source code with a code-tuple $F$.

\begin{definition}
  \label{def:treepair}
Let $m$ be a positive integer.
An \emph{$m$-code-tuple} $F(f_0, f_1, \ldots, f_{m-1}, \trans_0, \trans_1, \ldots, \trans_{m-1})$ is a tuple of
$m$ mappings $f_0, f_1, \ldots, f_{m-1} : \mathcal{S} \rightarrow \mathcal{C}^{\ast}$ and $m$ mappings $\trans_0, \trans_1, \ldots, \trans_{m-1} : \mathcal{S} \rightarrow [m]$.

We define $\mathscr{F}^{(m)}$ as the set of all $m$-code-tuples.
Also, we define
$\mathscr{F} \coloneqq  \mathscr{F}^{(1)} \cup \mathscr{F}^{(2)} \cup \mathscr{F}^{(3)} \cup \cdots.$
An element of $\mathscr{F}$ is called a \emph{code-tuple}.
\end{definition}

We write $F(f_0, f_1, \ldots, f_{m-1}, \trans_0, \trans_1, \ldots, \trans_{m-1})$ also as $F(f, \trans)$ or $F$ for simplicity.
For $F \in \mathscr{F}^{(m)}$, let $|F|$ denote the number of code tables of $F$, that is, $|F| \coloneqq m$.
We write $[|F|] = \{0, 1, 2, \ldots, |F|-1\}$ as $[F]$ for simplicity.

\begin{definition}
\label{def:assign}
For $F(f, \trans) \in \mathscr{F}, i \in [F]$, and $\pmb{b} \in \mathcal{C}^{\ast}$, we define
$\assign_{F, i}(\pmb{b}) \coloneqq \{s \in \mathcal{S} : f_i(s) = \pmb{b}\}$.
\end{definition}
Note that $f_i$ is injective if and only if $|\assign_{F, i}(\pmb{b})| \leq 1$ holds for any $\pmb{b} \in \mathcal{C}^{\ast}$.

\begin{table}
\caption{Examples of a code-tuple}
\label{tab:code-tuple}
\centering

\begin{tabular}{c | lclclclc}
\hline
$s \in \mathcal{S}$ & $f^{(\alpha)}_0$ & $\trans^{(\alpha)}_0$ & $f^{(\alpha)}_1$ & $\trans^{(\alpha)}_1$ & $f^{(\alpha)}_2$ & $\trans^{(\alpha)}_2$\\
\hline
a & 110 & 0 & 010 & 0 & $\lambda$ & 2\\
b & $\lambda$ & 1 & 011 & 2 & $\lambda$ & 2\\
c & 110 & 2 & 1 & 2 & $\lambda$ & 2\\
d & 111 & 0 & 10 & 1 & $\lambda$ & 2\\
\hline
\end{tabular}
\vspace{8pt}

\begin{tabular}{c | lclclc}
\hline
$s \in \mathcal{S}$ & $f^{(\beta)}_0$ & $\trans^{(\beta)}_0$ & $f^{(\beta)}_1$ & $\trans^{(\beta)}_1$ & $f^{(\beta)}_2$ & $\trans^{(\beta)}_2$\\
\hline
a & 11 & 1 & 0110 & 1 & 10 & 2\\
b & $\lambda$ & 1 & 0110 & 1 & 11 & 2\\
c & 101 & 2 & 01 & 1 & 1000 & 2\\
d & 1011 & 1 & 0111 & 1 & 1001 & 2\\
\hline
\end{tabular}
\vspace{8pt}

\begin{tabular}{c | lclclc}
\hline
$s \in \mathcal{S}$ & $f^{(\gamma)}_0$ & $\trans^{(\gamma)}_0$ & $f^{(\gamma)}_1$ & $\trans^{(\gamma)}_1$ & $f^{(\gamma)}_2$ & $\trans^{(\gamma)}_2$\\
\hline
a & 01 & 0 & 00 & 1 & 1100 & 1\\
b & 10 & 1 & $\lambda$ & 0 & 1110 & 0\\
c & 0100 & 0 & 00111 & 1 & 111000 & 2\\
d & 01 & 2 & 00111 & 2 & 110 & 2\\
\hline
\end{tabular}
\vspace{8pt}

\begin{tabular}{c | lclclc}
\hline
$s \in \mathcal{S}$ & $f^{(\delta)}_0$ & $\trans^{(\delta)}_0$ & $f^{(\delta)}_1$ & $\trans^{(\delta)}_1$ & $f^{(\delta)}_2$ & $\trans^{(\delta)}_2$\\
\hline
a & 01 & 0 & 00 & 1 & 100 & 1\\
b & 10 & 1 & $\lambda$ & 0 & 110 & 0\\
c & 0100 & 0 & 00111 & 1 & 110001 & 2\\
d & 011 & 2 & 001111 & 2 & 101 & 2\\
\hline
\end{tabular}
\vspace{8pt}

\begin{tabular}{c | lclclc}
\hline
$s \in \mathcal{S}$ & $f^{(\epsilon)}_0$ & $\trans^{(\epsilon)}_0$ & $f^{(\epsilon)}_1$ & $\trans^{(\epsilon)}_1$ & $f^{(\epsilon)}_2$ & $\trans^{(\epsilon)}_2$\\
\hline
a & 01 & 0 & 00 & 1 & 00 & 1\\
b & 10 & 1 & $\lambda$ & 0 & 10 & 0\\
c & 0100 & 0 & 00111 & 1 & 100011 & 2\\
d & 0111 & 2 & 0011111 & 2 & 011 & 2\\
\hline
\end{tabular}
\vspace{8pt}

\begin{tabular}{c | lclclc}
\hline
$s \in \mathcal{S}$ & $f^{(\zeta)}_0$ & $\trans^{(\zeta)}_0$ & $f^{(\zeta)}_1$ & $\trans^{(\zeta)}_1$ & $f^{(\zeta)}_2$ & $\trans^{(\zeta)}_2$\\
\hline
a & 10 & 0 & 01 & 1 & 00 & 1\\
b & 11 & 1 & $\lambda$ & 0 & 10 & 0\\
c & 1000 & 0 & 01001 & 1 & 100011 & 2\\
d & 1001 & 2 & 0100100 & 2 & 011 & 2\\
\hline
\end{tabular}
\vspace{8pt}

\begin{tabular}{c | lclclc}
\hline
$s \in \mathcal{S}$ & $f^{(\eta)}_0$ & $\trans^{(\eta)}_0$ & $f^{(\eta)}_1$ & $\trans^{(\eta)}_1$ & $f^{(\eta)}_2$ & $\trans^{(\eta)}_2$\\
\hline
a & 01 & 0 & 01 & 1 & 00 & 1\\
b & 1 & 1 & 1 & 0 & 101 & 0\\
c & 0001 & 0 & 01001 & 1 & 100011 & 2\\
d & 001 & 2 & 0100100 & 2 & 011 & 2\\
\hline
\end{tabular}
\vspace{8pt}

\begin{tabular}{c | lclclc}
\hline
$s \in \mathcal{S}$ & $f^{(\theta)}_0$ & $\trans^{(\theta)}_0$ & $f^{(\theta)}_1$ & $\trans^{(\theta)}_1$ & $f^{(\theta)}_2$ & $\trans^{(\theta)}_2$\\
\hline
a & 01 & 0 & 01 & 1 & 10 & 1\\
b & 1 & 1 & 1 & 0 & 011 & 0\\
c & 0001 & 0 & 01001 & 1 & 010011 & 2\\
d & 001 & 2 & 0100100 & 2 & 111 & 2\\
\hline
\end{tabular}
\vspace{8pt}

\begin{tabular}{c | lclc}
\hline
$s \in \mathcal{S}$ & $f^{(\iota)}_0$ & $\trans^{(\iota)}_0$ & $f^{(\iota)}_1$ & $\trans^{(\iota)}_1$\\
\hline
a & 01 & 1 & 01 & 1\\
b & 1 & 1 & 1 & 0\\
c & 0001 & 0 & 01001 & 1\\
d & 001 & 1 & 0100100 & 1\\
\hline
\end{tabular}
\vspace{8pt}

\begin{tabular}{c | lclclc}
\hline
$s \in \mathcal{S}$ & $f^{(\kappa)}_0$ & $\trans^{(\kappa)}_0$ & $f^{(\kappa)}_1$ & $\trans^{(\kappa)}_1$\\
\hline
a & 100 & 0 & 1100 & 0\\
b & 00 & 0 & 11 & 1\\
c & 01 & 0 & 01 & 0\\
d & 1 & 1 & 10 & 0\\
\hline
\end{tabular}

\end{table}

\begin{example}
Table \ref{tab:code-tuple} shows examples of a code-tuple for $\mathcal{S} = \{\mathrm{a}, \mathrm{b}, \mathrm{c}, \mathrm{d}\}$.
The code-tuples $F^{(\alpha)}, F^{(\beta)}, \allowbreak F^{(\gamma)}, \ldots, F^{(\theta)}$ are $3$-code-tuples and the code-tuples $F^{(\iota)}$ and $F^{(\kappa)}$ are $2$-code-tuples.
We have
\begin{equation}
\assign_{F^{(\alpha)}, 0}(110) = \{\mathrm{a}, \mathrm{c}\},\quad
\assign_{F^{(\beta)}, 1}(00000000) = \emptyset,\quad
\assign_{F^{(\alpha)}, 2}(\lambda) = \{\mathrm{a}, \mathrm{b}, \mathrm{c}, \mathrm{d}\}.
\end{equation}
\end{example}

\begin{example}
\label{ex:encode}
We consider encoding of a source sequence $\pmb{x} = x_1x_2x_3x_4 \coloneqq \mathrm{badb}$ with the code-tuple $F(f, \trans) \coloneqq F^{(\gamma)}$ in Table \ref{tab:code-tuple}.
If $x_1 = \mathrm{b}$ is encoded with the code table $f_0$, then the encoding process is as follows.
\begin{itemize}
\item $x_1 = \mathrm{b}$ is encoded to $f_0(\mathrm{b}) = 10$. The index of the next code table is $\trans_0(\mathrm{b}) = 1$.
\item $x_2 = \mathrm{a}$ is encoded to $f_1(\mathrm{a}) = 00$. The index of the next code table is $\trans_1(\mathrm{a}) = 1$.
\item $x_3 = \mathrm{d}$ is encoded to $f_1(\mathrm{d}) = 00111$. The index of the next code table is $\trans_1(\mathrm{d}) = 2$.
\item $x_4 = \mathrm{b}$ is encoded to $f_2(\mathrm{b}) = 1110$. The index of the next code table is $\trans_2(\mathrm{b}) = 0$.
\end{itemize}
As the result, we obtain a codeword sequence $f(\pmb{x}) \coloneqq f_0(\mathrm{b})f_1(\mathrm{a})f_1(\mathrm{d})f_2(\mathrm{b}) = 1000001111110$.

The decoding process of $f(\pmb{x}) = 1000001111110$ is as follows.
\begin{itemize}
\item After reading the prefix $10$ of $f(\pmb{x})$, the decoder can uniquely identify $x_1 = \mathrm{b}$ and $10 = f_0(\mathrm{b})$. The decoder can also know that $x_2$ is decoded with $f_{\trans_0(\mathrm{b})} = f_1$.
\item After reading the prefix $1000 = f_0(\mathrm{b})f_0(\mathrm{a})$ of $f(\pmb{x})$, the decoder still cannot uniquely identify  $x_2 = \mathrm{a}$ because there remain three possible cases: the case $x_2 = \mathrm{a}$, the case $x_2 = \mathrm{c}$, and the case $x_2 = \mathrm{d}$.
\item After reading the prefix $10000$ of $f(\pmb{x})$, the decoder can uniquely identify $x_2 = \mathrm{a}$ and $10000 = f_0(\mathrm{b})f_1(\mathrm{a})0$. The decoder can also know that $x_3$ is decoded with $f_{\trans_1(\mathrm{a})} = f_1$.
\item After reading the prefix $100000111 = f_0(\mathrm{b})f_1(\mathrm{a})\allowbreak f_1(\mathrm{d})$ of $f(\pmb{x})$, the decoder still cannot uniquely identify $x_3 = \mathrm{d}$ because there remain two possible cases: the case $x_3 = \mathrm{c}$ and the case $x_3 = \mathrm{d}$.
\item After reading the prefix $10000011111$ of $f(\pmb{x})$, the decoder can uniquely identify $x_3 = \mathrm{d}$ and $10000011111 = f_0(\mathrm{b})f_1(\mathrm{a})f_1(\mathrm{d})11$.
The decoder can also know that $x_4$ is decoded with $f_{\trans_1(\mathrm{d})} = f_2$.
\item After reading the entire sequence $f(\pmb{x}) = 1000001111110$, the decoder can uniquely identify $x_4 = \mathrm{b}$ and $1000001111110 = f_0(\mathrm{b})f_1(\mathrm{a})f_1(\mathrm{d})f_2(\mathrm{b})$.
\end{itemize}
Then the decoder recovers the original sequence $\pmb{x} = \mathrm{badb}$ correctly.
\end{example}

In encoding process of $\pmb{x} = x_1x_2 \ldots x_{n} \in \mathcal{S}^{\ast}$ with $F(f, \trans) \in \mathscr{F}^{(m)}$, 
the $m$ mappings $\trans_0, \trans_1, \ldots, \trans_{m-1}$ determine which code table to use to encode $x_2, x_3, \ldots, x_n$.
However, there are choices of which code table to use for the first symbol $x_1$.
For $i \in [F]$ and $\pmb{x} \in \mathcal{S}^{\ast}$, we define $f^{\ast}_i(\pmb{x}) \in \mathcal{C}^{\ast}$ as the codeword sequence in the case where $x_1$ is encoded with $f_i$.
Also, we define $\trans^{\ast}_i(\pmb{x}) \in [F]$ as the index of the code table used next after encoding $\pmb{x}$ in the case where $x_1$ is encoded with $f_i$.
We give formal definitions of $f^{\ast}_i$ and $\trans^{\ast}_i$ in the following Definition \ref{def:f_T} as recursive formulas.

\begin{definition}
 \label{def:f_T}
For $F(f, \trans) \in \mathscr{F}$ and $i \in [F]$, we define a mapping $f_i^{\ast} : \mathcal{S}^{\ast} \rightarrow \mathcal{C}^{\ast}$ and a mapping $\trans_i^{\ast} : \mathcal{S}^{\ast} \rightarrow [F]$ as
\begin{equation}
\label{eq:fstar}
f_i^{\ast}(\pmb{x}) = 
\begin{cases}
\lambda &\,\,\text{if}\,\, \pmb{x} = \lambda,\\
f_i(x_1)f_{\trans_i(x_1)}^{\ast}(\suff(\pmb{x})) &\,\,\text{if}\,\, \pmb{x} \neq \lambda,\\
\end{cases}
\end{equation}
\begin{equation}
\label{eq:tstar}
\trans_i^{\ast}(\pmb{x}) = 
\begin{cases}
i &\,\,\text{if}\,\, \pmb{x} = \lambda,\\
\trans^{\ast}_{\trans_i(x_1)}(\suff(\pmb{x})) &\,\,\text{if}\,\, \pmb{x} \neq \lambda\\
\end{cases}
\end{equation}
for $\pmb{x} = x_1 x_2 \ldots x_{n} \in \mathcal{S}^{\ast}$.
\end{definition}

\begin{example}
We consider $F(f, \trans) \coloneqq F^{(\gamma)}$ in Table \ref{tab:code-tuple}.
Then $f_0^{\ast}(\mathrm{badb})$ and $\trans^{\ast}_0(\mathrm{badb})$ is given as follows (cf. Example \ref{ex:encode}):
\begin{align*}
f_0^{\ast}(\mathrm{badb})
&= f_0(\mathrm{b}) f_1^{\ast}(\mathrm{adb})\\
&= f_0(\mathrm{b}) f_1(\mathrm{a}) f_1^{\ast}(\mathrm{db})\\
&= f_0(\mathrm{b}) f_1(\mathrm{a}) f_1(\mathrm{d}) f_2^{\ast}(\mathrm{b})\\
&= f_0(\mathrm{b}) f_1(\mathrm{a}) f_1(\mathrm{d}) f_2(\mathrm{b}) f_0^{\ast}(\lambda)\\
&= 1000001111110,
\end{align*}
\begin{align}
\trans^{\ast}_0(\mathrm{badb})
= \trans^{\ast}_{1}(\mathrm{adb})
= \trans^{\ast}_{1}(\mathrm{db})
= \trans^{\ast}_{2}(\mathrm{b})
= \trans^{\ast}_{0}(\mathrm{\lambda})
= 0.
\end{align}
\end{example}

The following Lemma \ref{lem:f_T} follows from Definition \ref{def:f_T}.
\begin{lemma}
\label{lem:f_T}
For any $F(f, \trans) \in \mathscr{F}$, $i \in [F]$, and $\pmb{x}, \pmb{y} \in \mathcal{S}^{\ast}$, 
the following statements (i)--(iii) hold.
\begin{itemize}
\item[(i)] $f_i^{\ast}(\pmb{x} \pmb{y}) = f_i^{\ast}(\pmb{x}) f^{\ast}_{\trans_i^{\ast}(\pmb{x})}(\pmb{y})$. 
\item[(ii)] $\trans_i^{\ast}(\pmb{x} \pmb{y}) = \trans^{\ast}_{\trans^{\ast}_i(\pmb{x})}(\pmb{y})$.
\item[(iii)] If $\pmb{x} \preceq \pmb{y}$, then $f^{\ast}_i(\pmb{x}) \preceq f^{\ast}_i(\pmb{y})$.
\end{itemize}
\end{lemma}

\subsection{$k$-bit Delay Decodable Code-tuples}

In Example \ref{ex:encode}, despite $f^{\ast}_0(\mathrm{ba}) = 1000$, to uniquely identify $x_1x_2 = \mathrm{ba}$, it is required to read $f^{\ast}_0(\mathrm{ba})0 = 10000$ including the additional $1$ bit. 
Namely, a decoding delay of $1$ bit occurs at the time to decode $x_2 = \mathrm{a}$.
Similarly, despite $f^{\ast}_0(\mathrm{bad}) = 100000111$, to uniquely identify $x_1x_2x_3 = \mathrm{bad}$, it is required to read $f^{\ast}_0(\mathrm{bad})11 = 10000011111$ including the additional $2$ bits.
Namely, a decoding delay of $2$ bits occurs at the time to decode $x_3 = \mathrm{d}$.
In general, in the decoding process with $F^{(\gamma)}$,
it is required to read the additional at most $2$ bits for the decoder to uniquely identify each symbol of a given source sequence.
We say that a code-tuple is \emph{$k$-bit delay decodable} if the decoder can always uniquely identify each source symbol by reading the additional $k$ bits of the codeword sequence.
The code-tuple $F^{(\gamma)}$ in Table \ref{tab:code-tuple} is an example of a $2$-bit delay decodable code-tuple.
To state the formal definition of a $k$-bit delay decodable code-tuple,
we introduce the following Definition \ref{def:pref}.

\begin{definition}
\label{def:pref}
For an integer $k \geq 0$, $F(f, \trans) \in \mathscr{F}, i \in [F]$, and $\pmb{b} \in \mathcal{C}^{\ast}$, we define 
\begin{equation}
\label{eq:pref1}
\PREF^k_{F, i}(\pmb{b}) \coloneqq \{\pmb{c} \in \mathcal{C}^k : \pmb{x} = x_1x_2\ldots x_{n} \in \mathcal{S}^{+}, f^{\ast}_i(\pmb{x}) \succeq \pmb{b}\pmb{c}, f_i(x_1) \succeq \pmb{b}  \},
\end{equation}
\begin{equation}
\label{eq:pref2}
\bar{\PREF}^k_{F, i}(\pmb{b}) \coloneqq \{\pmb{c} \in \mathcal{C}^k : \pmb{x} = x_1x_2\ldots x_{n} \in \mathcal{S}^{+}, f^{\ast}_i(\pmb{x}) \succeq \pmb{b}\pmb{c}, f_i(x_1) \succ \pmb{b}  \}.
\end{equation}
Namely, $\PREF^k_{F, i}(\pmb{b})$ (resp. $\bar{\PREF}^k_{F, i}(\pmb{b})$) is the set of all $\pmb{c} \in \mathcal{C}^k$ such that there exists $\pmb{x} = x_1x_2\ldots x_n \in \mathcal{S}^{+}$ satisfying $f^{\ast}_i(\pmb{x}) \succeq \pmb{b}\pmb{c}$ and $f_i(x_1) \succeq \pmb{b}$ (resp. $f_i(x_1) \succ \pmb{b}$).
\end{definition}

We write $\PREF^k_{F, i}(\lambda)$ (resp. $\bar{\PREF}^k_{F, i}(\lambda)$) as $\PREF^k_{F, i}$ (resp. $\bar{\PREF}^k_{F, i}$) for simplicity.
We have
\begin{eqnarray}
\PREF^k_{F, i} \overset{(\mathrm{A})}{=} \{\pmb{c} \in \mathcal{C}^k : \pmb{x} \in \mathcal{S}^{+}, f^{\ast}_i(\pmb{x}) \succeq \pmb{c}\}
\overset{(\mathrm{B})}{=} \{\pmb{c} \in \mathcal{C}^k : \pmb{x} \in \mathcal{S}^{\ast}, f^{\ast}_i(\pmb{x}) \succeq \pmb{c}\},\label{eq:pref3}
\end{eqnarray}
where (A) follows from (\ref{eq:pref1}), and (B) is justified as follows.
The relation ``$\subseteq$'' holds by $\mathcal{S}^{+} \subseteq \mathcal{S}^{\ast}$.
We show the relation ``$\supseteq$''.
We choose $\pmb{c} \in \mathcal{C}^k$ such that $f^{\ast}_i(\pmb{x}) \succeq \pmb{c}$ for some $\pmb{x} \in \mathcal{S}^{\ast}$ arbitrarily and show that $f^{\ast}_i(\pmb{x}') \succeq \pmb{c}$ for some $\pmb{x}' \in \mathcal{S}^{+}$.
The case $\pmb{x} \in \mathcal{S}^{+}$ is trivial.
In the case $\pmb{x} \in \{\lambda\} = \mathcal{S}^{\ast} \setminus \mathcal{S}^{+}$,
we have $\pmb{c} = \lambda$ since $\pmb{c} \preceq f^{\ast}_i(\pmb{x}) = f^{\ast}_i(\lambda) = \lambda$ by (\ref{eq:fstar}).
This leads to that any $\pmb{x}' \in \mathcal{S}^{+}$ satisfies $f^{\ast}_i(\pmb{x}') \succeq \lambda = \pmb{c}$ as desired.
Hence, the relation ``$\supseteq$'' holds.

\begin{example}
We consider $F(f, \trans) \coloneqq F^{(\beta)}$ in Table \ref{tab:code-tuple}.
First, we confirm $\PREF^3_{F, 0}(\pmb{b}) = \{100, 101, 111\}$ for $\pmb{b} = 101$ as follows.
\begin{itemize}
\item $100 \in \PREF^3_{F, 0}(\pmb{b})$ holds because $\pmb{x} = \mathrm{cc}$ satisfies $f^{\ast}_0(\pmb{x}) = 1011000 \succeq \pmb{b}100$ and $f_0(x_1) = 101 \succeq \pmb{b}$.
\item $101 \in \PREF^3_{F, 0}(\pmb{b})$ holds because $\pmb{x} = \mathrm{da}$ satisfies $f^{\ast}_0(\pmb{x})  = 10110110 \succeq \pmb{b}101$ and $f_0(x_1) = 1011 \succeq \pmb{b}$.
\item $111 \in \PREF^3_{F, 0}(\pmb{b})$ holds because $\pmb{x} = \mathrm{cbb}$ satisfies $f^{\ast}_0(\pmb{x}) = 1011111 \succeq \pmb{b}111$ and $f_0(x_1) = 101 \succeq \pmb{b}$.
\end{itemize}
Next, we confirm $\bar{\PREF}^3_{F, 0}(\pmb{b}) = \{101\}$ for $\pmb{b} = 101$ as follows.
\begin{itemize}
\item $101 \in \bar{\PREF}^3_{F, 0}(\pmb{b})$ holds because $\pmb{x} = \mathrm{da}$ satisfies $f^{\ast}_0(\pmb{x})  = 10110110 \succeq \pmb{b}101$ and $f_0(x_1) = 1011 \succ \pmb{b}$.
\end{itemize}
Also, we confirm $\bar{\PREF}^0_{F, 1}(\pmb{b}) = \{\lambda\}$ for $\pmb{b} = 011$ as follows.
\begin{itemize}
\item $\lambda \in \bar{\PREF}^0_{F, 1}(\pmb{b})$ holds because $\pmb{x} = \mathrm{a}$ satisfies $f^{\ast}_1(\pmb{x}) = 0110 \succeq \pmb{b} = \pmb{b}\lambda$ and $f_1(x_1) = 0110 \succ \pmb{b}$.
\end{itemize}
\end{example}

\begin{example}
Table \ref{tab:pref} shows $\PREF^1_{F, i}$ and $\PREF^2_{F, i}$ for the code-tuples $F$ in Table \ref{tab:code-tuple}.
The rightmost column of Table \ref{tab:pref} is used later in Example \ref{ex:classes}.
Also, Table \ref{tab:k-bitdelay} shows $\bar{\PREF}^2_{F, i}(f_i(s))$ for $F(f, \trans) \coloneqq F^{(\gamma)}$ in Table \ref{tab:code-tuple}.
\end{example}

\begin{table}
\caption{The set $\PREF^1_{F, i}$ and $\PREF^2_{F, i}$ for the code-tuples $F$ in Table \ref{tab:code-tuple}}
\label{tab:pref}
\centering
\begin{tabular}{c | ccc | ccc | c}
\hline
$F \in \mathscr{F}$ & $\PREF^1_{F, 0}$ & $\PREF^1_{F, 1}$ & $\PREF^1_{F, 2}$ & $\PREF^2_{F, 0}$ & $\PREF^2_{F, 1}$ & $\PREF^2_{F, 2}$ & \\
\hline
$F^{(\alpha)}$ & $\{0, 1\}$ & $\{0, 1\}$ & $\emptyset$ & $\{01, 10, 11\}$ & $\{01, 10\}$ & $\emptyset$ & $F \in \mathscr{F}_{2\hdec} \setminus \mathscr{F}_0$\\
$F^{(\beta)}$ & $\{0, 1\}$ & $\{0\}$ & $\{1\}$ & $\{01, 10, 11\}$ & $\{01\}$ & \{10, 11\} & $F \in \mathscr{F}_{\reg} \setminus \mathscr{F}_0$\\
$F^{(\gamma)}$ & $\{0, 1\}$ & $\{0, 1\}$ & $\{1\}$ & $\{01, 10\}$ & $\{00, 01, 10\}$ & \{11\} & $F \in \mathscr{F}_0 \setminus \mathscr{F}_1$\\
$F^{(\delta)}$  & $\{0, 1\}$ & $\{0, 1\}$ & $\{1\}$ & $\{01, 10\}$ & $\{00, 01, 10\}$ & $\{10, 11\}$ & $F \in \mathscr{F}_0 \setminus \mathscr{F}_1$\\
$F^{(\epsilon)}$ & $\{0, 1\}$ & $\{0, 1\}$ & $\{0, 1\}$ & $\{01, 10\}$ & $\{00, 01, 10\}$ & $\{00, 01, 10\}$ & $F \in \mathscr{F}_1 \setminus \mathscr{F}_2$\\
$F^{(\zeta)}$ & \{1\}& \{0, 1\}& \{0, 1\}& $\{10, 11\}$ & $\{01, 10, 11\}$ & $\{00, 01, 10\}$ & $F \in \mathscr{F}_0 \setminus \mathscr{F}_1$\\
$F^{(\eta)}$ & \{0, 1\}& \{0, 1\}& \{0, 1\}& $\{00, 01, 10, 11\}$ & $\{01, 10, 11\}$ & $\{00, 01, 10\}$ & $F \in \mathscr{F}_2 \setminus \mathscr{F}_3$\\
$F^{(\theta)}$ & \{0, 1\}& \{0, 1\}& \{0, 1\} & $\{00, 01, 10, 11\}$ & $\{01, 10, 11\}$ & $\{01, 10, 11\}$ & $F \in \mathscr{F}_3 \setminus \mathscr{F}_4$\\
$F^{(\iota)}$ & \{0, 1\}& \{0, 1\}& & $\{00, 01, 10, 11\}$ & $\{01, 10, 11\}$ & & $F \in \mathscr{F}_4 \setminus \mathscr{F}_{\AIFV}$\\
$F^{(\kappa)}$ & \{0, 1\}& \{0, 1\}& & $\{00, 01, 10, 11\}$ & $\{01, 10, 11\}$ & & $F \in \mathscr{F}_{\AIFV}$\\
\hline
\end{tabular}
\vspace{8pt}
\end{table}

We consider the situation where the decoder has already read the prefix $\pmb{b}'$ of a given codeword sequence and identified $x_1x_2\ldots x_l$ of the original sequence $\pmb{x}$.
Then we have $\pmb{b}' = f_{i_1}(x_1)f_{i_2}(x_2)\ldots f_{i_l}(x_l)\pmb{b}$ for some $\pmb{b} \in \mathcal{C}^{\ast}$.
We now consider identifying the next symbol $x_{l+1}$.
Let $i \coloneqq i_{l+1}$ and $\assign_{F, i}(\pmb{b}) = \{s_1, s_2, \ldots, s_r\}$.
Then there are the following $r+1$ possible cases for $x_{l+1}$: the case $x_{l+1} = s_1$, the case $x_{l+1} = s_2$, $\ldots,$ the case $x_{l+1} = s_r$, and the case $f_i(x_{l+1}) \succ \pmb{b}$.
For a code-tuple $F$ to be $k$-bit delay decodable, the decoder must be able to distinguish the $r+1$ cases by reading the following $k$ bits of the codeword sequence.
Namely, it is required that the following $r+1$ sets are disjoint:
\begin{itemize}
\item $\PREF^k_{F, \trans_i(s_1)}$, the set of all possible following $k$ bits in the case $x_{l+1} = s_1$,
\item $\PREF^k_{F, \trans_i(s_2)}$, the set of all possible following $k$ bits in the case $x_{l+1} = s_2$,
\item $\cdots$,
\item $\PREF^k_{F, \trans_i(s_r)}$, the set of all possible following $k$ bits in the case $x_{l+1} = s_r$,
\item $\bar{\PREF}^k_{F, i}(\pmb{b})$, the set of all possible following $k$ bits in the case $f_i(x_{l+1}) \succ \pmb{b}$.
\end{itemize}

\begin{example}
We obtain $f^{\ast}_0(\pmb{x}) = 1000001111110$ by encoding $\pmb{x} \coloneqq \mathrm{badb}$ with $F(f, \trans) \coloneqq F^{(\gamma)}$ in Table \ref{tab:code-tuple} (cf. Example \ref{ex:encode}).
We consider the decoding process of $f^{\ast}_0(\pmb{x})$.

\begin{itemize}
\item First, we suppose that the decoder already read the prefix $\pmb{b}' = 1000$ of $f^{\ast}_0(\pmb{x})$ and identified $x_1 = \mathrm{b}$.
Then we have $\pmb{b}' = f_0(x_1)00$ and $\assign_{F, 1}(00) = \{\mathrm{a}\}$, and the next symbol $x_2$ is decoded with $f_{\trans_0(\mathrm{b})} = f_1$.
Now, there are two possible cases for $x_2$: the case $x_2 = \mathrm{a}$ and the case $f_1(x_2) \succ 00$ (i.e., $x_2 = \mathrm{c}$ or $x_2 = \mathrm{d}$).
The decoder can distinguish these two cases by reading the following $2$ bits because 
\begin{itemize}
\item $\PREF^2_{F, \trans_1(\mathrm{a})}$, the set of all possible following $2$ bits in the case $x_{2} = \mathrm{a}$, and
\item $\bar{\PREF}^2_{F, 1}(00)$, the set of all possible following $2$ bits in the case $f_1(x_2) \succ \pmb{b}$,
\end{itemize}
are disjoint: $\PREF^2_{F, \trans_1(\mathrm{a})} \cap \bar{\PREF}^2_{F, 1}(f_1(\mathrm{a})) = \{00, 01, 10\} \cap \{11\} = \emptyset$.
Since the following $2$ bits are $00 \in \PREF^2_{F, \trans_1(\mathrm{a})}$, the decoder can identify $x_2 = \mathrm{a}$ indeed.

\item Next, we suppose that the decoder already read the prefix $\pmb{b}' = 100000$ of $f^{\ast}_0(\pmb{x})$ and identified $x_1x_2 = \mathrm{ba}$.
Then we have $\pmb{b}' = f^{\ast}_0(x_1x_2)00$ and $\assign_{F, 1}(00) = \{\mathrm{a}\}$, and the next symbol $x_3$ is decoded with $f_{\trans_1(\mathrm{a})} = f_1$.
Now, there are two possible cases for $x_3$: the case $x_3 = \mathrm{a}$ and the case $f_1(x_3) \succ 00$ (i.e., $x_3 = \mathrm{c}$ or $x_3 = \mathrm{d}$).
The decoder can distinguish these two cases by reading the following $2$ bits because 
\begin{itemize}
\item $\PREF^2_{F, \trans_1(\mathrm{a})}$, the set of all possible following $2$ bits in the case $x_3 = \mathrm{a}$, and
\item $\bar{\PREF}^2_{F, 1}(00)$, the set of all possible following $2$ bits in the case $f_1(x_3) \succ \pmb{b}$,
\end{itemize}
are disjoint: $\PREF^2_{F, \trans_1(\mathrm{a})} \cap \bar{\PREF}^2_{F, 1}(f_1(\mathrm{a})) = \{00, 01, 10\} \cap \{11\} = \emptyset$.
Since the following $2$ bits are $11 \in \bar{\PREF}^2_{F, 1}(00)$, the decoder can identify $f_1(x_3) \succ 00$, in particular, $x_3 \neq \mathrm{a}$ indeed.

\item Lastly, we suppose that the decoder already read the prefix $\pmb{b}' = 100000111$ of $f^{\ast}_0(\pmb{x})$ and identified $x_1x_2 = \mathrm{ba}$.
Then we have $\pmb{b}' = f^{\ast}_0(\mathrm{ba})00111$ and $\assign_{F, 1}(00111) = \{\mathrm{c}, \mathrm{d}\}$.
Now, there are two possible cases for $x_3$: the case $x_3 = \mathrm{c}$ and the case $x_3 = \mathrm{d}$.
The decoder can distinguish these two cases by reading the following $2$ bits because 
\begin{itemize}
\item $\PREF^2_{F, \trans_1(\mathrm{c})}$, the set of all possible following $2$ bits in the case $x_{2} = \mathrm{c}$, and
\item $\PREF^2_{F, \trans_1(\mathrm{d})}$, the set of all possible following $2$ bits in the case $x_{2} = \mathrm{d}$,
\end{itemize}
are disjoint: $\PREF^2_{F, \trans_1(\mathrm{c})} \cap \PREF^2_{F, \trans_1(\mathrm{d})} = \{00, 01, 10\} \cap \{11\} = \emptyset$.
Since the following $2$ bits are $11 \in \PREF^2_{F, \trans_1(\mathrm{d})}$, the decoder can identify $x_3 = \mathrm{d}$ indeed.
\end{itemize}
\end{example}

The discussion above leads to the following Definition \ref{def:k-bitdelay}.

 \begin{definition}
  \label{def:k-bitdelay}
 Let $k \geq 0$ be an integer. 
A code-tuple $F(f, \trans)$ is said to be \emph{$k$-bit delay decodable} if the following conditions (i) and (ii) hold.
\begin{enumerate}[(i)]
\item For any $i \in [F]$ and $s \in \mathcal{S}$, it holds that $\PREF^k_{F, \trans_i(s)} \cap \bar{\PREF}^k_{F, i}(f_i(s)) = \emptyset$.
\item For any $i \in [F]$ and $s, s' \in \mathcal{S}$, if $s \neq s'$ and $f_i(s) = f_i(s')$, then $\PREF^k_{F, \trans_i(s)} \cap \PREF^k_{F, \trans_i(s')} =  \emptyset$.
\end{enumerate}
 For an integer $k \geq 0$, we define $\mathscr{F}_{k\hdec}$ as the set of all $k$-bit delay decodable code-tuples, that is, 
 \begin{equation}
 \mathscr{F}_{k\hdec} \coloneqq \{F \in \mathscr{F} : F \text{ is } k \text{-bit delay decodable} \}.
 \end{equation}
\end{definition}

\begin{example}
\label{ex:k-bitdelay}
We confirm $F(f, \trans) \coloneqq F^{(\gamma)}$ in Table \ref{tab:code-tuple} is $2$-bit delay decodable as follows.

First, we see that $F$ satisfies Definition \ref{def:k-bitdelay} (i) as follows (cf. Tables \ref{tab:pref} and \ref{tab:k-bitdelay}).
\begin{itemize}
\item $\PREF^2_{F, \trans_0(\mathrm{a})} \cap \bar{\PREF}^2_{F, 0}(f_0(\mathrm{a})) = \PREF^2_{F, 0} \cap \bar{\PREF}^2_{F, 0}(f_0(\mathrm{a})) = \{01, 10\} \cap \{00\} = \emptyset$.
\item $\PREF^2_{F, \trans_0(\mathrm{b})} \cap \bar{\PREF}^2_{F, 0}(f_0(\mathrm{b})) = \PREF^2_{F, 1} \cap \bar{\PREF}^2_{F, 0}(f_0(\mathrm{b})) = \{00, 01, 10\} \cap \emptyset = \emptyset$.
\item $\PREF^2_{F, \trans_0(\mathrm{c})} \cap \bar{\PREF}^2_{F, 0}(f_0(\mathrm{c})) = \PREF^2_{F, 0} \cap \bar{\PREF}^2_{F, 0}(f_0(\mathrm{c})) = \{01, 10\} \cap \emptyset = \emptyset$.
\item $\PREF^2_{F, \trans_0(\mathrm{d})} \cap \bar{\PREF}^2_{F, 0}(f_0(\mathrm{d})) = \PREF^2_{F, 2} \cap \bar{\PREF}^2_{F, 0}(f_0(\mathrm{d})) = \{11\} \cap \{00\} = \emptyset$.

\item $\PREF^2_{F, \trans_1(\mathrm{a})} \cap \bar{\PREF}^2_{F, 1}(f_1(\mathrm{a})) = \PREF^2_{F, 1} \cap \bar{\PREF}^2_{F, 1}(f_1(\mathrm{a})) = \{00, 01, 10\} \cap \{11\} = \emptyset$.
\item $\PREF^2_{F, \trans_1(\mathrm{b})} \cap \bar{\PREF}^2_{F, 1}(f_1(\mathrm{b})) = \PREF^2_{F, 0} \cap \bar{\PREF}^2_{F, 1}(f_1(\mathrm{b})) = \{01, 10\} \cap \{00\} = \emptyset$.
\item $\PREF^2_{F, \trans_1(\mathrm{c})} \cap \bar{\PREF}^2_{F, 1}(f_1(\mathrm{c})) = \PREF^2_{F, 1} \cap \bar{\PREF}^2_{F, 1}(f_1(\mathrm{c})) = \{00, 01, 10\} \cap \emptyset = \emptyset$.
\item $\PREF^2_{F, \trans_1(\mathrm{d})} \cap \bar{\PREF}^2_{F, 1}(f_1(\mathrm{d})) = \PREF^2_{F, 2} \cap \bar{\PREF}^2_{F, 1}(f_1(\mathrm{d})) = \{11\} \cap \emptyset = \emptyset$.

\item $\PREF^2_{F, \trans_2(\mathrm{a})} \cap \bar{\PREF}^2_{F, 2}(f_2(\mathrm{a})) = \PREF^2_{F, 1} \cap \bar{\PREF}^2_{F, 2}(f_2(\mathrm{a})) = \{00, 01, 10\} \cap \emptyset = \emptyset$.
\item $\PREF^2_{F, \trans_2(\mathrm{b})} \cap \bar{\PREF}^2_{F, 2}(f_2(\mathrm{b})) = \PREF^2_{F, 0} \cap \bar{\PREF}^2_{F, 2}(f_2(\mathrm{b})) = \{01, 10\} \cap \{00\} = \emptyset$.
\item $\PREF^2_{F, \trans_2(\mathrm{c})} \cap \bar{\PREF}^2_{F, 2}(f_2(\mathrm{c})) = \PREF^2_{F, 2} \cap \bar{\PREF}^2_{F, 2}(f_2(\mathrm{c})) = \{11\} \cap \emptyset = \emptyset$.
\item $\PREF^2_{F, \trans_2(\mathrm{d})} \cap \bar{\PREF}^2_{F, 2}(f_2(\mathrm{d})) = \PREF^2_{F, 2} \cap \bar{\PREF}^2_{F, 2}(f_2(\mathrm{d})) = \{11\} \cap \{00, 01\} = \emptyset$.
\end{itemize}

Next, we see that $F$ satisfies Definition \ref{def:k-bitdelay} (ii) as follows (cf. Table \ref{tab:pref}).
\begin{itemize}
\item $\PREF^2_{F, \trans_0(\mathrm{a})} \cap \PREF^2_{F, \trans_0(\mathrm{d})} = \PREF^2_{F, 0} \cap \PREF^2_{F, 2} = \{01, 10\} \cap \{11\} = \emptyset$.
\item $\PREF^2_{F, \trans_1(\mathrm{c})} \cap \PREF^2_{F, \trans_1(\mathrm{d})} = \PREF^2_{F, 1} \cap \PREF^2_{F, 2} = \{00, 01, 10\} \cap \{11\} = \emptyset$.
\end{itemize}

Consequently, we have $F \in \mathscr{F}_{2\hdec}$.
\end{example}

\begin{example}
In a similar way to Example \ref{ex:k-bitdelay}, we can see that the code-tuples in Table \ref{tab:code-tuple} are $2$-bit delay decodable except for $F^{(\beta)}$. We state some more examples as follows.
\begin{itemize}
\item For $F(f, \trans) \coloneqq F^{(\alpha)}$, we have $F \not\in \mathscr{F}_{1\hdec}$ because $\PREF^1_{F, \trans_0(\mathrm{b})} \cap \bar{\PREF}^1_{F, 0}(f_0(\mathrm{b})) = \{0, 1\} \cap \{1\} = \{1\} \neq \emptyset$.
\item For $F(f, \trans) \coloneqq F^{(\beta)}$, for any integer $k \geq 0$, we have $F \not\in \mathscr{F}_{k \hdec}$ because
$\PREF^k_{F, \trans_1(\mathrm{a})} \cap \PREF^k_{F, \trans_1(\mathrm{b})} = \PREF^k_{F, 1} \cap \PREF^k_{F, 1} = \PREF^k_{F, 1} \neq \emptyset$.
\item For $F(f, \trans) \coloneqq F^{(\gamma)}$, we have $F \not\in \mathscr{F}_{1\hdec}$ because $\PREF^1_{F, \trans_1(\mathrm{c})} \cap \PREF^1_{F, \trans_1(\mathrm{d})} = \{0, 1\} \cap \{1\} = \{1\} \neq \emptyset$.
\end{itemize}
\end{example}

\begin{table}
\caption{The set $\bar{\PREF}^2_{F, i}(f_i(s))$ for $F \coloneqq F^{(\gamma)}$}
\label{tab:k-bitdelay}
\centering
\begin{tabular}{c | ccc}
\hline\\[-1em]
$s \in \mathcal{S}$ & $\bar{\PREF}^2_{F, 0}(f_0(s))$ & $\bar{\PREF}^2_{F, 1}(f_1(s))$ & $\bar{\PREF}^2_{F, 2}(f_2(s))$ \\
\hline
$\mathrm{a}$ & $\{00\}$ & $\{11\}$ & $\emptyset$\\
$\mathrm{b}$ & $\emptyset$ & $\{00\}$ & $\{00\}$ \\
$\mathrm{c}$ & $\emptyset$ & $\emptyset$ & $\emptyset$\\
$\mathrm{d}$ & $\{00\}$ & $\emptyset$ & $\{00, 01\}$\\
\hline
\end{tabular}
\vspace{8pt}
\end{table}

\begin{remark}
\label{rem:injective}
If all the code tables $f_0, f_1, \ldots, f_{|F|-1}$ are injective, then Definition \ref{def:k-bitdelay} (ii) holds since there are no $i \in [F]$ and $s, s' \in \mathcal{S}$ such that $s \neq s$ and $f_i(s) \neq f_i(s')$.

If $k = 0$, then the converse also holds as seen below.
We consider Definition \ref{def:k-bitdelay} (ii) for the case $k = 0$.
Then by (\ref{eq:pref3}), we have $\PREF^k_{F, \trans_i(s)} \cap \PREF^k_{F, \trans_i(s')} =  \{\lambda\} \cap \{\lambda\} = \{\lambda\} \neq \emptyset$ for any $i \in [F]$ and $s, s' \in \mathcal{S}$.
Hence, for $F$ to satisfy Definition \ref{def:k-bitdelay} (ii), it is required that for any $i \in [F]$ and $s, s' \in \mathcal{S}$, if $s \neq s'$, then $f_i(s) \neq f_i(s')$, that is, $f_0, f_1, \ldots, f_{|F|-1}$ are injective.
\end{remark}

\begin{remark}
A $k$-bit delay decodable code-tuple $F$ is not necessarily uniquely decodable, that is, the mappings $f^{\ast}_0, f^{\ast}_1, \ldots, f^{\ast}_{|F|-1}$ are not necessarily injective. 
Indeed, for $F(f, \trans) \coloneqq F^{(\gamma)} \in \mathscr{F}_{2\hdec}$ in Table \ref{tab:code-tuple}, we have ${f_0}^{\ast}(\mathrm{bc}) = 1000111 = {f_0}^{\ast}(\mathrm{bd})$. In general, it is possible that the decoder cannot uniquely recover the last few symbols of the original source sequence in the case where the rest of the codeword sequence is less than $k$ bits. In such a case, we should append additional information for practical use (cf. \cite[Remark 2]{Yamamoto2015}).
\end{remark}

We now state the basic properties of $\PREF^k_{F, i}(\pmb{b})$ and $\bar{\PREF}^k_{F, i}(\pmb{b})$ as the following Lemmas \ref{lem:leaf} and \ref{lem:pref-sum}.
 
\begin{lemma}
\label{lem:leaf}
For any $F(f, \trans) \in \mathscr{F}$ and $i \in [F]$, the following statements (i)--(iii) hold.
\begin{enumerate}[(i)]
\item For any $\pmb{b} \in \mathcal{C}^{\ast}$, we have
$\bar{\PREF}^0_{F, i}(\pmb{b}) \neq \emptyset \iff {}^{\exists}s \in \mathcal{S} ; f_i(s) \succ \pmb{b}$.
\item There exists $s \in \mathcal{S}$ such that $\bar{\PREF}^0_{F, i}(f_i(s)) = \emptyset$.
\item If $|\assign_{F, i}(\lambda)| \leq 1$, in particular $f_i$ is injective, then $\bar{\PREF}^0_{F, i} \neq \emptyset$.
\end{enumerate}
\end{lemma}

\begin{proof}[Proof of Lemma \ref{lem:leaf}]
(Proof of (i)):
We have
\begin{equation}
\lambda \in \bar{\PREF}^0_{F, i}(\pmb{b})
\overset{(\mathrm{A})}{\iff} {}^{\exists} \pmb{x} = x_1x_2\ldots x_n \in \mathcal{S}^{+}; (f^{\ast}_i(\pmb{x}) \succeq \pmb{b}, f_i(x_1) \succ \pmb{b})
\iff {}^{\exists} s \in \mathcal{S}; f_i(s) \succ \pmb{b}
\end{equation}
as desired, where (A) follows from (\ref{eq:pref2}).

(Proof of (ii)): 
Let $s \in \argmax \{|f_i(s')| : s' \in \mathcal{S}\}$.
Then there is no $s' \in \mathcal{S}$ such that $f_i(s) \prec f_i(s')$.
Hence, by (i) of this lemma, we obtain $\bar{\PREF}^0_{F, i}(f_i(s)) = \emptyset$.

(Proof of (iii)):
By $|\assign_{F, i}(\lambda)| \leq 1$ and the assumption that $\sigma \geq 2$,
there exists $s \in \mathcal{S}$ such that $f_i(s) \neq \lambda$.
This is equivalent to $\bar{\PREF}^0_{F, i} \neq \emptyset$ by (i) of this lemma.
\end{proof}

\begin{lemma}
\label{lem:pref-sum}
For any integer $k \geq 0$, $F(f, \trans) \in \mathscr{F}$, $i \in [F]$, and $\pmb{b} \in \mathcal{C}^{\ast}$, the following statements (i)--(iii) hold.
\begin{enumerate}[(i)]
\item 
\begin{equation}
\PREF^k_{F, i}(\pmb{b}) = \bar{\PREF}^k_{F, i}(\pmb{b}) \cup \Big( \bigcup_{s \in \assign_{F, i}(\pmb{b})} \PREF^k_{F, \trans_i(s)} \Big).
\end{equation}
\item If $F \in \mathscr{F}_{k\hdec}$, then
\begin{equation}
|\PREF^k_{F, i}(\pmb{b})| = |\bar{\PREF}^k_{F, i}(\pmb{b})| + \sum_{s \in \assign_{F, i}(\pmb{b})} |\PREF^k_{F, \trans_i(s)}|.
\end{equation}
\item If $k \geq 1$, then
\begin{equation}
\bar{\PREF}^k_{F, i}(\pmb{b}) = 0\PREF^{k-1}_{F, i}(\pmb{b}0) \cup 1\PREF^{k-1}_{F, i}(\pmb{b}1).
\end{equation}
\end{enumerate}
\end{lemma}

\begin{proof}[Proof of Lemma \ref{lem:pref-sum}]
(Proof of (i)): 
For any $\pmb{c} \in \mathcal{C}^k$, we have
\begin{eqnarray}
\pmb{c} \in \PREF^k_{F, i}(\pmb{b})
&\overset{(\mathrm{A})}{\iff}& {}^{\exists} \pmb{x} \in \mathcal{S}^{+}; (f^{\ast}_i(\pmb{x}) \succeq \pmb{b}\pmb{c}, f_i(x_1) \succeq \pmb{b})\\
&\iff& ({}^{\exists} \pmb{x} \in \mathcal{S}^{+}; (f^{\ast}_i(\pmb{x}) \succeq \pmb{b}\pmb{c}, f_i(x_1) \succ \pmb{b})) \,\,\text{or}\,\, ({}^{\exists} \pmb{x} \in \mathcal{S}^{+}; (f^{\ast}_i(\pmb{x}) \succeq \pmb{b}\pmb{c}, f_i(x_1) = \pmb{b}))\\
&\overset{(\mathrm{B})}{\iff}& \pmb{c} \in \bar{\PREF}^k_{F, i}(\pmb{b}) \,\,\text{or}\,\, {}^{\exists} \pmb{x} \in \mathcal{S}^{+}; (f^{\ast}_i(\pmb{x}) \succeq \pmb{b}\pmb{c}, f_i(x_1) = \pmb{b})\\
&\overset{(\mathrm{C})}{\iff}& \pmb{c} \in \bar{\PREF}^k_{F, i}(\pmb{b}) \,\,\text{or}\,\, {}^{\exists} \pmb{x} \in \mathcal{S}^{+}; (f^{\ast}_{\trans_i(x_1)}(\suff(\pmb{x})) \succeq \pmb{c}, f_i(x_1) = \pmb{b})\\
&\iff& \pmb{c} \in \bar{\PREF}^k_{F, i}(\pmb{b}) \,\,\text{or}\,\, {}^{\exists} s \in \mathcal{S}; {}^{\exists} \pmb{x} \in \mathcal{S}^{\ast}; (f^{\ast}_{\trans_i(s)}(\pmb{x}) \succeq \pmb{c}, f_i(s) = \pmb{b})\\
&\iff& \pmb{c} \in \bar{\PREF}^k_{F, i}(\pmb{b}) \,\,\text{or}\,\, {}^{\exists} s \in \assign_{F, i}(\pmb{b}); {}^{\exists} \pmb{x} \in \mathcal{S}^{\ast}; f^{\ast}_{\trans_i(s)}(\pmb{x}) \succeq \pmb{c}\\
&\overset{(\mathrm{D})}{\iff}& \pmb{c} \in \bar{\PREF}^k_{F, i}(\pmb{b}) \,\,\text{or}\,\, {}^{\exists} s \in \assign_{F, i}(\pmb{b}); \pmb{c} \in \PREF^k_{F, \trans_i(s)}\\
&\iff& \pmb{c} \in \bar{\PREF}^k_{F, i}(\pmb{b}) \,\,\text{or}\,\, \pmb{c} \in\bigcup_{s \in \assign_{F, i}(\pmb{b})}  \PREF^k_{F, \trans_i(s)}\\
&\iff& \pmb{c} \in \bar{\PREF}^k_{F, i}(\pmb{b}) \cup \Big( \bigcup_{s \in \assign_{F, i}(\pmb{b})} \PREF^k_{F, \trans_i(s)} \Big)
\end{eqnarray}
as desired, where $x_1$ denotes the first symbol of $\pmb{x}$, and 
(A) follows from (\ref{eq:pref1}),
(B) follows from (\ref{eq:pref2}),
(C) follows from (\ref{eq:fstar}),
and (D) follows from (\ref{eq:pref3}).

(Proof of (ii)): 
We have
\begin{equation}
|\PREF^k_{F, i}(\pmb{b})| \overset{(\mathrm{A})}{=} |\bar{\PREF}^k_{F, i}(\pmb{b}) \cup \Big( \bigcup_{s \in \assign_{F, i}(\pmb{b})} \PREF^k_{F, \trans_i(s)} \Big)|
\overset{(\mathrm{B})}{=} |\bar{\PREF}^k_{F, i}(\pmb{b})| + |\bigcup_{s \in \assign_{F, i}(\pmb{b})} \PREF^k_{F, \trans_i(s)}|
\overset{(\mathrm{C})}{=} |\bar{\PREF}^k_{F, i}(\pmb{b})| + \sum_{s \in \assign_{F, i}(\pmb{b})} |\PREF^k_{F, \trans_i(s)}|
\end{equation}
as desired, where
(A) follows from (i) of this lemma,
(B) follows from $F \in \mathscr{F}_{k\hdec}$ and Definition \ref{def:k-bitdelay} (i),
and (C) follows from $F \in \mathscr{F}_{k\hdec}$ and Definition \ref{def:k-bitdelay} (ii).

(Proof of (iii)): 
For any $\pmb{c} = c_1c_2\ldots c_k \in \mathcal{C}^k$, we have
\begin{eqnarray}
\pmb{c} \in \bar{\PREF}^k_{F, i}(\pmb{b})
&\overset{(\mathrm{A})}{\iff}& {}^{\exists} \pmb{x} \in \mathcal{S}^{+}; (f^{\ast}_i(\pmb{x}) \succeq \pmb{b}\pmb{c}, f_i(x_1) \succ \pmb{b})\\
&\iff& {}^{\exists} \pmb{x} \in \mathcal{S}^{+}; (f^{\ast}_i(\pmb{x}) \succeq \pmb{b}c_1\suff(\pmb{c}), f_i(x_1) \succeq \pmb{b}c_1)\\
&\iff& (c_1 = 0, {}^{\exists} \pmb{x} \in \mathcal{S}^{+}; (f^{\ast}_i(\pmb{x}) \succeq \pmb{b}0\suff(\pmb{c}), f_i(x_1) \succeq \pmb{b}0)) \,\,\text{or} \nonumber\\
&&(c_1 = 1, {}^{\exists} \pmb{x} \in \mathcal{S}^{+}; (f^{\ast}_i(\pmb{x}) \succeq \pmb{b}1\suff(\pmb{c}), f_i(x_1) \succeq \pmb{b}1))\\
&\overset{(\mathrm{B})}{\iff}& (c_1 = 0, \suff(\pmb{c}) \in \PREF^{k-1}_{F, i}(\pmb{b}0)) \,\,\text{or}\,\, (c_1 = 1, \suff(\pmb{c}) \in \PREF^{k-1}_{F, i}(\pmb{b}1))\\
&\iff& \pmb{c} \in 0\PREF^{k-1}_{F, i}(\pmb{b}0) \,\,\text{or}\,\, \pmb{c} \in 1\PREF^{k-1}_{F, i}(\pmb{b}1)\\
&\iff& \pmb{c} \in 0\PREF^{k-1}_{F, i}(\pmb{b}0) \cup 1\PREF^{k-1}_{F, i}(\pmb{b}1)
\end{eqnarray}
as desired, where $x_1$ denotes the first symbol of $\pmb{x}$, and
(A) follows from (\ref{eq:pref2}),
and (B) follows from (\ref{eq:pref1}).
\end{proof}

For $F(f, \trans) \coloneqq F^{(\alpha)}$ in Table \ref{tab:code-tuple}, we can see that ${f_2^\ast}(\pmb{x}) = \lambda$ holds for any $\pmb{x} \in \mathcal{S}^{\ast}$.
To exclude such abnormal and useless code-tuples, we introduce a class $\mathscr{F}_{\ext}$ in the following Definition \ref{def:F_ext}.

\begin{definition}
\label{def:F_ext}
A code-tuple $F$ is said to be \emph{extendable} if $\PREF^1_{F, i} \neq \emptyset$ for any $i \in [F]$.
We define $\mathscr{F}_{\ext}$ as the set of all extendable code-tuples, that is,
\begin{equation}
\mathscr{F}_{\ext} \coloneqq \{F \in \mathscr{F} : {}^{\forall}i \in [F]; \PREF^1_{F, i} \neq \emptyset\}.
\end{equation}
\end{definition}

\begin{example}
The code-tuple $F^{(\alpha)}$ in Table \ref{tab:code-tuple} is not extendable because $\PREF^1_{F^{(\alpha)}, 2} = \emptyset$ by Table \ref{tab:pref}.
The other code-tuples in Table \ref{tab:code-tuple} are extendable.
\end{example}

For extendable code-tuples, the following Lemmas \ref{lem:F_ext}--\ref{cor:pref-sum} hold.
See \cite{IEICE2023} for the proofs of Lemmas \ref{lem:F_ext}, \ref{lem:pref-inc}, and \ref{lem:longest}.
Lemma \ref{lem:pref-inc2} is a direct consequence of Lemma \ref{lem:pref-inc}.

\begin{lemma}[{\cite[Lemma 3]{IEICE2023}}]
\label{lem:F_ext}
A code-tuple $F(f, \trans)$ is extendable if and only if for any $i \in [F]$ and integer $l \geq 0$, 
there exists $\pmb{x} \in \mathcal{S}^{\ast}$ such that $|f^{\ast}_i(\pmb{x})| \geq l$.
\end{lemma}

\begin{lemma}[{\cite[Lemma 4]{IEICE2023}}]
\label{lem:pref-inc}
Let $k, k'$ be integers such that $0 \leq k \leq k'$.
For any $F \in \mathscr{F}_{\ext}, i \in [F]$, $\pmb{b} \in \mathcal{C}^{\ast}$, and $\pmb
{c} \in \mathcal{C}^k$, the following statements (i) and (ii) hold.
\begin{enumerate}[(i)]
\item $\pmb{c} \in \PREF^k_{F, i}(\pmb{b}) \iff {}^{\exists}\pmb{c}' \in \mathcal{C}^{k'-k}; \pmb{c} \pmb{c}' \in \PREF^{k'}_{F, i}(\pmb{b}).$
\item $\pmb{c} \in \bar{\PREF}^k_{F, i}(\pmb{b}) \iff {}^{\exists}\pmb{c}' \in \mathcal{C}^{k'-k}; \pmb{c} \pmb{c}' \in \bar{\PREF}^{k'}_{F, i}(\pmb{b}).$
\end{enumerate}
\end{lemma}

\begin{lemma}
\label{lem:pref-inc2}
For any $F \in \mathscr{F}_{\ext}$, $i \in [F]$, and $\pmb{b} \in \mathcal{C}^{\ast}$,
the following statements (i) and (ii) hold.
\begin{enumerate}[(i)]
\item
\begin{enumerate}[(a)]
\item For any integer $k \geq 0$, we have $\PREF^k_{F, i}(\pmb{b}) = \emptyset \iff \PREF^0_{F, i}(\pmb{b}) = \emptyset$.
\item For any integers $k$ and $k'$ such that $0 \leq k \leq k'$, we have $|\PREF^k_{F, i}(\pmb{b})| \leq |\PREF^{k'}_{F, i}(\pmb{b})|$.
\end{enumerate}

\item 
\begin{enumerate}[(a)]
\item For any integer $k \geq 0$, we have $\bar{\PREF}^k_{F, i}(\pmb{b}) = \emptyset \iff \bar{\PREF}^0_{F, i}(\pmb{b}) = \emptyset$.
\item For any integers $k$ and $k'$ such that $0 \leq k \leq k'$, we have $|\bar{\PREF}^k_{F, i}(\pmb{b})| \leq |\bar{\PREF}^{k'}_{F, i}(\pmb{b})|$.
\end{enumerate}

\end{enumerate}
\end{lemma}

\begin{lemma}[{\cite[Lemma 5]{IEICE2023}}]
\label{lem:longest}
For any integer $k \geq 0$, $F(f, \trans) \in \mathscr{F}_{\ext} \cap \mathscr{F}_{k\hdec}, i \in [F]$, and $\pmb{x} \in \mathcal{S}^{\ast}$, if $f^{\ast}_i(\pmb{x}) = \lambda$, then $|\pmb{x}| < |F|$.
\end{lemma}

\begin{lemma}
\label{cor:pref-sum}
For any integer $k \leq 2$, $F(f, \trans) \in \mathscr{F}_{2\hdec} \cap \mathscr{F}_{\ext}$, $i \in [F]$, and $s \in \mathcal{S}$,
we have $|\bar{\PREF}^k_{F, i}(f_i(s))| + |\PREF^2_{F, \trans_i(s)}| \leq 4$.
\end{lemma}

\begin{proof}[Proof of Lemma \ref{cor:pref-sum}]
We have
\begin{eqnarray}
|\bar{\PREF}^k_{F, i}(f_i(s))| + |\PREF^2_{F, \trans_i(s)}| 
\overset{(\mathrm{A})}{\leq} |\bar{\PREF}^2_{F, i}(f_i(s))| + |\PREF^2_{F, \trans_i(s)}|
\overset{(\mathrm{B})}{\leq} |\PREF^2_{F, i}(f_i(s))|
\leq 4
\end{eqnarray}
as desired, where 
(A) follows from $k \leq 2$, $F \in \mathscr{F}_{\ext}$, and Lemma \ref{lem:pref-inc2} (ii) (b),
and (B) follows from $F \in \mathscr{F}_{2\hdec}$ and Lemma \ref{lem:pref-sum} (ii).
\end{proof}

\subsection{Average Codeword Length of Code-Tuple}
\label{subsec:evaluation}

In this subsection, we introduce the average codeword length $L(F)$ of a code-tuple $F$.
First, for $F(f, \trans) \in \mathscr{F}$ and $i, j \in [F]$, we define the transition probability $Q_{i, j}(F)$ as the probability of using the code table $f_j$ next after using the code table $f_i$ in the encoding process.
 \begin{definition}
\label{def:transprobability}
For $F(f, \trans) \in \mathscr{F}$ and $i, j \in [F]$,
we define the \emph{transition probability} $Q_{i,j}(F)$ as
\begin{equation}
\label{eq:9x9htdrx1001}
 Q_{i,j}(F) \coloneqq \sum_{s \in \mathcal{S}, \trans_i(s) = j} \mu(s).
 \end{equation}
We also define the \emph{transition probability matrix} $Q(F)$ as the following $|F|\times|F|$ matrix:
 \begin{equation}
  \left[
    \begin{array}{cccc}
      Q_{0,0}(F) & Q_{0,1}(F) & \cdots & Q_{0, |F|-1}(F) \\
       Q_{1,0}(F) &  Q_{1,1}(F) & \cdots &  Q_{1, |F|-1}(F) \\
      \vdots & \vdots & \ddots & \vdots \\
      Q_{|F|-1, 0}(F) &  Q_{|F|-1, 1}(F) & \cdots &  Q_{|F|-1, |F|-1}(F) 
    \end{array}
  \right].
 \end{equation}
  \end{definition}

We fix $F \in \mathscr{F}$ and consider the encoding process with $F$.
Let $I_i \in [F]$ be the index of the code table used to encode the $i$-th symbol of a source sequence for $i = 1, 2, 3, \ldots$.
Then $\{I_i\}_{i = 1, 2, 3, \ldots}$ is a Markov process with the transition probability matrix $Q(F)$.
As stated later in Definition \ref{def:evaluation}, the average codeword length $L(F)$ of $F$ is defined depending on the stationary distribution $\pmb{\pi}$ of the Markov process $\{I_i\}_{i = 1, 2, 3, \ldots}$ (i.e., a solution of the simultaneous equations (\ref{eq:stationary1}) and (\ref{eq:stationary2})). 
To define $L(F)$ uniquely, we limit the scope of consideration to the class $\mathscr{F}_{\reg}$ defined in the following Definition \ref{def:regular}.

\begin{definition}
\label{def:regular}
A code-tuple $F$ is said to be \emph{regular} if
the following simultaneous equations (\ref{eq:stationary1}) and (\ref{eq:stationary2}) have the unique solution $\pmb{\pi} = (\pi_0, \pi_1, \ldots, \pi_{|F|-1})$:

\begin{numcases}{}
\pmb{\pi}Q(F) = \pmb{\pi}\label{eq:stationary1},\\
 \sum_{i \in [F]} \pi_i = 1. \label{eq:stationary2}
 \end{numcases}
 
 We define $\mathscr{F}_{\reg}$ as the set of all regular code-tuples, that is,
 \begin{equation}
 \mathscr{F}_{\reg} \coloneqq \{F \in \mathscr{F} : F \text{ is regular}\}.
 \end{equation}
 For $F \in \mathscr{F}_{\reg}$, we define $\pmb{\pi}(F) = (\pi_0(F), \pi_1(F), \ldots, \pi_{|F|-1}(F))$
 as the unique solution of the simultaneous equations (\ref{eq:stationary1}) and (\ref{eq:stationary2}).
 \end{definition}
 
Since the transition probability matrix $Q(F)$ depends on $\mu$, it might seem that the class $\mathscr{F}_{\reg}$ also depends on $\mu$.
However, we show later as Lemma \ref{lem:kernel} that in fact $\mathscr{F}_{\reg}$ is independent from $\mu$.
More precisely, whether a code-tuple $F(f, \trans)$ belongs to $\mathscr{F}_{\reg}$ depends only on $\trans_0, \trans_1, \ldots, \trans_{|F|-1}$.
 
\begin{remark}
\label{rem:transitionprobability}
Note that $Q(F), L_i(F), L(F)$ and $\pmb{\pi}(F)$ depend on $\mu$.
However, since we are now discussing on a fixed $\mu$,
the average codeword length $L_i(F)$ of $f_i$ (resp. the transition probability matrix $Q(F)$) is determined only by the mapping $f_i$ (resp. $\trans_0, \trans_1, \ldots, \trans_{|F|-1}$) and therefore $\pmb{\pi}(F)$ of a regular code-tuple $F$ is also determined only by $\trans_0, \trans_1, \ldots, \trans_{|F|-1}$.
\end{remark} 
  
For any $F \in \mathscr{F}_{\reg}$, the asymptotical performance (i.e. average codeword length per symbol) does not depend on from which code table we start encoding:
the average codeword length $L(F)$ of a regular code-tuple $F \in \mathscr{F}_{\reg}$ is
the weighted sum of the average codeword lengths of the code tables $f_0, f_1, \ldots, f_{|F|-1}$ weighted by the stationary distribution $\pmb{\pi}(F)$.
Namely, $L(F)$ is defined as the following Definition \ref{def:evaluation}. 
 
 \begin{definition} 
 \label{def:evaluation}
 For $F(f, \trans) \in \mathscr{F}$ and $i \in [F]$, we define the \emph{average codeword length $L_i(F)$ of the single code table} $f_i : \mathcal{S} \rightarrow \mathcal{C}^{\ast}$ as
  \begin{equation}
 L_i(F) \coloneqq \sum_{s \in \mathcal{S}} |f_i(s)| \cdot \mu(s).
  \end{equation}
For $F \in \mathscr{F}_{\reg}$,
we define the \emph{average codeword length $L(F)$ of the code-tuple $F$} as 
 \begin{equation}
 \label{eq:evaluation}
 L(F) \coloneqq \sum_{i \in [F]} \pi_i(F)L_i(F).
 \end{equation}
 \end{definition}

 \begin{example}
 \label{ex:evaluation}
We consider $F \coloneqq F^{(\gamma)}$ of Table \ref{tab:code-tuple}, where $(\mu(\mathrm{a}), \mu(\mathrm{b}), \mu(\mathrm{c}), \mu(\mathrm{d})) = (0.1, 0.2, 0.3, 0.4)$.

We have 
 \begin{equation}
  Q(F) = \left[
    \begin{array}{cccc}
      0.4 & 0.2 & 0.4 \\
       0.2 &  0.4 & 0.4\\
      0.2 &  0.1 & 0.7 
    \end{array}
  \right].
\end{equation}
The simultaneous equations (\ref{eq:stationary1}) and (\ref{eq:stationary2}) has the unique solution $\pmb{\pi}(F) = (\pi_0(F), \pi_1(F), \pi_2(F)) = (1/4, 5/28, 4/7)$.
Hence, we have $F \in \mathscr{F}_{\reg}$.
Also, we have
\begin{eqnarray}
L_0(F) = 2.6, \quad L_1(F) = 3.7, \quad L_2(F) = 4.2.
\end{eqnarray}
Therefore, the average codeword length $L(F)$ of the code-tuple $F$ is given as
\begin{eqnarray}
L(F) = \pi_0(F)L_0(F) + \pi_1(F)L_1(F) + \pi_2(F)L_2(F) \approx 3.7107.
\end{eqnarray}
 \end{example}

A regular code-tuple is characterized as a code-tuple $F$ such that the set $\kernel_F$, defined as the following Definition \ref{def:kernel}, is not empty.

 \begin{definition}
\label{def:kernel}
For $F(f, \trans) \in \mathscr{F}$, we define $\kernel_F$ as
\begin{equation}
\label{eq:7tuxvj14yeno}
\kernel_F \coloneqq \{i \in [F] : {}^{\forall}j \in [F]; {}^{\exists}\pmb{x} \in \mathcal{S}^{\ast}; \trans^{\ast}_j(\pmb{x}) = i\}.
\end{equation}
Namely, $\kernel_F$ is the set of indices $i$ of code tables such that for any $j \in [F]$, there exists $\pmb{x} \in \mathcal{S}^{\ast}$ such that $\trans^{\ast}_j(\pmb{x}) = i$.
\end{definition}

\begin{example}
First, we consider $F(f, \trans) \coloneqq F^{(\alpha)}$ in Table \ref{tab:code-tuple}. Then we confirm $\kernel_{F} = \{2\}$ as follows.
\begin{itemize}
\item $0 \not\in \kernel_{F}$ because there exists no $\pmb{x} \in \mathcal{S}^{\ast}$ such that $\trans^{\ast}_2(\pmb{x}) = 0$.
\item $1 \not\in \kernel_{F}$ because there exists no $\pmb{x} \in \mathcal{S}^{\ast}$ such that $\trans^{\ast}_2(\pmb{x}) = 1$.
\item $2 \in \kernel_{F}$ because $\trans^{\ast}_0(\mathrm{bc}) = \trans^{\ast}_1(\mathrm{c}) = \trans^{\ast}_2(\lambda) = 2$.
\end{itemize}

Next, we consider $F(f, \trans) \coloneqq F^{(\beta)}$ in Table \ref{tab:code-tuple}. Then we confirm $\kernel_{F} =  \emptyset$ as follows.
\begin{itemize}
\item $0 \not\in \kernel_{F}$ because there exists no $\pmb{x} \in \mathcal{S}^{\ast}$ such that $\trans^{\ast}_1(\pmb{x}) = 0$.
\item $1 \not\in \kernel_{F}$ because there exists no $\pmb{x} \in \mathcal{S}^{\ast}$ such that $\trans^{\ast}_2(\pmb{x}) = 1$.
\item $2 \not\in \kernel_{F}$ because there exists no $\pmb{x} \in \mathcal{S}^{\ast}$ such that $\trans^{\ast}_1(\pmb{x}) = 2$.
\end{itemize}

Lastly, we consider $F(f, \trans) \coloneqq F^{(\gamma)}$ in Table \ref{tab:code-tuple}. Then we confirm $\kernel_{F} =  \{0, 1, 2\}$ as follows.
\begin{itemize}
\item $0 \in \kernel_{F}$ because $\trans^{\ast}_0(\lambda) = \trans^{\ast}_1(\mathrm{b}) = \trans^{\ast}_2(\mathrm{b}) = 0$.
\item $1 \in \kernel_{F}$ because $\trans^{\ast}_0(\mathrm{b}) = \trans^{\ast}_1(\lambda) = \trans^{\ast}_2(\mathrm{a}) = 1$.
\item $2 \in \kernel_{F}$ because $\trans^{\ast}_0(\mathrm{d}) = \trans^{\ast}_1(\mathrm{d}) = \trans^{\ast}_2(\lambda) = 2$.
\end{itemize}

Similarly, we can see $\kernel_{F^{(\delta)}} = \kernel_{F^{(\epsilon)}} = \kernel_{F^{(\zeta)}} = \kernel_{F^{(\eta)}} = \kernel_{F^{(\theta)}} = \{0, 1, 2\}$ and $\kernel_{F^{(\iota)}} = \kernel_{F^{(\kappa)}} = \{0, 1\}$.
\end{example}

Regarding $\kernel_F$, the following Lemma \ref{lem:kernel} holds.

\begin{lemma}[{\cite[Lemmas 8 and 9]{IEICE2023}}]
\label{lem:kernel}
For any $F \in \mathscr{F}$, the following statements (i)--(iii) hold.
\begin{enumerate}[(i)]
\item $F \in \mathscr{F}_{\reg}$ if and only if $\kernel_F \neq \emptyset$.
\item If $F \in \mathscr{F}_{\reg}$, then for any $i \in [F]$, the following equivalence relation holds: $\pi_i(F) > 0 \iff i \in \kernel_F$.
\item For any $F \in \mathscr{F}_{\reg} \cap \mathscr{F}_{\ext} \cap \mathscr{F}_{2\hdec}$, there exists $\bar{F} \in \mathscr{F}_{\reg} \cap \mathscr{F}_{\ext} \cap \mathscr{F}_{2\hdec}$ such that $L(\bar{F}) = L(F)$ and $\kernel_{\bar{F}} = [\bar{F}]$.
\end{enumerate}
\end{lemma}

See [5, Lemmas 8 and 9] for the proof of Lemma 9.

\section{The Optimality of AIFV Code}
\label{sec:optimality}

In this section, we prove the optimality of AIFV codes as the main result of this paper.
As stated in the previous section, we limit the scope of consideration to regular, extendable, and $2$-bit delay decodable code-tuples.
Namely, we prove the optimality of AIFV codes in the class $\mathscr{F}_{0}$ defined as the following Definition \ref{def:subsetT0}.

\begin{definition}
\label{def:subsetT0}
We define $\mathscr{F}_0$ as
\begin{equation}
\mathscr{F}_0 \coloneqq \mathscr{F}_{\reg} \cap \mathscr{F}_{\ext} \cap \mathscr{F}_{2\hdec} = \{F \in \mathscr{F}_{\reg} \cap \mathscr{F}_{2\hdec}: {}^{\forall}i \in [F]; \PREF^1_{F, i} \neq \emptyset \}.
\end{equation}
\end{definition}

We consider \emph{optimal code-tuples} in the class $\mathscr{F}_0$.
The class $\mathscr{F}_0$ is an infinite set; however, an optimal code-tuple does exist indeed as stated in the following Lemma \ref{lem:goodsetTopt}.
See the proof of Lemma \ref{lem:goodsetTopt} for \cite[Appendix B]{IEICE2023}.

\begin{lemma}[{\cite[Appendix B]{IEICE2023}}]
 \label{lem:goodsetTopt}
There exists $F \in \mathscr{F}_0$ such that for any $F' \in \mathscr{F}_0$, it holds that $L(F) \leq L(F')$ 
\end{lemma}

We define the class $\mathscr{F}_{\opt}$ of all optimal code-tuples as follows.
\begin{definition}
\label{def:optimalset}
$\mathscr{F}_{\opt} \coloneqq \argmin_{F \in \mathscr{F}_0} L(F)$.
\end{definition}
Note that $\mathscr{F}_{\opt}$ depends on the source probability distribution $\mu$, and we are now discussing for an arbitrarily fixed $\mu$.

The class of AIFV codes can be stated with our notations as the following Definition \ref{def:subsetTaifv}.

\begin{definition}
\label{def:subsetTaifv}
We define $\mathscr{F}_{\AIFV}$ as the set of all $F(f, \trans) \in \mathscr{F}^{(2)}$ 
satisfying all of the following conditions (i)--(vii).
\begin{enumerate}[(i)]
\item $f_0$ and $f_1$ are injective.
\item For any $i \in [2]$ and $s \in \mathcal{S}$, it holds that $\bar{\PREF}^1_{F, i}(f_i(s)) \not\owns 1$ and $\bar{\PREF}^1_{F, i}(f_i(s)0) \not\owns 1$.
\item For any $i \in [2]$ and $s, s' \in \mathcal{S}$, it holds that $f_i(s') \neq f_i(s)0$.
\item For any $i \in [2]$ and $s \in \mathcal{S}$, it holds that
\begin{equation}
\trans_i(s) = \begin{cases}
0 &\,\,\text{if}\,\, \bar{\PREF}^0_{F, i}(f_i(s)) = \emptyset,\\
1 &\,\,\text{if}\,\, \bar{\PREF}^0_{F, i}(f_i(s)) \neq \emptyset.\\
\end{cases}
\end{equation}
\item For any $s \in \mathcal{S}$, it holds that $f_1(s) \neq \lambda$ and $f_1(s) \neq 0$.
\item $\bar{\PREF}^1_{F, 1}(0) \not\owns 0$.

\item For any $i \in [2]$ and $\pmb{b} \in \mathcal{C}^{\ast}$, if $|\bar{\PREF}^1_{F, i}(\pmb{b})| = 1$, then at least one of the following conditions (a) and (b) hold.
\begin{enumerate}[(a)]
\item $f_i(s)\pmb{c} = \pmb{b}$ for some $s \in \mathcal{S}$ and $\pmb{c} \in \mathcal{C}^0 \cup \mathcal{C}^1$.
\item $(i, \pmb{b}) = (1, 0)$.
\end{enumerate}
\end{enumerate} 
\end{definition}

\begin{example}
The code-tuple $F^{(\kappa)}$ in Table \ref{tab:code-tuple} is in $\mathscr{F}_{\AIFV}$.
\end{example}

Now, our main theorem can be stated as follows. 
\setcounter{theorem}{0}
\begin{theorem}
\label{thm:optimality}
$\mathscr{F}_{\opt} \cap \mathscr{F}_{\AIFV} \neq \emptyset$.
\end{theorem}

Theorem \ref{thm:optimality} claims that there exists an optimal AIFV code,
that is, the class of AIFV codes achieves the optimal average codeword length in $\mathscr{F}_0$.
We prove Theorem \ref{thm:optimality} through this section.
To prove this, we introduce four classes of code-tuples $\mathscr{F}_1$, $\mathscr{F}_2$, $\mathscr{F}_3$ and $\mathscr{F}_4$, as follows.

\begin{definition}
\label{def:classes}
We define $\mathscr{F}_1$, $\mathscr{F}_2$, $\mathscr{F}_3$ and $\mathscr{F}_4$ as follows.
\begin{itemize}
\item $\mathscr{F}_1 = \{F \in \mathscr{F}_{\reg} \cap \mathscr{F}_{2\hdec}: {}^{\forall}i \in [F]; \PREF^1_{F, i} = \{0, 1\} \}$.
\item $\mathscr{F}_2 = \{F \in \mathscr{F}_{\reg} \cap \mathscr{F}_{2\hdec}: {}^{\forall}i \in [F]; |\PREF^2_{F, i}| \geq 3 \}$.
\item $\mathscr{F}_3 = \{F \in \mathscr{F}_{\reg} \cap \mathscr{F}_{2\hdec}: {}^{\forall}i \in [F]; \PREF^2_{F, i} \supseteq \{01, 10, 11\} \}$.
\item $\mathscr{F}_4 = \{F \in \mathscr{F}_{\reg} \cap \mathscr{F}_{2\hdec} \cap \mathscr{F}^{(2)}: \PREF^2_{F, 0} = \{00, 01, 10, 11\}, \PREF^2_{F, 1} = \{01, 10, 11\}\}$.
\end{itemize}
\end{definition}
By the definitions, the classes defined above form a hierarchical structure as follows: 
\begin{equation}
\mathscr{F}_0 \supseteq \mathscr{F}_1
\overset{(\mathrm{A})}{\supseteq} \mathscr{F}_2 \supseteq \mathscr{F}_3 \supseteq \mathscr{F}_4
\overset{(\mathrm{B})}{\supseteq} \mathscr{F}_{\AIFV},
\end{equation}
where 
(A) follows from Lemma \ref{lem:pref-inc} (i),
and (B) is stated as the following Lemma \ref{lem:aifv-decodable}, which proof is in Appendix \ref{subsec:proof-aifv-decodable}.

\begin{lemma}
\label{lem:aifv-decodable}
$\mathscr{F}_4 \supseteq \mathscr{F}_{\AIFV}$.
\end{lemma}

\begin{example}
\label{ex:classes}
The rightmost column of Table \ref{tab:pref} indicates the class to which each code-tuple in Table \ref{tab:code-tuple} belongs.
\end{example}

We have $\mathscr{F}_{\opt} \cap \mathscr{F}_0 \neq \emptyset$ directly from Definition \ref{def:optimalset}.
We sequentially prove $\mathscr{F}_{\opt} \cap \mathscr{F}_i \neq \emptyset$ for $i = 1, 2, 3, 4$, in Subsection \ref{subsec:T1}, \ref{subsec:T2}, \ref{subsec:T3}, \ref{subsec:T4}, respectively. 
Then in Subsection \ref{subsec:optimality}, we finally prove Theorem \ref{thm:optimality} from $\mathscr{F}_{\opt} \cap \mathscr{F}_4 \neq \emptyset$.

\subsection{The Class $\mathscr{F}_1$}
\label{subsec:T1}

In this subsection, we state the following Lemma \ref{lem:goodsetT1} and some basic properties of the class $\mathscr{F}_1$.

\begin{lemma}[{\cite[Section III]{JSAIT2022}}]
 \label{lem:goodsetT1}
$\mathscr{F}_{\opt} \cap \mathscr{F}_1 \neq \emptyset$.
\end{lemma}

See \cite[Section III]{JSAIT2022} for the complete proof of Lemma \ref{lem:goodsetT1}.
The outline of the proof is as follows.

\begin{itemize}
\item First, we define an operation called \emph{rotation}, which transforms a given code-tuple $F$ into the code-tuple $\widehat{F}$ defined as Definition \ref{def:rotation}.
\item Next, we show that $\widehat{F} \in \mathscr{F}_0$ and $L(\widehat{F}) = L(F)$ hold for any $F \in \mathscr{F}_0$ as Lemma \ref{lem:rotation}.
\item Then we show that we can transform any $F \in \mathscr{F}_{\opt} \cap \mathscr{F}_0$ into some $F' \in \mathscr{F}_{\opt} \cap \mathscr{F}_1$ by repeating of rotation. 
This shows Lemma \ref{lem:goodsetT1} since $\mathscr{F}_{\opt} \cap \mathscr{F}_0 \neq \emptyset$.
\end{itemize}

\begin{definition}
\label{def:rotation}
For $F(f, \trans) \in \mathscr{F}_{\ext}$, we define $\widehat{F}(\widehat{f}, \widehat{\trans}) \in \mathscr{F}^{(|F|)}$ as follows.

For $i \in [F]$ and $s \in \mathcal{S}$,
\begin{equation}
\label{eq:8ckwdpt4s2b5}
\widehat{f}_i(s) \coloneqq
\begin{cases}
f_i(s) d_{F, \trans_i(s)} &\,\,\text{if}\,\, \PREF^1_{F, i} = \{0, 1\},\\
\suff(f_i(s) d_{F, \trans_i(s)}) &\,\,\text{if}\,\, \PREF^1_{F, i} \neq \{0, 1\},\\
\end{cases}
\end{equation}
\begin{equation}
\label{eq:g8ovoqynvwxb}
\widehat{\trans}_i(s) = \trans_i(s),
\end{equation}
where
\begin{equation}
\label{eq:emksr78o7m54}
d_{F, i} \coloneqq
\begin{cases}
0 &\,\,\text{if}\,\, \PREF^1_{F, i} = \{0\},\\
1 &\,\,\text{if}\,\, \PREF^1_{F, i} = \{1\},\\
\lambda &\,\,\text{if}\,\, \PREF^1_{F, i} = \{0, 1\}.\\
\end{cases}
\end{equation}
\end{definition}

\begin{example}
\label{ex:rotation}
We consider $F(f, \trans) \coloneqq F^{(\delta)}$ in Table \ref{tab:code-tuple}.
Then we have
\begin{equation}
d_{F, 0} = \lambda, d_{F, 1} = \lambda, d_{F, 2} = 1
\end{equation}
since $\PREF^1_{F, 0} = \{0, 1\}$, $\PREF^1_{F, 1} = \{0, 1\}$, and $\PREF^1_{F, 2} = \{1\}$ by Table \ref{tab:pref}, respectively.
We have 
\begin{itemize}
\item $\widehat{f}_0(\mathrm{a}) = f_0(\mathrm{a})d_{F, 0} = 01$ applying the first case of (\ref{eq:8ckwdpt4s2b5}) since $\PREF^1_{F, 0} = \{0, 1\}$,
\item $\widehat{f}_0(\mathrm{d}) = f_0(\mathrm{d})d_{F, 2} = 0111$ applying the first case of (\ref{eq:8ckwdpt4s2b5}) since $\PREF^1_{F, 0} = \{0, 1\}$,
\item $\widehat{f}_2(\mathrm{a}) = \suff(f_2(\mathrm{a})d_{F, 0}) = 00$ applying the second case of (\ref{eq:8ckwdpt4s2b5}) since $\PREF^1_{F, 2} \neq \{0, 1\}$,
\item $\widehat{f}_2(\mathrm{d}) = \suff(f_2(\mathrm{d})d_{F, 2}) = 011$ applying the first case of (\ref{eq:8ckwdpt4s2b5}) since $\PREF^1_{F, 2} = \{0, 1\}$.
\end{itemize}
We have $\widehat{F^{(\gamma)}} = F^{(\delta)}, \widehat{F^{(\delta)}} = F^{(\epsilon)}$ and $\widehat{F^{(\epsilon)}} = F^{(\epsilon)}$.
\end{example}

\begin{lemma}[{\cite[Section III]{JSAIT2022}}]
\label{lem:rotation}
For any $F(f, \trans) \in \mathscr{F}_{\ext}$, the following statements (i)--(iv) hold.
\begin{enumerate}[(i)]
\item $d_{F, i} \widehat{f^{\ast}_i}(\pmb{x}) = f^{\ast}_i(\pmb{x}) d_{F, \trans^{\ast}_i(\pmb{x})}$ for any  $i \in [F]$ and $\pmb{x} \in \mathcal{S}^{\ast}$.
\item $\widehat{F} \in \mathscr{F}_{\ext}$.
\item If $F \in \mathscr{F}_{\reg}$, then $\widehat{F} \in \mathscr{F}_{\reg}$ and $L(\widehat{F}) = L(F)$.
\item For any integer $k \geq 0$, if $F \in \mathscr{F}_{k\hdec}$, then $\widehat{F} \in \mathscr{F}_{k\hdec}$.
\end{enumerate}
\end{lemma}

See \cite[Section III]{JSAIT2022} for the proof of Lemma \ref{lem:rotation}.

We now state the basic properties of $\mathscr{F}_1$ as the following Lemmas \ref{lem:propertyT1} and \ref{lem:complete-leaf}.
See Appendix \ref{subsec:proof-propertyT1} and \ref{subsec:proof-complete-leaf} for the proofs of Lemmas \ref{lem:propertyT1} and \ref{lem:complete-leaf}, respectively.

\begin{lemma}
 \label{lem:propertyT1}
For any $F(f, \trans) \in \mathscr{F}_1$ and $i \in [F]$, the following statements (i)--(vi) hold.
\begin{enumerate}[(i)]
\item $\PREF^2_{F, i} \supseteq \{0a, 1b\}$ for some $a, b \in \mathcal{C}$. In particular, $|\PREF^2_{F, i}| \geq 2$.
\item If $|\PREF^2_{F, i}| = 2$, then following statements (a) and (b) hold.
\begin{enumerate}[(a)]
\item For any $s \in \mathcal{S}$, we have $|f_i(s)| \geq 2$.
\item $\PREF^2_{F, i} = \bar{\PREF}^2_{F, i} = \{0a, 1b\}$ for some $a, b \in \mathcal{C}$.
\end{enumerate}
\item For any $s, s' \in \mathcal{S}$, if $s \neq s'$ and $f_i(s) = f_i(s')$, then $|\PREF^2_{F, \trans_i(s)}| = |\PREF^2_{F, \trans_i(s')}| = 2$.
\item For any $s \in \mathcal{S}$, we have 
\begin{equation}
|\assign_{F, i}(f_i(s))|\leq
\begin{cases}
1 &\,\,\text{if}\,\, \bar{\PREF}^0_{F, i}(f_i(s)) \neq \emptyset,\\
2 &\,\,\text{if}\,\, \bar{\PREF}^0_{F, i}(f_i(s)) = \emptyset.\\
\end{cases}
\end{equation}
\item For any $s, s' \in \mathcal{S}$, we have $f_i(s') \neq f_i(s)0$ and $f_i(s') \neq f_i(s)1$.
\item For any $s \in \mathcal{S}$, we have $|\bar{\PREF}^1_{F, i}(f_i(s) 0)| \leq 1$ and $|\bar{\PREF}^1_{F, i}(f_i(s) 1)| \leq 1$.
\end{enumerate}
\end{lemma}

\begin{lemma}
 \label{lem:complete-leaf}
For any $F(f, \trans) \in \mathscr{F}_{\opt} \cap \mathscr{F}_1, i \in \kernel_{F}$ and $s \in \mathcal{S}$, if $\bar{\PREF}^0_{F, i}(f_i(s)) = \emptyset$ and $|\assign_{F, i}(f_i(s))| = 1$, then $|\PREF^2_{F, \trans_i(s)}| = 4$.
\end{lemma}

\subsection{The Class $\mathscr{F}_2$}
\label{subsec:T2}

In this subsection, we prove $\mathscr{F}_{\opt} \cap \mathscr{F}_2 \neq \emptyset$ and some properties of the class $\mathscr{F}_2$.

\begin{itemize}
\item First, we define an operation called \emph{dot operation}, which transforms a given code-tuple $F \in \mathscr{F}_1$ into the code-tuple $\dot{F}$ defined as Definition \ref{def:tilde}.
\item Next, we consider the code-tuple $\widehat{\dot{F}}$, obtained from $F$ by applying dot operation firstly and rotation secondly.
We show that $\widehat{\dot{F}} \in \mathscr{F}_1$ and $L(\widehat{\dot{F}}) = L(F)$ hold for any $F \in \mathscr{F}_1$.
\item Then we show that we can transform any $F \in \mathscr{F}_{\opt} \cap \mathscr{F}_1$ into some $F' \in \mathscr{F}_{\opt} \cap \mathscr{F}_2$ by repeating dot operation and rotation alternately.
This shows $\mathscr{F}_{\opt} \cap \mathscr{F}_2 \neq \emptyset$ since $\mathscr{F}_{\opt} \cap \mathscr{F}_1 \neq \emptyset$ by Lemma \ref{lem:goodsetT1}.
\end{itemize}

To state the definition of $\dot{F}$, we first introduce decomposition of a codeword called \emph{$\gamma$-decomposition}.
Fix $F(f, \trans) \in \mathscr{F}_1, i \in [F]$, and $s \in \mathcal{S}$, and define $\assign^{\prec}_{F, i}(f_i(s)) \coloneqq \{s' \in \mathcal{S} : f_i(s') \prec f_i(s) \}$.
By Lemma \ref{lem:leaf} (i), we have $|\bar{\PREF}^0_{F, i}(f_i(s'))| \neq \emptyset$ for any $s' \in \assign^{\prec}_{F, i}(f_i(s))$, which leads to $|\assign_{F, i}(f_i(s'))| = 1$ by Lemma \ref{lem:propertyT1} (iv).
Thus, without loss of generality, we may assume
\begin{equation}
\label{eq:xczvww7bzthf}
f_i(s_1) \prec f_i(s_2) \prec \cdots \prec f_i(s_{\rho}),
\end{equation}
where $\assign^{\prec}_{F, i}(f_i(s)) = \{s_1, s_2, \ldots, s_{\rho-1}\}$ and $s_{\rho} \coloneqq s$.
Then there uniquely exist $\gamma(s_1), \gamma(s_2), \ldots, \gamma(s_{\rho}) \in \mathcal{C}^{\ast}$ such that 
\begin{equation}
f_i(s_r) = \begin{cases}
\gamma(s_1)&\,\,\text{if}\,\,  r = 1,\\
f_i(s_{r-1})\gamma(s_r) &\,\,\text{if}\,\,  r = 2, 3, \ldots, \rho\\
\end{cases}
\end{equation}
for any $r = 1, 2, \ldots, \rho$.
We can represent $f_i(s)$ as
\begin{equation}
\label{eq:jtxwml98jk99}
f_i(s) = \gamma(s_1)\gamma(s_2)\ldots \gamma(s_{\rho}).
\end{equation}

\begin{definition}
\label{def:gamma}
For $F(f, \trans) \in \mathscr{F}_1, i \in [F]$, and $s \in \mathcal{S}$,
we define \emph{$\gamma$-decomposition} of $f_i(s)$ as the representation in (\ref{eq:jtxwml98jk99}).
Note that $s_{\rho} = s$.
\end{definition}

\begin{example}
\label{ex:4jkiglhmtrmq}
We consider $F(f, \trans) \coloneqq F^{(\epsilon)}$ in Table \ref{tab:code-tuple}.

\begin{itemize}
\item First, we consider the $\gamma$-decomposition of $f_1(\mathrm{d})$.
We have $\assign^{\prec}_{F, 1}(f_1(\mathrm{d})) = \{\mathrm{a}, \mathrm{b}, \mathrm{c}\}$.
Since $f_1(\mathrm{b}) = \lambda \prec f_1(\mathrm{a}) = 00 \prec f_1(\mathrm{c}) = 00111$.
Thus, we obtain the $\gamma$-decomposition of $f_1(\mathrm{d})$ as 
\begin{equation}
\label{eq:4r67wru78v8m}
f_1(\mathrm{d}) = \gamma(s_1)\gamma(s_2)\gamma(s_3)\gamma(s_4),
\end{equation}
where
\begin{equation}
s_1 = \mathrm{b}, s_2 = \mathrm{a}, s_3 = \mathrm{c}, s_4 = \mathrm{d},
\end{equation}
\begin{equation}
\label{eq:9uqqbuhquew9}
\gamma(s_1) = \lambda, \gamma(s_2) = 00, \gamma(s_3) = 111, \gamma(s_4) = 11.
\end{equation}

\item Next, we consider the $\gamma$-decomposition of $f_0(\mathrm{c})$.
We have $\assign^{\prec}_{F, 0}(f_0(\mathrm{c})) = \{\mathrm{a}\}$.
Thus we obtain the $\gamma$-decomposition as
\begin{equation}
\label{eq:4ztdhmqezxcp}
f_0(\mathrm{c}) = \gamma(s_1)\gamma(s_2),
\end{equation}
where
\begin{equation}
s_1 = \mathrm{a}, s_2 = \mathrm{c},
\end{equation}
\begin{equation}
\label{eq:0kfj81nwr5vq}
\gamma(s_1) = 01, \gamma(s_2) = 00.
\end{equation}
\end{itemize}
\end{example}

We show the basic properties of $\gamma$-decomposition as the following Lemma \ref{lem:dot-length}.

\begin{lemma}
\label{lem:dot-length}
For any $F(f, \trans) \in \mathscr{F}_1$, $i \in [F]$ and $s \in \mathcal{S}$, the following statements (i)--(iii) hold, where $\gamma(s_1)\gamma(s_2)\ldots \gamma(s_{\rho})$ is the $\gamma$-decomposition of $f_i(s)$. 
\begin{enumerate}[(i)]
\item $\assign_{F, i}(\lambda) \neq \emptyset \iff f_i(s_1) = \gamma(s_1) = \lambda$.
\item For any $r = 1, 2, \ldots, \rho$, if $r \geq 2$ or $|\PREF^2_{F, i}| = 2$, then $|\gamma(s_r)| \geq 2$.
\item For any $r = 2, \ldots, \rho$, we have $g_1g_2 \in \bar{\PREF}^2_{F, i}(f_i(s_{r-1}))$, where $\gamma(s_r) = g_1g_2\ldots g_l$.
\end{enumerate}
\end{lemma}

\begin{proof}[Proof of Lemma \ref{lem:dot-length}]
(Proof of (i)): Directly from the definition of $\gamma$-decomposition.

(Proof of (ii)):
We prove for the following two cases separately: the case $r \geq 2$ and the case $r = 1, |\PREF^2_{F, i}| = 2$.
\begin{itemize}
\item The case $r \geq 2$: We have $|\gamma(s_r)| \geq 1$ by (\ref{eq:xczvww7bzthf}).
If we assume $\gamma(s_r) = c$ for some $c \in \mathcal{C}$, then $f_i(s_r) = f_i(s_{r-1})\gamma(s_r) = f_i(s_{r-1})c$ holds, which conflicts with Lemma \ref{lem:propertyT1} (v).
This shows $|\gamma(s_r)| \geq 2$ as desired.
\item The case $r = 1, |\PREF^2_{F, i}| = 2$: By Lemma \ref{lem:propertyT1} (ii) (a),
we have $|\gamma(s_1)| = |f_i(s_1)| \geq 2$.
\end{itemize}

(Proof of (iii)):
By (ii) of this lemma, we have $|\gamma(s_r)| \geq 2$.
Hence, we have $f_i(s_r) = f_i(s_{r-1})\gamma(s_r) \succeq f_i(s_{r-1})g_1g_2$,
which leads to $g_1g_2 \in \bar{\PREF}^2_{F, i}(f_i(s_{r-1}))$ as desired.
\end{proof}

Using $\gamma$-decomposition, we now state the definition of $\dot{F}$ as the following Definition \ref{def:tilde}.

\begin{definition}
\label{def:tilde}
For $F(f, \trans) \in \mathscr{F}_1$, we define $\dot{F}(\dot{f}, \dot{\trans}) \in \mathscr{F}^{(|F|)}$ as
\begin{equation}
\label{eq:gae1ondmq51s}
\dot{f}_i(s) \coloneqq \dot{\gamma}(s_1)\dot{\gamma}(s_2)\ldots \dot{\gamma}(s_{\rho}),
\end{equation}
\begin{equation}
\label{eq:bpsfrfe2v0zi}
\dot{\trans}_i(s) \coloneqq \trans_i(s)
\end{equation}
for $i \in [F]$ and $s \in \mathcal{S}$.
Here, $\dot{\gamma}(s_r)$ is defined as
\begin{equation}
\label{eq:6807mxs1xye4}
\dot{\gamma}(s_r) \coloneqq 
\begin{cases}
a_{F, i} g_1 g_3 g_4 \ldots g_l &\,\,\text{if}\,\,  r = 1, |\PREF^2_{F, i}| = 2,\\
\gamma(s_r) &\,\,\text{if}\,\,  r = 1, |\PREF^2_{F, i}| \geq 3,\\
\bar{a}_{F, \trans_i(s_{r-1})} g_1 g_3 g_4 \ldots g_l &\,\,\text{if}\,\, r \geq 2, |\bar{\PREF}^1_{F, i}(f_i(s_{r-1}))| = 2,\\
\bar{a}_{F, \trans_i(s_{r-1})} 0 g_3 g_4 \ldots g_l  &\,\,\text{if}\,\, r \geq 2, |\bar{\PREF}^1_{F, i}(f_i(s_{r-1}))| = 1, |\bar{\PREF}^1_{F, \trans_i(s_{r-1})}| = 1,\\
\bar{a}_{F, \trans_i(s_{r-1})} 1 g_3 g_4 \ldots g_l &\,\,\text{if}\,\, r \geq 2, |\bar{\PREF}^1_{F, i}(f_i(s_{r-1}))| = 1, |\bar{\PREF}^1_{F, \trans_i(s_{r-1})}|  = 2, |\PREF^2_{F, \trans_i(s_{r-1})}| = 2,\\
\gamma(s_r) &\,\,\text{if}\,\, r \geq 2, |\bar{\PREF}^1_{F, i}(f_i(s_{r-1}))| = 1, |\bar{\PREF}^1_{F, \trans_i(s_{r-1})}|  = 2, |\PREF^2_{F, \trans_i(s_{r-1})}| \geq 3\\
\end{cases}
\end{equation}
for $r = 1, 2, \ldots, \rho$, where $\gamma(s_1)\gamma(s_2)\ldots \gamma(s_{\rho})$ is the $\gamma$-decomposition of $f_i(s)$ and $\gamma(s_r) = g_1g_2\ldots g_l$.
Also, $a_{F, i} \in \mathcal{C}$ is defined by the following recursive formula:
\begin{equation}
\label{v47umw9fprm9}
a_{F, i} \coloneqq 
\begin{cases}
a_{F, \trans_i(s')} &\,\,\text{if}\,\, \assign_{F, i}(\lambda) = \{s'\} \,\,\text{for some}\,\, s' \in \mathcal{S'},\\
0 &\,\,\text{if}\,\, |\assign_{F, i}(\lambda)| \neq 1, \PREF^2_{F, i} \owns 00,\\
1 &\,\,\text{if}\,\, |\assign_{F, i}(\lambda)| \neq 1, \PREF^2_{F, i} \not\owns 00
\end{cases}
\end{equation}
and $\bar{a}_{F, i}$ denotes the negation of $a_{F, i}$, that is, $\bar{a}_{F, i} \coloneqq 1 - a_{F, i}$.

We refer to the operation of obtaining the code-tuple $\dot{F}$ from a given code-tuple $F \in \mathscr{F}_1$ as \emph{dot operation}.
\end{definition}

\begin{remark}
In Definition \ref{def:tilde}, it holds that $|\gamma(s_r)| < 2$ only if $r = 1$ and $|\PREF^2_{F, i}| \geq 3$ by Lemma \ref{lem:dot-length} (ii).
Hence, the right hand side of (\ref{eq:6807mxs1xye4}) has enough length so that $\dot{\gamma}(s_r)$ is well-defined for every case.
\end{remark}

\begin{example}
We consider $F(f, \trans) \coloneqq F^{(\epsilon)}$ in Table \ref{tab:code-tuple}.
Then $a_{F, i}, i \in [F]$ are given as follows.
\begin{itemize}
\item $a_{F, 0} = 1$ applying the third case of (\ref{v47umw9fprm9}) since $|\assign_{F, 0}(\lambda)| \neq 1$ and $\PREF^2_{F, 0} \not\owns 00$.
\item $a_{F, 2} = 0$ applying the second case of (\ref{v47umw9fprm9}) since $|\assign_{F, 2}(\lambda)| \neq 1$ and $\PREF^2_{F, 0} \owns 00$.
\item $a_{F, 1} = a_{F, 0} = 1$ applying the first case of (\ref{v47umw9fprm9}) since $|\assign_{F, 1}(\lambda)| = \{\mathrm{b}\}$.
\end{itemize}

The codeword $\dot{f}_0(\mathrm{c})$ is obtained as follows since the $\gamma$-decomposition of $f_0(\mathrm{c})$ is given as (\ref{eq:4ztdhmqezxcp})--(\ref{eq:0kfj81nwr5vq}).
\begin{itemize}
\item we have $\dot{\gamma}(s_1) = a_{F, 0}0 = 10$ applying the first case of (\ref{eq:6807mxs1xye4}) since $|\PREF^2_{F, 0}| = 2$,
\item we have $\dot{\gamma}(s_2) = \bar{a}_{F, \trans_0(s_1)}0 = \bar{a}_{F, 1}0 = 00$ applying the third case of (\ref{eq:6807mxs1xye4}) since $|\bar{\PREF}^1_{F, 0}(f_0(s_1))| = |\bar{\PREF}^1_{F, 0}(01)| = 2$.
\end{itemize}
Therefore, we obtain $\dot{f}_0(\mathrm{c}) = \dot{\gamma}(s_1)\dot{\gamma}(s_2) = 1000$.

The codeword $\dot{f}_1(\mathrm{d})$ is obtained as follows
 since the $\gamma$-decomposition of $f_1(\mathrm{d})$ is given as (\ref{eq:4r67wru78v8m})--(\ref{eq:9uqqbuhquew9}).
\begin{itemize}
\item we have $\dot{\gamma}(s_1) = \gamma(s_1) = \lambda$ applying the second case of (\ref{eq:6807mxs1xye4}) since $|\PREF^2_{F, 1}| \geq 3$,
\item we have $\dot{\gamma}(s_2) = \bar{a}_{F, \trans_0(s_1)}1 = \bar{a}_{F, 0}1 = 01$ applying the fifth case of (\ref{eq:6807mxs1xye4}) since $|\bar{\PREF}^1_{F, 1}(f_1(s_1))| = |\bar{\PREF}^1_{F, 1}| = 1$, $|\bar{\PREF}^1_{F, \trans_1(s_1)}| = |\bar{\PREF}^1_{F, 0}| = 2$, and $|\PREF^2_{F, \trans_1(s_1)}| = |\PREF^2_{F, 0}| = 2$,
\item we have $\dot{\gamma}(s_3) = \bar{a}_{F, \trans_1(s_2)}00 = \bar{a}_{F, 1}1 = 001$ applying the fourth case of (\ref{eq:6807mxs1xye4}) since $|\bar{\PREF}^1_{F, 1}(f_1(s_2))| = |\bar{\PREF}^1_{F, 1}(00)| = 1$ and $|\bar{\PREF}^1_{F, \trans_1(s_2)}| = |\bar{\PREF}^1_{F, 1}| = 1$.
\end{itemize}
Therefore, we obtain $\dot{f}_1(\mathrm{d}) = \dot{\gamma}(s_1)\dot{\gamma}(s_2)\dot{\gamma}(s_3) = 01001$.

The code table $F^{(\zeta)}$ in Table \ref{tab:code-tuple} is obtained as $\dot{F}^{(\epsilon)}$.
Moreover, the code table $F^{(\eta)}$ in Table \ref{tab:code-tuple} is obtained as $\widehat{F^{(\zeta)}} (= \widehat{\dot{F}^{(\epsilon)}})$.
\end{example}

Now we enumerate some properties of $\dot{F}$ as the following Lemmas \ref{lem:tilde-equal}--\ref{lem:tilde-pref2}.

\begin{lemma}
\label{lem:tilde-equal}
For any $F(f, \trans) \in \mathscr{F}_1$ and $i \in [F]$, the following statements (i)--(iii) hold.
\begin{enumerate}[(i)]
\item Let $s \in \mathcal{S}$ and let $\gamma(s_1)\gamma(s_2)\ldots \gamma(s_{\rho})$ be the $\gamma$-decomposition of $f_i(s)$.
Then we have $|\dot{\gamma}(s_r)| = |\gamma(s_r)|$ for any $r = 1, 2, \ldots, \rho$.
\item For any $s \in \mathcal{S}$, we have $|\dot{f}_i(s)| = |f_i(s)|$.
\item For any $s, s' \in \mathcal{S}$, we have $f_i(s) \preceq f_i(s') \iff \dot{f}_i(s) \preceq \dot{f}_i(s')$.
\end{enumerate}
\end{lemma}

\begin{proof}[Proof of Lemma \ref{lem:tilde-equal}]
(Proof of (i)): Directly from (\ref{eq:6807mxs1xye4}).

(Proof of (ii)): 
We have
\begin{eqnarray}
|\dot{f}_i(s)|
= |\dot{\gamma}(s_1)| + |\dot{\gamma}(s_2)| + \cdots + |\dot{\gamma}(s_{\rho})|
\overset{(\mathrm{A})}{=} |\gamma(s_1)| + |\gamma(s_2)| + \cdots + |\gamma(s_{\rho})|
=  |f_i(s)|,
\end{eqnarray}
where (A) follows from (i) of this lemma.

(Proof of (iii)): See Appendix \ref{subsec:proof-tilde-equal}.
\end{proof}

\begin{lemma}
\label{lem:tilde-pref}
For any $F(f,\trans) \in \mathscr{F}_1$ and $i \in [F]$, the following statements (i) and (ii) hold.
\begin{enumerate}[(i)]
\item 
\begin{enumerate}[(a)]
\item If $|\PREF^2_{F, i}| = 2$, then $\PREF^2_{\dot{F}, i} = \{a_{F, i}0, a_{F, i}1\}$.
\item For any $s \in \mathcal{S}$, if $|\PREF^2_{F, j}| \geq 3$, then
\begin{equation}
\label{eq:ldq1qk35kx1g}
\PREF^2_{\dot{F}, j} \subseteq
\begin{cases}
\{00, 01, 10, 11\} &\,\,\text{if}\,\, |\bar{\PREF}^1_{F, i}(f_i(s))| = 0, \\
\{a_{F, j}0, a_{F, j}1, \bar{a}_{F, j}1\}  &\,\,\text{if}\,\,  |\bar{\PREF}^1_{F, i}(f_i(s))| = 1, |\bar{\PREF}^1_{F, j}| = 1,\\
\PREF^2_{F, j} &\,\,\text{if}\,\,  |\bar{\PREF}^1_{F, i}(f_i(s))| = 1, |\bar{\PREF}^1_{F, j}| = 2,\\
\end{cases}
\end{equation}
where $j \coloneqq \trans_i(s) = \dot{\trans}_i(s)$.
\end{enumerate}

\item For any $s \in \mathcal{S}$, we have
\begin{equation}
\label{eq:v5z4ipil23dc}
\bar{\PREF}^2_{\dot{F}, i}(\dot{f}_i(s)) \subseteq
\begin{cases}
\emptyset &\,\,\text{if}\,\, |\bar{\PREF}^1_{F, i}(f_i(s))| = 0,\\
\{\bar{a}_{F, j}0, \bar{a}_{F, j}1\} &\,\,\text{if}\,\, |\bar{\PREF}^1_{F, i}(f_i(s))| \geq 1, |\PREF^2_{F, j}| = 2,\\
\{\bar{a}_{F, j}0\}  &\,\,\text{if}\,\, |\bar{\PREF}^1_{F, i}(f_i(s))| \geq 1, |\PREF^2_{F, j}| \geq 3, |\bar{\PREF}^1_{F, j}| = 1,\\
\bar{\PREF}^2_{F, i}(f_i(s)) &\,\,\text{if}\,\, |\bar{\PREF}^1_{F, i}(f_i(s))| \geq 1, |\PREF^2_{F, j}| \geq 3, |\bar{\PREF}^1_{F, j}| = 2,\\
\end{cases}
\end{equation}
where $j \coloneqq \trans_i(s) = \dot{\trans}_i(s)$.
\end{enumerate}
\end{lemma}
See Appendix \ref{subsec:proof-tilde-pref} for the proof of Lemma \ref{lem:tilde-pref}.

The next lemma relates to $d_{F, i}$ and $a_{F, i}$ defined in Definitions \ref{def:rotation} and \ref{def:tilde}, respectively.

\begin{lemma}
\label{lem:tilde-pref2}
For any $F(f,\trans) \in \mathscr{F}_1$ and $i \in[F]$, the following statements (i) and (ii) hold.
\begin{enumerate}[(i)]
\item If $|\PREF^2_{F, i}| = 2$, then $d_{\dot{F}, i} = a_{F, i}$.
\item For any $s, s' \in \mathcal{S}$, if $s \neq s'$ and $\dot{f}_i(s) = \dot{f}_i(s')$,
then $d_{\dot{F}, \dot{\trans}_i(s)} = a_{F, \trans_i(s)} \neq a_{F, \trans_i(s')} = d_{\dot{F}, \dot{\trans}_i(s')}$.
\end{enumerate}
\end{lemma}

See Appendix \ref{subsec:proof-tilde-pref2} for the proof of Lemma \ref{lem:tilde-pref2}.

Using the properties above, we now prove the following Lemma \ref{lem:transA-preserve}.

\begin{lemma}
\label{lem:transA-preserve}
For any $F \in \mathscr{F}_1$, we have $\widehat{\dot{F}} \in \mathscr{F}_1$ and $L(\widehat{\dot{F}}) = L(F)$.
\end{lemma}

\begin{proof}[Proof of Lemma \ref{lem:transA-preserve}]
It suffices to prove the following three statements (i)--(iii) for any $F \in \mathscr{F}_1$.
\begin{enumerate}[(i)]
\item $\widehat{\dot{F}} \in \mathscr{F}_{2\hdec}$. 
\item $\PREF^1_{\widehat{\dot{F}}, i} = \{0, 1\}$ for any $i \in [F]$.
\item $\widehat{\dot{F}} \in \mathscr{F}_{\reg}$ and $L(\widehat{\dot{F}}) = L(F)$.
\end{enumerate}

(Proof of (i)):
It suffices to prove $\dot{F} \in \mathscr{F}_{2\hdec}$ because this implies $\widehat{\dot{F}} \in \mathscr{F}_{2\hdec}$ by Lemma \ref{lem:rotation} (iv).

We first show that $\dot{F}$ satisfies Definition \ref{def:k-bitdelay} (i).
Choose $i \in [F]$ and $s \in \mathcal{S}$ arbitrarily and put $j \coloneqq \trans_i(s)$.
We consider the following two cases separately: the case $|\bar{\PREF}^1_{F, i}(f_i(s))| = 0$ and the case $|\bar{\PREF}^1_{F, i}(f_i(s))| \geq 1$.
\begin{itemize}
\item The case $|\bar{\PREF}^1_{F, i}(f_i(s))| = 0$:
We have
\begin{equation}
\PREF^2_{\dot{F}, j} \cap \bar{\PREF}^2_{\dot{F}, i}(\dot{f}_i(s))
\overset{(\mathrm{A})}{\subseteq} \{00, 01, 10, 11\} \cap \bar{\PREF}^2_{\dot{F}, i}(\dot{f}_i(s))
\overset{(\mathrm{B})}{\subseteq} \{00, 01, 10, 11\} \cap \emptyset
= \emptyset
\end{equation}
as desired, where
(A) follows from $|\bar{\PREF}^1_{F, i}(f_i(s))| = 0$ and the first case of (\ref{eq:ldq1qk35kx1g}),
and (B) follows from $|\bar{\PREF}^1_{F, i}(f_i(s))| = 0$ and the first case of (\ref{eq:v5z4ipil23dc}).

\item The case $|\bar{\PREF}^1_{F, i}(f_i(s))| \geq 1$:
We consider the following three cases separately: the case $|\PREF^2_{F, j}| = 2$, the case $|\PREF^2_{F, j}| \geq 3, |\bar{\PREF}^1_{F, j}| = 1$, and the case $|\PREF^2_{F, j}| \geq 3, |\bar{\PREF}^1_{F, j}| = 2$.
\begin{itemize}
\item The case $|\PREF^2_{F, j}| = 2$: We have
\begin{equation}
\PREF^2_{\dot{F}, j} \cap \bar{\PREF}^2_{\dot{F}, i}(\dot{f}_i(s))
\overset{(\mathrm{A})}{=}\{a_{F, j}0, a_{F, j}1\} \cap \bar{\PREF}^2_{\dot{F}, i}(\dot{f}_i(s))
\overset{(\mathrm{B})}{\subseteq}\{a_{F, j}0, a_{F, j}1\} \cap \{\bar{a}_{F, j}0, \bar{a}_{F, j}1\} = \emptyset
\end{equation}
as desired, where
(A) follows from $|\PREF^2_{F, j}| = 2$ and Lemma \ref{lem:tilde-pref} (i) (a),
and (B) follows from $|\bar{\PREF}^1_{F, i}(f_i(s))| \geq 1$, $|\PREF^2_{F, j}| = 2$, and the second case of (\ref{eq:v5z4ipil23dc}).

\item The case $|\PREF^2_{F, j}| \geq 3$:
Then we have $|\bar{\PREF}^1_{F, i}(f_i(s))| \leq 1$ by Lemma \ref{cor:pref-sum}.
Combining this with $|\bar{\PREF}^1_{F, i}(f_i(s))| \geq 1$, we obtain
\begin{equation}
\label{eq:fccspjsv4h5m}
|\bar{\PREF}^1_{F, i}(f_i(s))| = 1.
\end{equation}
\begin{itemize}
\item The case $|\bar{\PREF}^1_{F, j}| = 1$: We have
\begin{equation}
\PREF^2_{\dot{F}, j} \cap \bar{\PREF}^2_{\dot{F}, i}(\dot{f}_i(s))
\overset{(\mathrm{A})}{\subseteq} \{a_{F, j}0, a_{F, j}1, \bar{a}_{F, j}1\} \cap \bar{\PREF}^2_{\dot{F}, i}(\dot{f}_i(s))
\overset{(\mathrm{B})}{\subseteq} \{a_{F, j}0, a_{F, j}1, \bar{a}_{F, j}1\} \cap \{\bar{a}_{F, j}0\} = \emptyset,
\end{equation}
where (A) follows from (\ref{eq:fccspjsv4h5m}), $|\bar{\PREF}^1_{F, j}| = 1$, and the second case of (\ref{eq:ldq1qk35kx1g}),
and (B) follows from $|\bar{\PREF}^1_{F, i}(f_i(s))| \geq 1$, $|\PREF^2_{F, j}| \geq 3$, $|\bar{\PREF}^1_{F, j}| = 1$, and the third case of (\ref{eq:v5z4ipil23dc}).

\item The case $|\bar{\PREF}^1_{F, j}| = 2$: 
We have
\begin{equation}
\PREF^2_{\dot{F}, j} \cap \bar{\PREF}^2_{\dot{F}, i}(\dot{f}_i(s))
\overset{(\mathrm{A})}{\subseteq} \PREF^2_{F, j} \cap \bar{\PREF}^2_{\dot{F}, i}(\dot{f}_i(s))
\overset{(\mathrm{B})}{\subseteq} \PREF^2_{F, j} \cap \bar{\PREF}^2_{F, i}(f_i(s))
\overset{(\mathrm{C})}{=} \emptyset,
\end{equation}
where 
(A) follows from (\ref{eq:fccspjsv4h5m}), $|\bar{\PREF}^1_{F, j}| = 2$, and the third case of (\ref{eq:ldq1qk35kx1g}),
(B) follows from $|\bar{\PREF}^1_{F, i}(f_i(s))| \geq 1$, $|\PREF^2_{F, j}| \geq 3$, $|\bar{\PREF}^1_{F, j}| = 2$, and the fourth case of (\ref{eq:v5z4ipil23dc}),
and (C) follows from $F \in \mathscr{F}_{2\hdec}$.
\end{itemize}

\end{itemize}
\end{itemize}
These cases show that $\dot{F}$ satisfies Definition \ref{def:k-bitdelay} (i).

Next, we show that $\dot{F}$ satisfies Definition \ref{def:k-bitdelay} (ii).
Choose $i \in [F]$ and $s, s' \in \mathcal{S}$ such that 
\begin{equation}
\label{eq:flmngh4pebrz}
s \neq s', \quad \dot{f}_i(s) = \dot{f}_i(s')
\end{equation}
arbitrarily and put $j \coloneqq \trans_i(s)$.
Since (\ref{eq:flmngh4pebrz}) and Lemma \ref{lem:tilde-equal} (iii) lead to $f_i(s) = f_i(s')$, we have
\begin{equation}
\label{eq:4cz3f30drctg}
|\PREF^2_{F, \trans_i(s)}| = |\PREF^2_{F, \trans_i(s')}| = 2
\end{equation}
applying Lemma \ref{lem:propertyT1} (iii).
Hence, we obtain
\begin{equation}
\label{eq:8t0598t6za5v}
\PREF^2_{\dot{F}, \trans_i(s)} \cap \PREF^2_{\dot{F}, \trans_i(s')}
\overset{(\mathrm{A})}{=} \{a_{F, \trans_i(s)}0, a_{F, \trans_i(s)}1\} \cap \{a_{F, \trans_i(s')}0, a_{F, \trans_i(s')}1\}
\overset{(\mathrm{B})}{=}  \emptyset
\end{equation}
as desired, where
(A) follows from (\ref{eq:4cz3f30drctg}) and Lemma \ref{lem:tilde-pref} (i) (a),
and (B) follows since $a_{F, \trans_i(s)} \neq a_{F, \trans_i(s')}$ by (\ref{eq:flmngh4pebrz}) and Lemma \ref{lem:tilde-pref2} (ii).

(Proof of (ii)):
We prove for the following two cases separately: (I) the case $\assign_{F, i}(\lambda) = \emptyset$; (II) the case $\assign_{F, i}(\lambda) \neq \emptyset$.
\begin{enumerate}[(I)]
\item The case $\assign_{F, i}(\lambda) = \emptyset$: 
It suffices to show
\begin{equation}
\label{eq:uif3yx8nh5hg}
{}^{\forall}c \in \mathcal{C}; {}^{\exists}\pmb{x} \in \mathcal{S}^{\ast}; \dot{f}^{\ast}_i(\pmb{x}) \succeq d_{\dot{F}, i}c
\end{equation}
because this implies that for any $c \in \mathcal{C}$, there exists $\pmb{x} \in \mathcal{S}^{\ast}$ such that 
\begin{equation}
d_{\dot{F}, i}c \overset{(\mathrm{A})}{\preceq} \dot{f}^{\ast}_i(\pmb{x}) \preceq \dot{f}^{\ast}_i(\pmb{x}) d_{\dot{F}, \trans^{\ast}_i(\pmb{x})} \overset{(\mathrm{B})}{=} d_{\dot{F}, i} \widehat{\dot{f}}^{\ast}_i(\pmb{x}),
\end{equation}
where
(A) follows from (\ref{eq:uif3yx8nh5hg}),
and (B) follows from Lemma \ref{lem:rotation} (i).
This shows that $\widehat{\dot{f}}^{\ast}_i(\pmb{x}) \succeq c$ for some $\pmb{x} \in \mathcal{S}^{\ast}$,
which leads to $c \in \PREF^1_{\widehat{\dot{F}}, i}$ as desired.
Thus, we prove (\ref{eq:uif3yx8nh5hg}) considering the following two cases separately: the case $|\PREF^2_{F, i}| = 2$ and the case $|\PREF^2_{F, i}| \geq 3$.
\begin{itemize}
\item The case $|\PREF^2_{F, i}| = 2$: For any $c \in \mathcal{C}$, we have
\begin{equation}
\PREF^2_{\dot{F}, i} \overset{(\mathrm{A})}{=} \{a_{F, i}0, a_{F, i}1\} \overset{(\mathrm{B})}{=} \{d_{\dot{F}, i}0, d_{\dot{F}, i}1\} \owns d_{\dot{F}, i}c,
\end{equation}
where
(A) follows from Lemma \ref{lem:tilde-pref} (i) (a),
and (B) follows from Lemma \ref{lem:tilde-pref2} (i).
Hence, there exists $\pmb{x} \in \mathcal{S}^{+}$ such that $\dot{f}^{\ast}_{\dot{F}, i}(\pmb{x}) \succeq d_{\dot{F}, i}c$ as desired.

\item The case $|\PREF^2_{F, i}| \geq 3$:
Choose $c \in \mathcal{C}$ arbitrarily.
We have $\PREF^1_{F, i} = \{0, 1\} \owns c$ by $F \in \mathscr{F}_1$.
Hence, there exists $\pmb{x} = x_1x_2\ldots x_n\in \mathcal{S}^{+}$ such that $f^{\ast}_i(\pmb{x}) \succeq c$.
Let $\gamma(s_1)\gamma(s_2)\ldots \gamma(s_{\rho})$ be the $\gamma$-decomposition of $f_i(x_1)$.
We have
\begin{equation}
\label{eq:etbrghjnh494}
\dot{f}^{\ast}_i(\pmb{x}) \succeq \dot{f}_i(x_1) \succeq \dot{\gamma}(s_1) \overset{(\mathrm{A})}{=} \gamma(s_1) \overset{(\mathrm{B})}{\succeq} c,\end{equation}
where
(A) follows from $|\PREF^2_{F, i}| \geq 3$ and the second case of (\ref{eq:6807mxs1xye4}),
and (B) follows from $\assign_{F, i}(\lambda) = \emptyset$ and Lemma \ref{lem:dot-length} (i).

Since $c$ is arbitrarily chosen, we have $\PREF^1_{\dot{F}, i} = \{0, 1\}$ by (\ref{eq:etbrghjnh494}).
This implies $d_{\dot{F}, i} = \lambda$ by $(\ref{eq:emksr78o7m54})$.
Therefore, by (\ref{eq:etbrghjnh494}), we obtain $\dot{f}^{\ast}_i(\pmb{x}) \succeq c = d_{\dot{F}, i}c$ for any $c \in \mathcal{C}$ as desired.
\end{itemize}

\item The case $\assign_{F, i}(\lambda) \neq \emptyset$:
By Lemma \ref{lem:longest}, we can choose the longest sequence $\pmb{x} \in \mathcal{S}^{+}$ such that $f^{\ast}_i(\pmb{x}) = \lambda$.
Then $\assign_{F, \trans^{\ast}_i(\pmb{x})}(\lambda) = \emptyset$.
Hence, from the result of the case (I) above, we have $\PREF^2_{\widehat{\dot{F}}, \trans^{\ast}_i(\pmb{x})} = \{0, 1\}$.
Thus, we obtain
\begin{equation}
\PREF^2_{\widehat{\dot{F}}, i}
\overset{(\mathrm{A})}{\supseteq} \PREF^2_{\widehat{\dot{F}}, \trans^{\ast}_i(x_1)}
\overset{(\mathrm{A})}{\supseteq} \PREF^2_{\widehat{\dot{F}}, \trans^{\ast}_i(x_1x_2)}
\overset{(\mathrm{A})}{\supseteq} \cdots
\overset{(\mathrm{A})}{\supseteq} \PREF^2_{\widehat{\dot{F}}, \trans^{\ast}_i(\pmb{x})} = \{0, 1\}
\end{equation}
as desired, where (A)s follow from Lemma \ref{lem:pref-sum} (i).
\end{enumerate}

(Proof of (iii)):
We have 
\begin{equation}
\label{eq:1x4ezuoeffus}
Q(F) \overset{(\mathrm{A})}{=} Q(\dot{F}) \overset{(\mathrm{B})}{=} Q(\widehat{\dot{F}}),
\end{equation}
where
(A) follows from (\ref{eq:bpsfrfe2v0zi}),
and (B) follows from (\ref{eq:g8ovoqynvwxb}) (cf. Remark \ref{rem:transitionprobability}).
Hence, $F \in \mathscr{F}_{\reg}$ implies $\widehat{\dot{F}} \in \mathscr{F}_{\reg}$.
Also, we have 
\begin{equation}
L(F) \overset{(\mathrm{A})}{=} L(\dot{F}) \overset{(\mathrm{B})}{=} L(\widehat{\dot{F}}),
\end{equation}
where
(A) follows from (\ref{eq:1x4ezuoeffus}) and Lemma \ref{lem:tilde-equal} (ii) (cf. Remark \ref{rem:transitionprobability}),
and (B) follows from Lemma \ref{lem:rotation} (iii).
\end{proof}

For $F \in \mathscr{F}_1$ and an integer $t \geq 0$, we define
\begin{equation}
F^{(t)} =
\begin{cases}
F &\,\,\text{if}\,\,  t = 0,\\
\widehat{\dot{F^{(t-1)}}} &\,\,\text{if}\,\, t > 0.\\
\end{cases}
\end{equation}
Namely, $F^{(t)}$ is the code-tuple obtained by applying dot operation and rotation to $F$ $t$ times.
We now prove that any code-tuple of $\mathscr{F}_1$ is transformed into a code-tuple of $\mathscr{F}_2$ by repeating of dot operation and rotation, that is, $\mathcal{M}_{F^{(t)}} = \emptyset$ holds for a sufficiently large $t$, where $\mathcal{M}_F \coloneqq \{i \in [F] : |\PREF^2_{F, i}| = 2\}$.
To prove this fact, we use the following Lemma \ref{lem:transA-arg4}.
See Appendix \ref{subsec:proof-transA-arg4} for the proof of Lemma  \ref{lem:transA-arg4}.

\begin{lemma}
\label{lem:transA-arg4}
For any $F \in \mathscr{F}_{\opt} \cap \mathscr{F}_1$ such that $\kernel_F = [F]$ and two integers $t$ and $t'$ such that $0 \leq t < t'$, it holds that $\mathcal{M}_{F^{(t)}} \cap \mathcal{M}_{F^{(t')}}  = \emptyset$.
\end{lemma}

\begin{lemma}
 \label{lem:goodsetT2}
$\mathscr{F}_{\opt} \cap \mathscr{F}_2 \neq \emptyset$.
\end{lemma}

\begin{proof}[Proof of Lemma \ref{lem:goodsetT2}]
By Lemma \ref{lem:goodsetT1}, there exists $F \in \mathscr{F}_{\opt} \cap \mathscr{F}_1$. 
By Lemma \ref{lem:kernel} (iii), we may assume $\kernel_F = [F]$ without loss of generality.
Consider $|F|+1$ code-tuples $F^{(0)}, F^{(1)}, \ldots, F^{(|F|)}$.
Because Lemma \ref{lem:transA-arg4} shows that the $|F|+1$ sets $\mathcal{M}_{F^{(0)}}, \mathcal{M}_{F^{(1)}}, \ldots, \mathcal{M}_{F^{(|F|)}}$ are disjoint, there exists an integer $\bar{t} \in \{0, 1, 2, \ldots, |F|\}$ such that $\mathcal{M}_{F^{(\bar{t})}} = \emptyset$.
This shows that $|\PREF^2_{F^{(\bar{t})}, i}| \geq 3$ for any $i \in [F]$. 
Since $F^{(\bar{t})} \in \mathscr{F}_{\opt} \cap \mathscr{F}_1$ by Lemma \ref{lem:transA-preserve}, 
we obtain $F^{(\bar{t})} \in \mathscr{F}_{\opt} \cap \mathscr{F}_2$.
\end{proof}

We state some properties of $\mathscr{F}_2$ as the following Lemmas \ref{lem:propertyT2} and \ref{lem:complete-leaf2}.

\begin{lemma}
 \label{lem:propertyT2}
For any $F(f, \trans) \in \mathscr{F}_2$ and $i \in [F]$, the mapping $f_i$ is injective.
\end{lemma}

\begin{proof}[Proof of Lemma \ref{lem:propertyT2}]
For any $s \in \mathcal{S}$, we have
\begin{equation}
|\assign_{F, i}(f_i(s))| = \frac{3|\assign_{F, i}(f_i(s))|}{3} \overset{(\mathrm{A})}{\leq} \frac{\sum_{s' \in \assign_{F, i}(f_i(s))} |\PREF^2_{F, \trans_i(s')}|}{3} \overset{(\mathrm{B})}{\leq} \frac{|\PREF^2_{F, i}(f_i(s))|}{3} \leq \frac{4}{3},
\end{equation}
where
(A) follows since $|\PREF^2_{F, \trans_i(s')}| \geq 3$ for any $s' \in \assign_{F, i}(f_i(s))$ from $F \in \mathscr{F}_2$,
and (B) follows from Lemma \ref{lem:pref-sum} (ii).
Therefore, we have $|\assign_{F, i}(f_i(s))| \leq 1$ for any $s \in \mathcal{S}$.
This shows that $f_i$ is injective as desired.
\end{proof}

\begin{lemma}
 \label{lem:complete-leaf2}
For any $F(f, \trans) \in \mathscr{F}_{\opt} \cap \mathscr{F}_2$, there exists $i \in \kernel_F$ such that $|\PREF^2_{F, i}| = 4$.
\end{lemma}
\begin{proof}[Proof of Lemma \ref{lem:complete-leaf2}]
Choose $p \in \kernel_F$.
By Lemma \ref{lem:leaf} (ii), there exists $s \in \mathcal{S}$ such that $\bar{\PREF}^0_{F, p}(f_p(s)) = \emptyset$.
Also, by Lemma \ref{lem:propertyT2}, we have $|\assign_{F, p}(f_p(s))| = 1$.
Hence, by Lemma \ref{lem:complete-leaf}, we obtain $|\PREF^2_{F, i}| = 4$ for $i \coloneqq \trans_p(s)$.

By $p \in \kernel_F$, for any $j \in [F]$, there exists $\pmb{x} \in \mathcal{S}^{\ast}$ such that $\trans^{\ast}_j(\pmb{x}) = p$, which leads to
\begin{equation}
\trans^{\ast}_j(\pmb{x}s)
\overset{(\mathrm{A})}{=} \trans_{\trans^{\ast}_j(\pmb{x})}(s)
= \trans_p(s)
= i,
\end{equation}
where (A) follows from Lemma \ref{lem:f_T} (ii).
This shows $i \in \kernel_{F}$.
\end{proof}

\subsection{The Class $\mathscr{F}_3$}
\label{subsec:T3}

In this subsection, we prove $\mathscr{F}_{\opt} \cap \mathscr{F}_3 \neq \emptyset$,
which proof is outlined as follows.

\begin{itemize}
\item First, we define the code-tuple $\ddot{F}$ as Definition \ref{def:tilde2} for a given code-tuple $F \in \mathscr{F}_2$.
\item Then we show that $\ddot{F} \in \mathscr{F}_{\opt} \cap \mathscr{F}_3$ holds for any $F \in \mathscr{F}_{\opt} \cap \mathscr{F}_2$.
This shows $\mathscr{F}_{\opt} \cap \mathscr{F}_3 \neq \emptyset$ since $\mathscr{F}_{\opt} \cap \mathscr{F}_2 \neq \emptyset$ by Lemma \ref{lem:goodsetT2}.
\end{itemize}

\begin{definition}
\label{def:tilde2}
For $F(f, \trans) \in \mathscr{F}_2$, we define $\ddot{F}(\ddot{f}, \ddot{\trans}) \in \mathscr{F}^{(|F|)}$ as 
\begin{equation}
\label{eq:w5alb8rohz3x}
\ddot{f}_i(s) \coloneqq \ddot{\gamma}(s_1)\ddot{\gamma}(s_2)\ldots \ddot{\gamma}(s_{\rho}),
\end{equation}
\begin{equation}
\label{eq:hlr7zybfne7c}
\ddot{\trans}_i(s) \coloneqq \trans_i(s)
\end{equation}
for $i \in [F]$ and $s \in \mathcal{S}$.
Here, $\ddot{\gamma}(s_r)$ is defined as
\begin{equation}
\label{eq:oezl1k2insm5}
\ddot{\gamma}(s_r) = 
\begin{cases} 
\gamma(s_r) &\,\,\text{if}\,\, r = 1, |\PREF^2_{F, i}| = 4,\\
1 &\,\,\text{if}\,\, r = 1, |\PREF^2_{F, i}| = 3, |\gamma(s_r)| = 1,\\
01g_3 g_4 \ldots g_l &\,\,\text{if}\,\, r = 1, |\PREF^2_{F, i}| = 3, |\gamma(s_r)| \geq 2, g_1\bar{g_2} \not\in \PREF^2_{F, i},\\
1g_2g_3 g_4 \ldots g_l &\,\,\text{if}\,\, r = 1, |\PREF^2_{F, i}| = 3, |\gamma(s_r)| \geq 2, g_1\bar{g_2} \in \PREF^2_{F, i},\\
00g_3 g_4 \ldots g_l &\,\,\text{if}\,\, r \geq 2\\
\end{cases}
\end{equation}
for $r = 1, 2, \ldots, \rho$, where $\gamma(s_1)\gamma(s_2)\ldots \gamma(s_{\rho})$ is the $\gamma$-decomposition of $f_i(s)$ and $\gamma(s_r) = g_1g_2\ldots g_l$.
\end{definition}

\begin{example}
We consider $F(f, \trans) \coloneqq F^{(\eta)}$ in Table \ref{tab:code-tuple}.

\begin{itemize}
\item The $\gamma$-decomposition of $f_0(\mathrm{d})$ is $ f_0(\mathrm{d}) = \gamma(s_1)$,
where $\gamma(s_1) = 001$.
We have $\ddot{\gamma}(s_1) = \gamma(s_1) = 001$ applying the first case of (\ref{eq:oezl1k2insm5}) since $|\PREF^2_{F, 0}| = 4$.
Hence, we have $\ddot{f}_0(\mathrm{d}) = \ddot{\gamma}(s_1) = 001$.

\item The $\gamma$-decomposition of $f_1(\mathrm{c})$ is $f_1(\mathrm{c}) = \gamma(s_1)\gamma(s_2)$, where $\gamma(s_1) = 01$ and $\gamma(s_2) = 001$.
We have $\ddot{\gamma}(s_1) = 01$ applying the third case of (\ref{eq:oezl1k2insm5}) since $|\PREF^2_{F, 1}| = 3$ and $00 \not\in \PREF^2_{F, 1}$.
We have $\ddot{\gamma}(s_2) = 001$ applying the fifth case of (\ref{eq:oezl1k2insm5}).
Hence, we have $\ddot{f}_1(\mathrm{c}) = \ddot{\gamma}(s_1)\ddot{\gamma}(s_2) = 01001$.

\item The $\gamma$-decomposition of $f_1(\mathrm{b})$ is $ f_1(\mathrm{b}) = \gamma(s_1)$,
where $\gamma(s_1) = 1$.
We have $\ddot{\gamma}(s_1) = 1$ applying the second case of (\ref{eq:oezl1k2insm5}) since $|\PREF^2_{F, 1}| = 3$ and $|\gamma(s_1)| = 1$.
Hence, we have $\ddot{f}_1(\mathrm{b}) = \ddot{\gamma}(s_1) = 1$.

\item The $\gamma$-decomposition of $f_2(\mathrm{d})$ is $ f_2(\mathrm{d}) = \gamma(s_1)$,
where $\gamma(s_1) = 011$.
We have $\ddot{\gamma}(s_1) = 111$ applying the fourth case of (\ref{eq:oezl1k2insm5}) since $|\PREF^2_{F, 2}| = 3$ and $01 \in \PREF^2_{F, 2}$.
Hence, we have $\ddot{f}_2(\mathrm{d}) = \ddot{\gamma}(s_1) = 111$.
\end{itemize}

The code table $F^{(\theta)}$ in Table \ref{tab:code-tuple} is obtained as $\ddot{F}^{(\eta)}$.
\end{example}

We state some properties of $\ddot{F}$ as the following Lemmas \ref{lem:tilde2-equal} and \ref{lem:tilde2-pref} (cf. Lemmas \ref{lem:tilde-equal} and \ref{lem:tilde-pref}).

\begin{lemma}
\label{lem:tilde2-equal}
For any $F(f, \trans) \in \mathscr{F}_2$ and $i \in [F]$, the following statements (i)--(iii) hold.
\begin{enumerate}[(i)]
\item Let $s \in \mathcal{S}$ and let $\gamma(s_1)\gamma(s_2)\ldots \gamma(s_{\rho})$ be the $\gamma$-decomposition of $f_i(s)$.
Then we have $|\ddot{\gamma}(s_r)| = |\gamma(s_r)|$ for any $r = 1, 2, \ldots, \rho$.
\item For any $s \in \mathcal{S}$, we have $|\ddot{f}_i(s)| = |f_i(s)|$.
\item For any $s, s' \in \mathcal{S}$, we have $f_i(s) \preceq f_i(s') \iff \ddot{f}_i(s) \preceq \ddot{f}_i(s') $.
\end{enumerate}
\end{lemma}

\begin{proof}[Proof of Lemma \ref{lem:tilde2-equal}]
(Proof of (i)): Directly from (\ref{eq:oezl1k2insm5}).

(Proof of (ii)): 
We have
\begin{eqnarray}
|\ddot{f}_i(s)|
=  |\ddot{\gamma}(s_1)| + |\ddot{\gamma}(s_2)| + \cdots + |\ddot{\gamma}(s_{\rho})|
\overset{(\mathrm{A})}{=} |\gamma(s_1)| + |\gamma(s_2)| + \cdots + |\gamma(s_{\rho})|
= |f_i(s)|,
\end{eqnarray}
where (A) follows from (i) of this lemma.

(Proof of (iii)): See Appendix \ref{subsec:proof-tilde2-equal}.
\end{proof}

\begin{lemma}
\label{lem:tilde2-pref}
For any $F \in \mathscr{F}_2$ and $i \in [F]$, the following statements (i) and (ii) hold.
\begin{enumerate}[(i)]
\item 
\begin{equation}
\label{eq:inrzieqeb2kf}
\PREF^2_{\ddot{F}, i} =
\begin{cases}
\{01, 10, 11\} &\,\,\text{if}\,\, |\PREF^2_{F, i}| = 3,\\
\{00, 01, 10, 11\} &\,\,\text{if}\,\, |\PREF^2_{F, i}| = 4.\\
\end{cases}
\end{equation}
\item For any $s \in \mathcal{S}$, we have
\begin{equation}
\label{eq:2oc1eobowcme}
\bar{\PREF}^2_{\ddot{F}, i}(\ddot{f}_i(s)) =
\begin{cases}
\emptyset &\,\,\text{if}\,\, \bar{\PREF}^0_{F, i}(f_i(s)) = \emptyset,\\
\{00\} &\,\,\text{if}\,\, \bar{\PREF}^0_{F, i}(f_i(s)) \neq \emptyset.\\
\end{cases}
\end{equation}
\end{enumerate}
\end{lemma}
See Appendix \ref{subsec:proof-lem:tilde2-pref} for the proof of Lemma \ref{lem:tilde2-pref}.

Using the properties above, we prove the main result of this subsection as the following Lemma \ref{lem:goodsetT3}.

\begin{lemma}
\label{lem:goodsetT3}
$\mathscr{F}_{\opt} \cap \mathscr{F}_3 \neq \emptyset$.
\end{lemma}

\begin{proof}[Proof of Lemma \ref{lem:goodsetT3}]
By Lemma \ref{lem:goodsetT2}, there exists $F(f, \trans) \in \mathscr{F}_2 \cap \mathscr{F}_{\opt}$.
We have 
\begin{equation}
\label{eq:4dj9z757iacx}
Q(\ddot{F}) = Q(F)
\end{equation}
by (\ref{eq:hlr7zybfne7c}) (cf. Remark \ref{rem:transitionprobability}).

Now, we show $\ddot{F} \in \mathscr{F}_{\opt} \cap \mathscr{F}_3$ as follows.

\begin{itemize}
\item (Proof of $\ddot{F} \in \mathscr{F}_{\reg}$): From $F \in \mathscr{F}_2 \subseteq \mathscr{F}_{\reg}$ and  (\ref{eq:4dj9z757iacx}).

\item (Proof of $\ddot{F} \in \mathscr{F}_{2\hdec}$):
We first show that $\ddot{F}$ satisfies Definition \ref{def:k-bitdelay} (i).
We choose $i \in [\ddot{F}]$ and $s \in \mathcal{S}$ arbitrarily and consider the following two cases separately: the case $\bar{\PREF}^0_{F, i}(f_i(s)) = \emptyset$ and the case $\bar{\PREF}^0_{F, i}(f_i(s)) \neq \emptyset$.
\begin{itemize}
\item The $\bar{\PREF}^0_{F, i}(f_i(s)) = \emptyset$:
We have
\begin{equation}
\PREF^2_{\ddot{F}, \ddot{\trans}_i(s)} \cap \bar{\PREF}^2_{\ddot{F}, i}(\ddot{f}_i(s))
\overset{(\mathrm{A})}{=} \PREF^2_{\ddot{F}, \ddot{\trans}_i(s)} \cap \emptyset
= \emptyset,
\end{equation}
where (A) follows from $\bar{\PREF}^0_{F, i}(f_i(s)) = \emptyset$ and the first case of (\ref{eq:2oc1eobowcme}).

\item The case $\bar{\PREF}^0_{F, i}(f_i(s)) \neq \emptyset$:
By Lemma \ref{cor:pref-sum}, we have $|\PREF^2_{F, \trans_i(s)}| \leq 3$.
In particular, it holds that
\begin{equation}
\label{eq:ic6rt2pqtzl4}
|\PREF^2_{F, \trans_i(s)}| = 3
\end{equation}
by $F \in \mathscr{F}_2$.
Thus, we have
\begin{equation}
\PREF^2_{\ddot{F}, \ddot{\trans}_i(s)} \cap \bar{\PREF}^2_{\ddot{F}, i}(\ddot{f}_i(s))
\overset{(\mathrm{A})}{=} \{01, 10, 11\} \cap \bar{\PREF}^2_{\ddot{F}, i}(\ddot{f}_i(s)) \overset{(\mathrm{B})}{=} \{01, 10, 11\} \cap \{00\} = \emptyset,
\end{equation}
where
(A) follows from (\ref{eq:ic6rt2pqtzl4}) and the first case of (\ref{eq:inrzieqeb2kf}),
and (B) follows from $\bar{\PREF}^0_{F, i}(f_i(s)) \neq \emptyset$ and the second case of (\ref{eq:2oc1eobowcme}).
\end{itemize}
These cases show that $\ddot{F}$ satisfies Definition \ref{def:k-bitdelay} (i).

Also, by $F \in \mathscr{F}_2$ and Lemma \ref{lem:propertyT2},
all the mappings $f_0, f_1, \ldots, f_{|F|-1}$ are injective.
This proves that $\ddot{F}$ satisfies Definition \ref{def:k-bitdelay} (ii) (cf. Remark \ref{rem:injective}).

\item (Proof of $\ddot{F} \in \mathscr{F}_{\opt}$): For any $i \in [F]$, we have $L_i(\ddot{F}) = L_i(F)$ by Lemma \ref{lem:tilde2-equal} (ii) and we have $\pi_i(\ddot{F}) = \pi_i(F)$ by (\ref{eq:4dj9z757iacx}) (cf. Remark \ref{rem:transitionprobability}).
Hence, we have $L(\ddot{F}) = L(F)$, which leads to $\ddot{F} \in \mathscr{F}_{\opt}$ by $F \in \mathscr{F}_{\opt}$.

\item (Proof of ${}^{\forall} i \in [\ddot{F}]; \PREF^2_{\ddot{F}, i} \supseteq \{01, 10, 11\}$):
Choose $i \in [\ddot{F}]$ arbitrarily.
Since $|\PREF^2_{F, i}| \geq 3$ by $F \in \mathscr{F}_2$, we obtain $\PREF^2_{\ddot{F}, i} \supseteq \{01, 10, 11\}$ applying Lemma \ref{lem:tilde2-pref} (i).
\end{itemize}
\end{proof}

\subsection{The Class $\mathscr{F}_4$}
\label{subsec:T4}

In this subsection, we show $\mathscr{F}_{\opt} \cap \mathscr{F}_4 \neq \emptyset$ using the following Lemma \ref{lem:differ} obtained by \cite[Theorem 1]{IEICE2023} with $k = 2$.
See \cite{IEICE2023} for the original statement and the proof.

\begin{lemma}
\label{lem:differ}
For any $F \in \mathscr{F}_0$, there exists $F^{\dagger} \in \mathscr{F}_0$ satisfying the following conditions (a)--(c),
where $\prefset^2_F \coloneqq \{\PREF^2_{F, i} : i \in [F]\}$ for $F \in \mathscr{F}$.
\begin{enumerate}[(a)]
\item $L(F^{\dagger}) \leq L(F)$.
\item $\prefset^2_{F^{\dagger}} \subseteq \prefset^2_F$.
\item $|\prefset^2_{F^{\dagger}}| = |F^{\dagger}|$.
\end{enumerate}
\end{lemma}

Now, we prove the following desired Lemma \ref{lem:goodsetT4}.
\begin{lemma}
  \label{lem:goodsetT4}
$\mathscr{F}_{\opt} \cap \mathscr{F}_4 \neq \emptyset $.
\end{lemma}

\begin{proof}[Proof of Lemma \ref{lem:goodsetT4}]
By Lemma \ref{lem:goodsetT3}, there exists $F \in \mathscr{F}_{\opt} \cap \mathscr{F}_3$.
Applying Lemma \ref{lem:differ}, there exists $F^{\dagger}(f^{\dagger}, \trans^{\dagger}) \in \mathscr{F}_{\opt} \cap \mathscr{F}_3$ satisfying $|F^{\dagger}| = |\prefset^2_{F^{\dagger}}|$.
By Lemma \ref{lem:complete-leaf2}, there exists $i \in \kernel_{F^{\dagger}}$ such that $\PREF^2_{F^{\dagger}, i} = \{00, 01, 10, 11\}$.
Hence, $F^{\dagger}$ satisfies exactly one of the following conditions (a) and (b).
\begin{enumerate}[(a)]
\item $|F^{\dagger}| = 2, \PREF^2_{F^{\dagger}, 0} = \{00, 01, 10, 11\}, \PREF^2_{F^{\dagger}, 1} = \{01, 10, 11\}$ (by swapping the indices of $(f^{\dagger}_0, \trans^{\dagger}_0)$ and $(f^{\dagger}_1, \trans^{\dagger}_1)$ if necessary).
\item $|F^{\dagger}| = 1, \PREF^2_{F^{\dagger}, 0} = \{00, 01, 10, 11\}$.
\end{enumerate}
In the case (a), we have $F^{\dagger} \in \mathscr{F}_{\opt} \cap \mathscr{F}_4$ as desired.
In the case (b), we can see that the code-tuple $F'(f', \trans') \in \mathscr{F}^{(2)}$ defined as below satisfies $F' \in \mathscr{F}_{\opt} \cap \mathscr{F}_4$ as desired:
\begin{eqnarray}
f'_0(s_r) \coloneqq f^{\dagger}_0(s_r), \quad \trans'_0(s_r) \coloneqq \trans^{\dagger}_0(s_r),
\end{eqnarray}
\begin{equation}
f'_1(s_r) = \begin{cases}
01 &\textrm{if}\,\, r = 1,\\
1^{r-1}0 &\textrm{if}\,\, 2 \leq r \leq \sigma-1,\\
1^{\sigma-1} &\textrm{if}\,\, r = \sigma,
\end{cases}
\quad \trans'_1(s_r) = 0
\end{equation}
for $s_r \in \mathcal{S}$, where we suppose $\mathcal{S} = \{s_1, s_2, \ldots, s_{\sigma}\}$ and the notation $1^l$ denotes the sequence obtained by concatenating $l$ copies of $1$ for an integer $l \geq 1$.
\end{proof}

\subsection{Proof of $\mathscr{F}_{\opt} \cap \mathscr{F}_{\AIFV} \neq \emptyset $}
\label{subsec:optimality}

Finally, we prove the following Theorem \ref{thm:optimality} as the main result of this paper.

\setcounter{theorem}{0}
\begin{theorem}
\label{thm:optimality}
$\mathscr{F}_{\opt} \cap \mathscr{F}_{\AIFV} \neq \emptyset$.
\end{theorem}

\begin{proof}[Proof of Theorem \ref{thm:optimality}]
By Lemma \ref{lem:goodsetT4}, there exists $F \in \mathscr{F}_{\opt} \cap \mathscr{F}_4$.
We have $0 \in \kernel_{F}$ by Lemma \ref{lem:complete-leaf2}.
We consider the following two cases separately: the case $\kernel_{F} = \{0, 1\}$ and the case $\kernel_{F} = \{0\}$.
\begin{itemize}

\item The case $\kernel_{F} = \{0, 1\}$:
We prove $F \in \mathscr{F}_{\AIFV}$ by showing that $F$ satisfies Definition \ref{def:subsetTaifv} (i)--(vii).
\begin{itemize}
\item (Proof of (i)): Directly from Lemma \ref{lem:propertyT2}.

\item (Proof of (ii)): Choose $s \in \mathcal{S}$ arbitrarily. 
We first prove $\bar{\PREF}^1_{F, i}(f_i(s)) \not\owns 1$ by contradiction assuming $\bar{\PREF}^1_{F, i}(f_i(s)) \owns 1$.
Then by Lemma \ref{lem:pref-inc} (ii), we have
\begin{equation}
\label{eq:aj6ct0ru9hkt}
\bar{\PREF}^2_{F, i}(f_i(s)) \owns 1c
\end{equation}
for some $c \in \mathcal{C}$.
On the other hand, by $F \in \mathscr{F}_4$, we have
\begin{equation}
\label{eq:z3k5ek3ib02v}
\PREF^2_{F, \trans_i(s)} \owns 10, 11.
\end{equation}
By (\ref{eq:aj6ct0ru9hkt}) and (\ref{eq:z3k5ek3ib02v}), we obtain
$\PREF^2_{F, \trans_i(s)} \cap \bar{\PREF}^2_{F, i}(f_i(s))  \neq \emptyset$,
which leads to $F \not \in \mathscr{F}_{2\hdec}$.
This conflicts with $F \in \mathscr{F}_4 \subseteq \mathscr{F}_{2\hdec}$.

Next, we prove $\bar{\PREF}^1_{F, i}(f_i(s)0) \not\owns 1$ by contradiction assuming
\begin{equation}
\label{eq:fuz0ot0x7u6z}
\bar{\PREF}^1_{F, i}(f_i(s)0) \owns 1.
\end{equation} 
Then we have
\begin{eqnarray}
\PREF^2_{F, \trans_i(s)} \cap \bar{\PREF}^2_{F, i}(f_i(s))
&\overset{(\mathrm{A})}{\supseteq}& \PREF^2_{F, \trans_i(s)} \cap  0\PREF^1_{F, i}(f_i(s)0)\\
&\overset{(\mathrm{B})}{\supseteq}& \PREF^2_{F, \trans_i(s)} \cap  0\bar{\PREF}^1_{F, i}(f_i(s)0)\\
&\overset{(\mathrm{C})}{\supseteq}& \PREF^2_{F, \trans_i(s)} \cap  0\{1\},\\
&\overset{(\mathrm{D})}{\supseteq}& \{01, 10, 11\} \cap \{01\}\\
&=& \{01\}\\
&\neq& \emptyset
\end{eqnarray}
where
(A) follows from Lemma \ref{lem:pref-sum} (iii),
(B) follows from Lemma \ref{lem:pref-sum} (i),
(C) follows from (\ref{eq:fuz0ot0x7u6z}),
and (D) follows from $F \in \mathscr{F}_4 \subseteq \mathscr{F}_3$.
Hence, we obtain $F \not\in \mathscr{F}_{2\hdec}$, which conflicts with $F \in \mathscr{F}_4 \subseteq \mathscr{F}_{2\hdec}$.

\item (Proof of (iii)): Directly from Lemma \ref{lem:propertyT1} (v).

\item (Proof of (iv)): Choose $i \in [F]$ and $s \in \mathcal{S}$ arbitrarily and consider the following two cases separately: the case $\bar{\PREF}^0_{F, i}(f_i(s)) = \emptyset$ and the case $\bar{\PREF}^0_{F, i}(f_i(s)) \neq \emptyset$:

\begin{itemize}
\item The case $\bar{\PREF}^0_{F, i}(f_i(s)) = \emptyset$:
We have $|\PREF^2_{F, \trans_i(s)}| = 4$ applying Lemma \ref{lem:complete-leaf} since $i \in \{0, 1\} = \kernel_{F}$ holds and $f_i$ is injective by Lemma \ref{lem:propertyT2}.
Hence, we obtain $\trans_i(s) = 0$ by $F \in \mathscr{F}_4$.

\item The case $\bar{\PREF}^0_{F, i}(f_i(s)) \neq \emptyset$:
We have $|\PREF^2_{F, \trans_i(s)}| \leq 3$ by Lemma \ref{cor:pref-sum}, 
Hence, we obtain $\trans_i(s) = 1$ by $F \in \mathscr{F}_4$.
\end{itemize}

\item (Proof of (v)):
We choose $i \in [F]$ arbitrarily and prove that if $f_i(s) = \lambda$ or $f_i(s) = 0$ for some $s \in \mathcal{S}$, then $\PREF^2_{F, i} \neq \{01, 10, 11\}$, which is equivalent to $i = 0$.
Choose $s \in \mathcal{S}$ such that $f_i(s) = \lambda$ or $f_i(s) = 0$.
We consider the following two cases separately: the case $f_i(s) = \lambda$ and the case $f_i(s) = 0$.
\begin{itemize}
\item The case $f_i(s) = \lambda$:
By Lemma \ref{lem:propertyT2}, the mapping $f_i$ is injective.
Thus, by Lemma \ref{lem:leaf} (iii), we have $\bar{\PREF}^0_{F, i} \neq \emptyset$.
Hence, by Lemma \ref{lem:pref-inc2} (ii) (a), we have
 \begin{equation}
\label{eq:5n1gycwcwfy6}
 \bar{\PREF}^2_{F, i} \neq \emptyset.
 \end{equation}
Also, we have
\begin{equation}
\label{eq:es0lyeh24u7q}
\bar{\PREF}^2_{F, i} \overset{(\mathrm{A})}{\subseteq} \mathcal{C}^2 \setminus \PREF^2_{F, i} \overset{(\mathrm{B})}{\subseteq} \mathcal{C}^2 \setminus \{01, 10, 11\} = \{00\},
\end{equation}
where
(A) follows from $F \in \mathscr{F}_4 \subseteq \mathscr{F}_{2\hdec}$,
and (B) follows from $F \in \mathscr{F}_4 \subseteq \mathscr{F}_3$.
Thus, we obtain
\begin{equation}
\PREF^2_{F, i} \overset{(\mathrm{A})}{\supseteq} \bar{\PREF}^2_{F, i} \overset{(\mathrm{B})}{=} \{00\}.
\end{equation}
where
(A) follows from Lemma \ref{lem:pref-sum} (i),
and (B) follows from (\ref{eq:5n1gycwcwfy6}) and (\ref{eq:es0lyeh24u7q}).
This shows $\PREF^2_{F, i} \neq \{01, 10, 11\}$ as desired.

\item The case $f_i(s) = 0$: 
We have
\begin{equation}
\PREF^2_{F, i} \overset{(\mathrm{A})}{\supseteq} \bar{\PREF}^2_{F, i} \overset{(\mathrm{B})}{\supseteq} 0\PREF^1_{F, i}(0) \overset{(\mathrm{C})}{=} 0\PREF^1_{F, i}(f_i(s))
\overset{(\mathrm{D})}{\supseteq} 0\PREF^1_{F, \trans_i(s)} \overset{(\mathrm{E})}{=} 0\{0, 1\} \owns 00,
\end{equation}
where
(A) follows from Lemma \ref{lem:pref-sum} (i),
(B) follows from Lemma \ref{lem:pref-sum} (iii),
(C) follows from $f_i(s) = 0$,
(D) follows from Lemma \ref{lem:pref-sum} (i),
and (E) follows from $F \in \mathscr{F}_4 \subseteq \mathscr{F}_1$.
This leads to $\PREF^2_{F, i} \neq \{01, 10, 11\}$.
\end{itemize}

\item (Proof of (vi)): We prove by contradiction assuming $\bar{\PREF}^1_{F, 1}(0) \owns 0$.
We have
\begin{equation}
\PREF^2_{F, 1} \overset{(\mathrm{A})}{\supseteq} \bar{\PREF}^2_{F, 1}
\overset{(\mathrm{B})}{\supseteq} 0\PREF^1_{F, 1}(0)
\overset{(\mathrm{C})}{\supseteq} 0\bar{\PREF}^1_{F, 1}(0)
\overset{(\mathrm{D})}{\owns} 00,
\end{equation}
where
(A) follows from Lemma \ref{lem:pref-sum} (i),
(B) follows from Lemma \ref{lem:pref-sum} (iii),
(C) follows from Lemma \ref{lem:pref-sum} (i),
and (D) follows from $\bar{\PREF}^1_{F, 1}(0) \owns 0$.
This shows $\PREF^2_{F, 1} \neq \{01, 10, 11\}$, which conflicts with $F \in \mathscr{F}_4$.

\item (Proof of (vii)):
We prove by contradiction assuming that there exist $i \in [F]$ and $\pmb{b} \in \mathcal{C}^{\ast}$ such that all of the following conditions (a)--(c) hold.
\begin{enumerate}[(a)]
\item $|\bar{\PREF}^1_{F, i}(\pmb{b})| = 1$.
\item $f_i(s)\pmb{c} \neq \pmb{b}$ for any $s \in \mathcal{S}$ and $\pmb{c} \in \mathcal{C}^0 \cup \mathcal{C}^1$.
\item $(i, \pmb{b}) \neq (1, 0)$.
\end{enumerate}

We have
\begin{eqnarray}
\label{eq:d5xf7ohenrre}
|\PREF^1_{F, i}(\pmb{b})| \overset{(\mathrm{A})}{=} |\bar{\PREF}^1_{F, i}(\pmb{b})| + \sum_{s \in \assign_{F, i}(\pmb{b})} |\PREF^1_{F, \trans_i(s)}|
\overset{(\mathrm{B})}{=} |\bar{\PREF}^1_{F, i}(\pmb{b})|
\overset{(\mathrm{C})}{=} 1,
\end{eqnarray}
where
(A) follows from Lemma \ref{lem:pref-sum} (ii),
(B) follows since $\assign_{F, i}(\pmb{b}) = \emptyset$ by the condition (b),
and (C) follows from the condition (a).

We consider the following three cases separately: the case $|\pmb{b}| = 0$, the case $|\pmb{b}| = 1$, and the case $|\pmb{b}| \geq 2$.
\begin{itemize}
\item The case $|\pmb{b}| = 0$: 
By (\ref{eq:d5xf7ohenrre}), we have $|\PREF^1_{F, i}|  = |\PREF^1_{F, i}(\pmb{b})| = 1$, which conflicts with $F \in \mathscr{F}_4 \subseteq \mathscr{F}_1$.

\item The case $|\pmb{b}| = 1$: We have
\begin{equation}
\label{eq:drhmpfslcqv1}
\PREF^2_{F, i} \overset{(\mathrm{A})}{=} \bar{\PREF}^2_{F, i} \cup \Big( \bigcup_{s \in \assign_{F, i}(\lambda)} \PREF^2_{F, \trans_i(s)} \Big)
\overset{(\mathrm{B})}{=} \bar{\PREF}^2_{F, i} \overset{(\mathrm{C})}{=} 0\PREF^1_{F, i}(0) \cup 1\PREF^1_{F, i}(1),
\end{equation}
where
(A) follows from Lemma \ref{lem:pref-sum} (i),
(B) follows because $\assign_{F, i}(\lambda) = \emptyset$ by $|\pmb{b}| = 1$ and the condition (b),
and (C) follows from Lemma \ref{lem:pref-sum} (iii).

On the other hand, we have $\PREF^2_{F, 0} = \{00, 01, 10, 11\}$ and $\PREF^2_{F, 1} = \{01, 10, 11\}$ by $F \in \mathscr{F}_4$.
Hence, comparing with (\ref{eq:drhmpfslcqv1}), we have $\PREF^1_{F, 0}(0) = \PREF^1_{F, 0}(1) = \PREF^1_{F, 1}(1) = \{0, 1\}$ and $\PREF^1_{F, 1}(0) = \{1\}$.
Therefore, by (\ref{eq:d5xf7ohenrre}) and $|\pmb{b}| = 1$, it must hold that $(i, \pmb{b}) = (1, 0)$, which conflicts with the condition (c).

\item The case $|\pmb{b}| \geq 2$: 
By the condition (a), we have
\begin{equation}
\label{eq:z3nwkrpl9ris}
\bar{\PREF}^1_{F, i}(\pmb{b}) = \{a\}
\end{equation}
for some $a \in \mathcal{C}$.
Then there exists $\pmb{x} = x_1x_2\ldots x_n \in \mathcal{S}^{+}$ such that
\begin{equation}
\label{eq:g13s78d26rg2}
f^{\ast}_i(\pmb{x}) \succeq \pmb{b}a, \quad f_i(x_1) \succ \pmb{b}.
\end{equation}
Hence, by $|\pmb{b}| \geq 2$, we have $f_i(x_1) \succ b_1b_2$, which leads to
\begin{equation}
\label{eq:b0896wplc0bb}
b_1b_2 \in \PREF^2_{F, i},
\end{equation}
where $b_1b_2$ denotes the prefix of length $2$ of $\pmb{b}$.
By $i \in \{0, 1\} = \kernel_{F}$ and (\ref{eq:b0896wplc0bb}), we have $\pmb{b}\bar{a} \in \PREF^{\ast}_{F, i}$ applying Lemma \ref{lem:complete} stated in Appendix \ref{subsec:proof-complete-leaf}.
Hence, there exists $\pmb{y} = y_1y_2\ldots y_{n'} \in \mathcal{S}^{+}$ such that
\begin{equation}
\label{eq:mbq9u28kebxt}
f^{\ast}_i(\pmb{y}) \succeq \pmb{b}\bar{a}.
\end{equation}
Then exactly one of $f_i(y_1) \succ \pmb{b}$ and $f_i(y_1) \preceq \pmb{b}$ holds.
Now, the latter $f_i(y_1) \preceq \pmb{b}$ holds because the former $f_i(y_1) \succ \pmb{b}$ implies $\bar{a} \in \bar{\PREF}^1_{F, i}(\pmb{b})$ by (\ref{eq:mbq9u28kebxt}), which conflicts with (\ref{eq:z3nwkrpl9ris}).
Therefore, there exists $\pmb{c} = c_1c_2\ldots c_l \in \mathcal{C}^{\ast}$ such that $f_i(y_1)\pmb{c} = \pmb{b}$.
By the condition (b), we have $|\pmb{c}| \geq 2$ so that
\begin{equation}
\label{eq:ret8gt7iszol}
f_i(y_1)c_1c_2 \preceq \pmb{b}.
\end{equation}
We have
\begin{equation}
f_i(y_1)f^{\ast}_{\trans_i(y_1)}(\suff(\pmb{y}))
= f^{\ast}_i(\pmb{y})
\overset{(\mathrm{A})}{\succeq} \pmb{b}\bar{a}
\succeq \pmb{b}
\overset{(\mathrm{B})}{\succeq} f_i(y_1)c_1c_2,
\end{equation}
where
(A) follows from (\ref{eq:mbq9u28kebxt}),
and (B) follows from (\ref{eq:ret8gt7iszol}).
Comparing both sides, we obtain $f^{\ast}_{\trans_i(y_1)}(\suff(\pmb{y})) \succeq c_1c_2$,
which leads to
\begin{equation}
\label{eq:q1zyk9xz64en}
c_1c_2 \in \PREF^2_{F, \trans_i(y_1)}.
\end{equation}
Also, by (\ref{eq:g13s78d26rg2}) and (\ref{eq:ret8gt7iszol}), we have $f_i(x_1) \succ f_i(y_1)c_1c_2$, which leads to
\begin{equation}
\label{eq:sonmfpcgjp39}
c_1c_2 \in \bar{\PREF}^2_{F, i}(f_i(y_1)).
\end{equation}
By (\ref{eq:q1zyk9xz64en}) and (\ref{eq:sonmfpcgjp39}), we obtain $\bar{\PREF}^2_{F, i}(f_i(y_1)) \cap \PREF^2_{F, \trans_i(y_1)} \neq \emptyset$, which conflicts with $F \in \mathscr{F}_{2\hdec}$.
\end{itemize}
\end{itemize}

\item The case $\kernel_{F} = \{0\}$:
We define $F'(f', \trans') \in \mathscr{F}^{(2)}$ as
\begin{eqnarray}
f'_0(s_r) \coloneqq f_0(s_r), \quad \trans'_0(s_r) \coloneqq \trans_0(s_r),
\end{eqnarray}
\begin{equation}
f'_1(s_r) = \begin{cases}
01 &\textrm{if}\,\, r = 1,\\
1^{r-1}0 &\textrm{if}\,\, 2 \leq r \leq \sigma-1,\\
1^{\sigma-1} &\textrm{if}\,\, r = \sigma,
\end{cases}
\quad \trans'_1(s_r) = 0
\end{equation}
for $s_r \in \mathcal{S}$, where we suppose $\mathcal{S} = \{s_1, s_2, \ldots, s_{\sigma}\}$ and the notation $1^l$ denotes the sequence obtained by concatenating $l$ copies of $1$ for an integer $l \geq 1$.
We can show that $F'$ satisfies Definition \ref{def:subsetTaifv} (i)--(vii) in a similar way to the case $\kernel_{F} = \{0, 1\}$.
\end{itemize}
\end{proof}

\section{conclusion}
\label{sec:conclusion}

We proved the optimality of binary AIFV codes in the class of $2$-bit delay decodable codes with a finite number of code tables.
First, we introduced a code-tuple as a model of a time-variant encoder with a finite number of code tables.
Next, we defined the class $\mathscr{F}_{k\hdec}$ (resp. $\mathscr{F}_{\ext}$,  $\mathscr{F}_{\reg}$) of $k$-bit delay decodable (resp. extendable, regular) code-tuples.
Then we proved Theorem \ref{thm:optimality} that the class of AIFV codes $\mathscr{F}_{\AIFV}$ achieves the optimal average codeword length in $\mathscr{F}_0 = \mathscr{F}_{\reg} \cap \mathscr{F}_{\ext} \cap \mathscr{F}_{2\hdec}$
by introducing the classes $\mathscr{F}_1, \mathscr{F}_2, \mathscr{F}_3, \mathscr{F}_4$ and showing $\mathscr{F}_{\opt} \cap \mathscr{F}_i \neq \emptyset$ sequentially for $i = 1, 2, 3, 4$ and finally $\mathscr{F}_{\opt} \cap \mathscr{F}_{\AIFV} \neq \emptyset$.

\appendix

\subsection{Proof of Lemma \ref{lem:aifv-decodable}}
\label{subsec:proof-aifv-decodable}

To prove Lemma \ref{lem:aifv-decodable}, we first show the following Lemma \ref{lem:aifv-decodable0}.

\begin{lemma}
\label{lem:aifv-decodable0}
For any $F \in \mathscr{F}_{\AIFV}$, the following conditions (i)--(iii) hold.
\begin{enumerate}[(i)]
\item $\PREF^1_{F, 0} = \PREF^1_{F, 1} = \{0, 1\}$.
\item For any $i \in [F]$ and $b \in \mathcal{C}$, if $\assign_{F, i}(\lambda) = \emptyset$ and $(i, b) \neq (1, 0)$, then $\PREF^1_{F, i}(b) = \{0, 1\}$.
\item For any $i \in [F]$ and $s \in \mathcal{S}$, if $\bar{\PREF}^0_{F, i}(f_i(s)) \neq \emptyset$, then $\bar{\PREF}^2_{F, i}(f_i(s)) = \{00\}$.
\end{enumerate}
\end{lemma}

\begin{proof}[Proof of Lemma \ref{lem:aifv-decodable0}]
(Proof of (i)):
We first show
\begin{equation}
\label{eq:2welgcl2lq9m}
\PREF^1_{F, 1} = \{0, 1\}.
\end{equation}
To prove it, it suffices to show $|\bar{\PREF}^1_{F, 1}| = 2$ because this implies 
$\PREF^1_{F, 1} \supseteq \bar{\PREF}^1_{F, 1} = \{0, 1\}$ by Lemma \ref{lem:pref-sum} (i).
\begin{itemize}
\item We obtain $|\bar{\PREF}^1_{F, 1}| \neq 0$ by applying Lemma \ref{lem:pref-inc2} (ii) (a) because $|\bar{\PREF}^0_{F, 1}| \neq 0$ by Definition \ref{def:subsetTaifv} (i) and Lemma \ref{lem:leaf} (iii).
\item Also, we have $|\bar{\PREF}^1_{F, 1}| \neq 1$ because neither the condition (a) nor (b) of Definition \ref{def:subsetTaifv} (vii) holds for $(i, \pmb{b}) = (1, \lambda)$ by Definition \ref{def:subsetTaifv} (v).
\end{itemize}
These show (\ref{eq:2welgcl2lq9m}).

Next, we show $\PREF^1_{F, 0} = \{0, 1\}$ by considering the following two cases separately: the case $\assign_{F, 0}(\lambda) = \emptyset$ and the case $\assign_{F, 0}(\lambda) \neq \emptyset$.
\begin{itemize}
\item The case $\assign_{F, 0}(\lambda) = \emptyset$: By a similar argument to derive (\ref{eq:2welgcl2lq9m}).

\item The case $\assign_{F, 0}(\lambda) \neq \emptyset$:
We have
\begin{equation}
\PREF^1_{F, 0}
\overset{(\mathrm{A})}{\supseteq} \bigcup_{s \in \assign_{F, 0}(\lambda)} \PREF^1_{F, \trans_0(s)}
\overset{(\mathrm{B})}{=} \bigcup_{s \in \assign_{F, 0}(\lambda)} \PREF^1_{F, 1}
\overset{(\mathrm{C})}{=} \bigcup_{s \in \assign_{F, 0}(\lambda)} \{0, 1\}
\overset{(\mathrm{D})}{=} \{0, 1\},
\end{equation}
where (A) follows from Lemma \ref{lem:pref-sum} (i),
(B) follows from Definition \ref{def:subsetTaifv} (iv) because $\bar{\PREF}^0_{F, 0}(f_0(s)) = \bar{\PREF}^0_{F, 0} \neq \emptyset$ by Definition \ref{def:subsetTaifv} (i) and Lemma \ref{lem:leaf} (iii),
(C) follows from (\ref{eq:2welgcl2lq9m}),
and (D) follows from $\assign_{F, 0}(\lambda) \neq \emptyset$.
\end{itemize}

(Proof of (ii)):
Assume $\assign_{F, i}(\lambda) = \emptyset$ and $(i, b) \neq (1, 0)$.
We consider the following two cases separately: the case $\assign_{F, i}(b) = \emptyset$ and the case $\assign_{F, i}(b) \neq \emptyset$.

\begin{itemize}
\item The case $\assign_{F, i}(b) = \emptyset$:
It suffices to show $|\bar{\PREF}^1_{F, 1}(b)| = 2$ because this implies $\PREF^1_{F, i}(b) \supseteq \bar{\PREF}^1_{F, i}(b) = \{0, 1\}$ by Lemma \ref{lem:pref-sum} (i).
\begin{itemize}
\item We have $b \in \{0, 1\} = \PREF^1_{F, i}$ by (i) of this lemma.
Hence, there exists $\pmb{x} = x_1x_2\ldots x_n \in \mathcal{S}^{+}$ such that $f^{\ast}_i(\pmb{x}) \succeq b$.
Since $\assign_{F, i}(\lambda) = \assign_{F, i}(b) = \emptyset$, we have $f_i(x_1) \succ b$ and thus $|\bar{\PREF}^1_{F, i}(b)| \neq 0$.
\item Also, by Definition \ref{def:subsetTaifv} (vii), it must hold that $|\bar{\PREF}^1_{F, i}(b)| \neq 1$ since $\assign_{F, i}(\lambda) = \assign_{F, i}(b) = \emptyset$ and $(i, b) \neq (1, 0)$.
\end{itemize}
These show $\PREF^1_{F, i}(b) = \{0, 1\}$ as desired.

\item The case $\assign_{F, i}(b) \neq \emptyset$:
We have
\begin{equation}
\PREF^1_{F, i}(b)
\overset{(\mathrm{A})}{\supseteq} \bigcup_{s \in \assign_{F, i}(b)} \PREF^1_{F, \trans_i(s)}
\overset{(\mathrm{B})}{=} \bigcup_{s \in \assign_{F, i}(b)} \{0, 1\}
\overset{(\mathrm{C})}{=} \{0, 1\}
\end{equation}
as desired, where (A) follows from Lemma \ref{lem:pref-sum} (i),
(B) follows from (i) of this lemma,
and (C) follows from $\assign_{F, i}(b) \neq \emptyset$.
\end{itemize}

(Proof of (iii)):
Assume $\bar{\PREF}^0_{F, i}(f_i(s)) \neq \emptyset$.
Then we have $\bar{\PREF}^1_{F, i}(f_i(s)) \neq \emptyset$ by Lemma \ref{lem:pref-inc2} (ii) (a).
Since $1 \not\in \bar{\PREF}^1_{F, i}(f_i(s))$ by Definition \ref{def:subsetTaifv} (ii),
it must hold that
\begin{equation}
\label{eq:gjmq0830tudo}
\bar{\PREF}^1_{F, i}(f_i(s)) = \{0\}.
\end{equation}
We have
\begin{equation}
\label{eq:8ly8azq6ao2o}
0\PREF^1_{F, i}(f_i(s)0) \cup 1\PREF^1_{F, i}(f_i(s)1)\\
\overset{(\mathrm{A})}{=} \bar{\PREF}^2_{F, i}(f_i(s)) \overset{(\mathrm{B})}{\subseteq} \{00, 01\},
\end{equation}
where
(A) follows from Lemma \ref{lem:pref-sum} (iii),
and (B) follows from by (\ref{eq:gjmq0830tudo}) and Lemma \ref{lem:pref-inc} (ii).
Comparing both sides of (\ref{eq:8ly8azq6ao2o}), we have
\begin{equation}
\label{eq:fqn4tmw6e84z}
1\PREF^1_{F, i}(f_i(s)1) = \emptyset.
\end{equation}
Thus, we obtain
\begin{eqnarray}
\bar{\PREF}^2_{F, i}(f_i(s))
&\overset{(\mathrm{A})}{=}& 0\PREF^1_{F, i}(f_i(s)0) \cup 1\PREF^1_{F, i}(f_i(s)1)\\
&\overset{(\mathrm{B})}{=}& 0\PREF^1_{F, i}(f_i(s)0)\\
&\overset{(\mathrm{C})}{=}&  0\Big(\bar{\PREF}^1_{F, i}(f_i(s)0) \cup \Big( \bigcup_{s' \in \assign_{F, i}(f_i(s)0)} \PREF^1_{F, \trans_i(s')} \Big)\Big)\\
&=&  0\bar{\PREF}^1_{F, i}(f_i(s)0) \cup \Big( \bigcup_{s' \in \assign_{F, i}(f_i(s)0)} 0\PREF^1_{F, \trans_i(s')} \Big)\\
&\overset{(\mathrm{D})}{=}& 0\bar{\PREF}^1_{F, i}(f_i(s)0)\\
&\overset{(\mathrm{E})}{=}& \{00\} \label{eq:gaoigeogjiae},
\end{eqnarray}
where
(A) follows from Lemma \ref{lem:pref-sum} (iii),
(B) follows from (\ref{eq:fqn4tmw6e84z}),
(C) follows from Lemma \ref{lem:pref-sum} (i),
(D) follows since $\assign_{F, i}(f_i(s)0) = \emptyset$ by Definition \ref{def:subsetTaifv} (iii),
and (E) follows from (\ref{eq:gjmq0830tudo}).

\end{proof}

\begin{proof}[Proof of Lemma \ref{lem:aifv-decodable}]
We fix $F \in \mathscr{F}_{\AIFV}$ arbitrarily and show $F \in \mathscr{F}_{\reg}$,  $F \in \mathscr{F}_{2\hdec}$, $\PREF^2_{F, 0} = \{00, 01, 10, 11\}$ and $\PREF^2_{F, 1} = \{01, 10, 11\}$.

(Proof of $F \in \mathscr{F}_{\reg}$):
By Lemma \ref{lem:leaf} (ii), the following (\ref{eq:h5nzoz10dp8e}) holds, which implies
\begin{eqnarray}
\lefteqn{{}^{\forall}i \in [F]; {}^{\exists} s \in \mathcal{S}; \bar{\PREF}^0_{F, i}(f_i(s)) = \emptyset} \label{eq:h5nzoz10dp8e}\\
&\overset{(\mathrm{A})}{\implies}& {}^{\forall}i \in [F]; {}^{\exists} s \in \mathcal{S}; \trans_i(s) = 0\\
&\overset{(\mathrm{B})}{\implies}& \kernel_F \owns 0\\
&\overset{(\mathrm{C})}{\implies}& F \in \mathscr{F}_{\reg},
\end{eqnarray}
where
(A) follows from Definition \ref{def:subsetTaifv} (iv),
(B) follows from (\ref{eq:7tuxvj14yeno}),
and (C) follows from Lemma \ref{lem:kernel} (i).

(Proof of $\PREF^2_{F, 1} = \{01, 10, 11\}$):
We have $0 \in \{0, 1\} = \PREF^1_{F, 1}$ by Lemma \ref{lem:aifv-decodable0} (i).
Hence, there exists $\pmb{x} = x_1x_2\ldots x_n \in \mathcal{S}^{+}$ such that $f^{\ast}_1(\pmb{x}) \succeq 0$.
By Definition \ref{def:subsetTaifv} (v), we have $f_1(x_1) \succ 0$ and thus
\begin{equation}
\label{eq:sc1kozwljv2y}
\bar{\PREF}^1_{F, 1}(0) \neq \emptyset.
\end{equation}
Therefore, we obtain
\begin{equation}
\label{eq:37a9r185alx0}
\PREF^1_{F, 1}(0) 
\overset{(\mathrm{A})}{=} \bar{\PREF}^1_{F, 1}(0) \cup \Big( \bigcup_{s' \in \assign_{F, 1}(0)} \PREF^1_{F, \trans_1(s')} \Big)
\overset{(\mathrm{B})}{=} \bar{\PREF}^1_{F, 1}(0)
\overset{(\mathrm{C})}{=} \{1\},
\end{equation}
where
(A) follows from Lemma \ref{lem:pref-sum} (i),
(B) follows since $\assign_{F, 1}(0) = \emptyset$ by Definition \ref{def:subsetTaifv} (v),
and (C) follows from (\ref{eq:sc1kozwljv2y}) and Definition \ref{def:subsetTaifv} (vi).
Thus, we obtain
\begin{eqnarray}
\PREF^2_{F, 1}
&\overset{(\mathrm{A})}{=}&  \bar{\PREF}^2_{F, 1} \cup \Big( \bigcup_{s' \in \assign_{F, 1}(\lambda)} \PREF^2_{F, \trans_1(s')} \Big)\\
&\overset{(\mathrm{B})}{=}&  \bar{\PREF}^2_{F, 1}\\
&\overset{(\mathrm{C})}{=}&  0\PREF^1_{F, 1}(0) \cup 1\PREF^1_{F, 1}(1)\\
&\overset{(\mathrm{D})}{=}&  0\{1\} \cup 1\PREF^1_{F, 1}(1)\\
&\overset{(\mathrm{E})}{=}& 0\{1\} \cup 1\{0, 1\}\\
&=& \{01, 10, 11\} \label{eq:l0ko6bqqdfsf}
\end{eqnarray}
as desired, where
(A) follows from Lemma \ref{lem:pref-sum} (i),
(B) follows since $\assign_{F, 1}(\lambda) = \emptyset$ by Definition \ref{def:subsetTaifv} (v),
(C) follows from Lemma \ref{lem:pref-sum} (iii),
(D) follows from (\ref{eq:37a9r185alx0}),
and (E) follows from Lemma \ref{lem:aifv-decodable0} (ii) since $\assign_{F, 1}(\lambda) = \emptyset$ by Definition \ref{def:subsetTaifv} (v).

(Proof of $\PREF^2_{F, 0} = \{00, 01, 10, 11\}$):
We consider the following two cases separately:
the case $\assign_{F, 0}(\lambda) = \emptyset$ and the case $\assign_{F, 0}(\lambda) \neq \emptyset$.
\begin{itemize}
\item The case $\assign_{F, 0}(\lambda) = \emptyset$:
We have
\begin{equation}
\PREF^2_{F, 0}
\overset{(\mathrm{A})}{\supseteq} \bar{\PREF}^2_{F, 0}
\overset{(\mathrm{B})}{=} 0\PREF^1_{F, 0}(0) \cup 1\PREF^1_{F, 1}(1)
\overset{(\mathrm{C})}{=} 0\{0, 1\} \cup 1\PREF^1_{F, 1}(1)
\overset{(\mathrm{D})}{=} 0\{0, 1\} \cup 1\{0, 1\}
= \{00, 01, 10, 11\}
\end{equation}
as desired, where
(A) follows from Lemma \ref{lem:pref-sum} (i),
(B) follows from Lemma \ref{lem:pref-sum} (iii),
(C) follows from Lemma \ref{lem:aifv-decodable0} (ii) since $\assign_{F, 0}(\lambda) = \emptyset$,
and (D) follows from Lemma \ref{lem:aifv-decodable0} (ii) since $\assign_{F, 1}(\lambda) = \emptyset$ by Definition \ref{def:subsetTaifv} (v).

\item The case $\assign_{F, 0}(\lambda) \neq \emptyset$:
Let $s \in \assign_{F, 0}(\lambda) \neq \emptyset$.
We have
\begin{equation}
\label{eq:e2czit3jvugw}
\bar{\PREF}^0_{F, 0}(f_0(s)) = \bar{\PREF}^0_{F, 0} \neq \emptyset
\end{equation}
by Definition \ref{def:subsetTaifv} (i) and Lemma \ref{lem:leaf} (iii),
and thus we have $\trans_0(s) = 1$ by Definition \ref{def:subsetTaifv} (iv).
Hence, we have
\begin{eqnarray}
\PREF^2_{F, 0}
&\overset{(\mathrm{A})}{=}&  \bar{\PREF}^2_{F, 0} \cup \Big( \bigcup_{s' \in \assign_{F, 0}(\lambda)} \PREF^2_{F, \trans_0(s')} \Big)\\
&\overset{(\mathrm{B})}{\supseteq}&  \bar{\PREF}^2_{F, 0}(f_0(s)) \cup \PREF^2_{F, \trans_0(s)}\\
&\overset{(\mathrm{C})}{=}&  \bar{\PREF}^2_{F, 0}(f_0(s)) \cup \PREF^2_{F, 1}\\
&\overset{(\mathrm{D})}{=}&  \{00\} \cup \PREF^2_{F, 1}\\
&\overset{(\mathrm{E})}{=}&  \{00\} \cup \{01, 10, 11\}\\
&=& \{00, 01, 10, 11\}
\end{eqnarray}
as desired, where
(A) follows from Lemma \ref{lem:pref-sum} (i),
(B) follows from $s \in \assign_{F, 0}(\lambda)$,
(C) follows from $\trans_0(s) = 1$,
(D) follows from (\ref{eq:e2czit3jvugw}) and Lemma \ref{lem:aifv-decodable0} (iii),
and (E) follows from (\ref{eq:l0ko6bqqdfsf}).
\end{itemize}

(Proof of $F \in \mathscr{F}_{2\hdec}$):
Since $f_0$ and $f_1$ are injective by Definition \ref{def:subsetTaifv} (i), the code-tuple $F$ satisfies Definition \ref{def:k-bitdelay} (ii) (cf. Remark \ref{rem:injective}).
We show that $F$ satisfies Definition \ref{def:k-bitdelay} (i). 
We choose $i \in [2]$ and $s \in \mathcal{S}$ arbitrarily and show $\PREF^2_{F, \trans_i(s)} \cap \bar{\PREF}^2_{F, i}(f_i(s)) = \emptyset$ for the following two cases: the case $\bar{\PREF}^0_{F, i}(f_i(s)) = \emptyset$ and the case $\bar{\PREF}^0_{F, i}(f_i(s)) \neq \emptyset$.

\begin{itemize}
\item The case $\bar{\PREF}^0_{F, i}(f_i(s)) = \emptyset$:
We have
\begin{equation}
\PREF^2_{F, \trans_i(s)} \cap \bar{\PREF}^2_{F, i}(f_i(s))
\overset{(\mathrm{A})}{=} \PREF^2_{F, \trans_i(s)} \cap \emptyset
= \emptyset
\end{equation}
as desired, where
(A) follows from $\bar{\PREF}^0_{F, i}(f_i(s)) = \emptyset$ and Lemma \ref{lem:pref-inc2} (ii) (a).

\item The case $\bar{\PREF}^0_{F, i}(f_i(s)) \neq \emptyset$:
We have
\begin{equation}
\PREF^2_{F, \trans_i(s)} \cap \bar{\PREF}^2_{F, i}(f_i(s))
\overset{(\mathrm{A})}{=} \PREF^2_{F, 1} \cap \bar{\PREF}^2_{F, i}(f_i(s))
\overset{(\mathrm{B})}{=} \{01, 10, 11\} \cap \bar{\PREF}^2_{F, i}(f_i(s))
\overset{(\mathrm{C})}{=} \{01, 10, 11\} \cap \{00\}
= \emptyset
\end{equation}
as desired, where
(A) follows from $\bar{\PREF}^0_{F, i}(f_i(s)) \neq \emptyset$ and Definition \ref{def:k-bitdelay} (iv),
(B) follows from (\ref{eq:l0ko6bqqdfsf}),
and (C) follows from $\bar{\PREF}^0_{F, i}(f_i(s)) \neq \emptyset$ and  Lemma \ref{lem:aifv-decodable0} (iii).
\end{itemize}
\end{proof}

\subsection{Proof of Lemma \ref{lem:propertyT1}}
\label{subsec:proof-propertyT1}

\begin{proof}[Proof of Lemma \ref{lem:propertyT1}]
(Proof of (i)): We have $\PREF^1_{F, i} = \{0, 1\}$ by $F \in \mathscr{F}_1$.
Hence, by Lemma \ref{lem:pref-inc} (i), there exist $a, b \in \mathcal{C}$ such that $0a, 1b \in \PREF^2_{F, i}$.

(Proof of (ii) (a)):
Assume $|\PREF^2_{F, i}| = 2$.
We prove by contradiction assuming that $|f_i(s)| \leq 1$ for some $s \in \mathcal{S}$.
We consider the following two cases separately: the case $|f_i(s)| = 0$ and the case $|f_i(s)| = 1$.
\begin{itemize}
\item The case $|f_i(s)| = 0$: We have
\begin{equation}
\label{eq:107tomuofv9w}
|\bar{\PREF}^0_{F, i}| + 2 |\assign_{F, i}(\lambda)| \overset{(\mathrm{A})}{\leq} |\bar{\PREF}^2_{F, i}| + 2 |\assign_{F, i}(\lambda)| \overset{(\mathrm{B})}{\leq} |\bar{\PREF}^2_{F, i}| + \sum_{s' \in \assign_{F, i}(\lambda)} |\PREF^2_{F, \trans_i(s')}|  \overset{(\mathrm{C})}{=} |\PREF^2_{F, i}| \overset{(\mathrm{D})}{=} 2,
\end{equation}
where
(A) follows from Lemma \ref{lem:pref-inc2} (ii) (b),
(B) follows since $|\PREF^2_{F, \trans_i(s')}| \geq 2$ for any $s' \in \assign_{F, i}(\lambda)$ by (i) of this lemma,
(C) follows from Lemma \ref{lem:pref-sum} (ii),
and (D) follows directly from the assumption.

Also, by $|f_i(s)| = 0$, we have
\begin{equation}
\label{eq:ox7dgz0qdi5h}
|\assign_{F, i}(\lambda)| \geq |\{s\}| = 1.
\end{equation}
By (\ref{eq:107tomuofv9w}) and (\ref{eq:ox7dgz0qdi5h}), we have
\begin{equation}
\label{eq:tcsrvoq3y8wx}
|\bar{\PREF}^0_{F, i}| = 0
\end{equation}
and
\begin{equation}
\label{eq:igbwb7inko25}
|\assign_{F, i}(\lambda)| = 1.
\end{equation}
By (\ref{eq:igbwb7inko25}) and Lemma \ref{lem:leaf} (iii), we obtain $\bar{\PREF}^0_{F, i} \neq \emptyset$, which conflicts with (\ref{eq:tcsrvoq3y8wx}).

\item The case $|f_i(s)| = 1$:
Put $f_i(s) = c \in \mathcal{C}$. We have
\begin{equation}
\label{eq:ow89loddy2iv}
\PREF^2_{F, i} \overset{(\mathrm{A})}{\supseteq} \bar{\PREF}^2_{F, i} \overset{(\mathrm{B})}{\supseteq} c \PREF^1_{F, i} \overset{(\mathrm{C})}{=} c\{0, 1\} = \{c0, c1\},
\end{equation}
where
(A) follows from Lemma \ref{lem:pref-sum} (i),
(B) follows from Lemma \ref{lem:pref-sum} (iii),
and (C) follows from $F \in \mathscr{F}_1$.
Also, by (i) of this lemma, we have
\begin{equation}
\label{eq:ulpchogyfb6d}
\PREF^2_{F, i} \supseteq \{ca, \bar{c}b\}
\end{equation}
for some $a, b \in \mathcal{C}$.
By (\ref{eq:ow89loddy2iv}) and (\ref{eq:ulpchogyfb6d}), we have $|\PREF^2_{F, i}| \geq |\{c0, c1, \bar{c}b\}| = 3$,
which conflicts with $|\PREF^2_{F, i}| = 2$.
\end{itemize}

(Proof of (ii) (b)):
Assume $|\PREF^2_{F, i}| = 2$.
We have
\begin{equation}
\bar{\PREF}^2_{F, i}
\overset{(\mathrm{A})}{=} \bar{\PREF}^2_{F, i} \cup \Big( \bigcup_{s \in \assign_{F, i}(\lambda)} \PREF^k_{F, \trans_i(s)} \Big)
\overset{(\mathrm{B})}{=} \PREF^2_{F, i}
\overset{(\mathrm{C})}{=} \{0a, 1b\}
\end{equation}
for some $a, b \in \mathcal{C}$ as desired, where 
(A) follows because $\assign_{F, i}(\lambda) = \emptyset$ by (ii) (a) of this lemma,
(B) follows from Lemma \ref{lem:pref-sum} (i),
and (C) follows from (i) of this lemma and $|\PREF^2_{F, i}| = 2$.

(Proof of (iii)):
Assume $s \neq s'$ and $f_i(s) = f_i(s')$.
We have
\begin{equation}
\label{eq:qumspsj9n1mg}
|\bar{\PREF}^2_{F, i}(f_i(s))| + |\PREF^2_{F, \trans_i(s)}| + |\PREF^2_{F, \trans_i(s')}|
\overset{(\mathrm{A})}{\leq} |\bar{\PREF}^2_{F, i}(f_i(s))| + \sum_{s'' \in \assign_{F, i}(f_i(s))} |\PREF^2_{F, \trans_i(s'')}|
\overset{(\mathrm{B})}{=} |\PREF^2_{F, i}(f_i(s))| \leq 4,
\end{equation}
where
(A) follows from $s \neq s'$ and $f_i(s) = f_i(s')$,
and (B) follows from Lemma \ref{lem:pref-sum} (ii).

Also, by (i) of this lemma, we have
\begin{equation}
\label{eq:xmldbb8cv9rl}
|\PREF^2_{F, \trans_i(s)}| \geq 2, \quad |\PREF^2_{F, \trans_i(s')}| \geq 2.
\end{equation}
By (\ref{eq:qumspsj9n1mg}) and (\ref{eq:xmldbb8cv9rl}), it must hold that
 $|\bar{\PREF}^2_{F, i}(f_i(s))| = 0$ and $|\PREF^2_{F, \trans_i(s)}| = |\PREF^2_{F, \trans_i(s')}| = 2$
as desired.

(Proof of (iv)): 
We have
\begin{eqnarray}
|\assign_{F, i}(f_i(s))| 
&=& \frac{2|\assign_{F, i}(f_i(s))|}{2}\\
&\overset{(\mathrm{A})}{\leq}& \frac{\sum_{s' \in \assign_{F, i}(f_i(s))} |\PREF^2_{F, \trans_i(s')}|}{2}\\
&\overset{(\mathrm{B})}{=}& \frac{|\PREF^2_{F, i}(f_i(s))| -  |\bar{\PREF}^2_{F, i}(f_i(s))|}{2}\\
&\leq& \frac{4 -  |\bar{\PREF}^2_{F, i}(f_i(s))|}{2}\\
&\overset{(\mathrm{C})}{\leq}& \frac{4 -  |\bar{\PREF}^0_{F, i}(f_i(s))|}{2}\\
&\leq& \begin{cases}
\frac{3}{2} &\,\,\text{if}\,\, \bar{\PREF}^0_{F, i}(f_i(s)) \neq \emptyset,\\
2 &\,\,\text{if}\,\, \bar{\PREF}^0_{F, i}(f_i(s)) = \emptyset,\\
\end{cases}
\end{eqnarray}
as desired, where
(A) follows since $|\PREF^2_{F, \trans_i(s')}| \geq 2$ for any $s' \in \assign_{F, i}(f_i(s))$ by (i) of this lemma,
(B) follows from Lemma \ref{lem:pref-sum} (ii),
and (C) follows from Lemma \ref{lem:pref-inc2} (ii) (b).

(Proof of (v)): We prove by contradiction assuming that there exist $s, s' \in \mathcal{S}$ and $c \in \mathcal{C}$ such that
\begin{equation}
\label{eq:la7vtp6riedu}
f_i(s') = f_i(s) c.
\end{equation}
By (i) of this lemma, we have
\begin{equation}
\label{eq:pd2a6k309u22}
\PREF^2_{F, \trans_i(s)} \owns cc'
\end{equation}
for some $c' \in \mathcal{C}$.
Also, we have
\begin{eqnarray}
\label{eq:vaezvhwa74vs}
\bar{\PREF}^2_{F, i}(f_i(s)) \overset{(\mathrm{A})}{\supseteq} c\PREF^1_{F, i}(f_i(s)c) \overset{(\mathrm{B})}{=} c\PREF^1_{F, i}(f_i(s')) \overset{(\mathrm{C})}{\supseteq} c\PREF^1_{F, \trans_i(s')} \overset{(\mathrm{D})}{=} c\{0, 1\} \owns cc',
\end{eqnarray}
where
(A) follows from Lemma \ref{lem:pref-sum} (iii),
(B) follows from (\ref{eq:la7vtp6riedu}),
(C) follows from Lemma \ref{lem:pref-sum} (i),
and (D) follows from $F \in \mathscr{F}_1$.
By (\ref{eq:pd2a6k309u22}) and (\ref{eq:vaezvhwa74vs}), we obtain $\PREF^2_{F, \trans_i(s)} \cap \bar{\PREF}^2_{F, i}(f_i(s)) \neq \emptyset$, which conflicts with $F \in \mathscr{F}_{2\hdec}$.

(Proof of (vi)): We prove by contradiction assuming that there exist $s \in \mathcal{S}$ and $c \in \mathcal{C}$ such that
\begin{equation}
\label{eq:o75as6t1z1y3}
\bar{\PREF}^1_{F, i}(f_i(s) c) = \{0, 1\}.
\end{equation}
By (i) of this lemma, we have
\begin{equation}
\label{eq:b13ch8706u09}
\PREF^2_{F, \trans_i(s)} \owns cc'
\end{equation}
for some $c' \in \mathcal{C}$.
Also, we have
\begin{eqnarray}
\label{eq:tmf097q6kika}
\bar{\PREF}^2_{F, i}(f_i(s)) \overset{(\mathrm{A})}{\supseteq} c\PREF^1_{F, i}(f_i(s)c) \overset{(\mathrm{B})}{\supseteq} c\bar{\PREF}^1_{F, i}(f_i(s)c) \overset{(\mathrm{C})}{=} c\{0, 1\}\owns cc',
\end{eqnarray}
where
(A) follows from Lemma \ref{lem:pref-sum} (iii),
(B) follows from Lemma \ref{lem:pref-sum} (i),
and (C) follows from (\ref{eq:o75as6t1z1y3}).
By (\ref{eq:b13ch8706u09}) and (\ref{eq:tmf097q6kika}), we obtain
$\PREF^2_{F, \trans_i(s)} \cap \bar{\PREF}^2_{F, i}(f_i(s)) \neq \emptyset$, which conflicts with $F \in \mathscr{F}_{2\hdec}$.
\end{proof}

\subsection{Proof of Lemma \ref{lem:complete-leaf}}
\label{subsec:proof-complete-leaf}

To prove Lemma \ref{lem:complete-leaf}, we use the following Lemma \ref{lem:complete} obtained by \cite[Theorem 2]{IEICE2023} with $k = 2$.
See \cite{IEICE2023} for the original statement and the proof.

\begin{lemma}
\label{lem:complete}
For any $F \in \mathscr{F}_{\opt}$, $i \in \kernel_F$, and $\pmb{b} = b_1b_2\ldots b_l \in \mathcal{C}^{\ast}$, if $|\pmb{b}| \geq 2$ and $b_1b_2 \in \PREF^2_{F, i}$, then $\pmb{b} \in \PREF^{\ast}_{F, i}$,
where $\PREF^{\ast}_{F, i} \coloneqq \PREF^0_{F, i} \cup \PREF^1_{F, i} \cup \PREF^2_{F, i} \cup \cdots$.
\end{lemma}

\begin{proof}[Proof of Lemma \ref{lem:complete-leaf}]
Assume $\bar{\PREF}^0_{F, i}(f_i(s)) = \emptyset$ and $|\assign_{F, i}(f_i(s))| = 1$.
We prove by contradiction assuming $|\PREF^2_{F, \trans_i(s)}| < 4$, that is, there exists
\begin{equation}
\label{eq:jze6jqdy4507}
\pmb{b} = b_1b_2 \in \mathcal{C}^2 \setminus \PREF^2_{F, \trans_i(s)}.
\end{equation}
First, we put
\begin{equation}
\label{eq:tvzly8tb85os}
\pmb{d} = d_1d_2\ldots d_l \coloneqq f_i(s)\pmb{b}
\end{equation}
and show
\begin{equation}
\label{eq:2opxijz9ic97}
d_1d_2 \in \PREF^2_{F, i}
\end{equation}
considering the following three cases separately: the case $|f_i(s)| = 0$, the case $|f_i(s)| = 1$, and the case $|f_i(s)| \geq 2$.
\begin{itemize}
\item The case $|f_i(s)| = 0$:
We have
\begin{equation}
\bar{\PREF}^0_{F, i}(f_i(s))
\overset{(\mathrm{A})}{=} \bar{\PREF}^0_{F, i}
\overset{(\mathrm{B})}{\neq} \emptyset,
\end{equation}
where
(A) follows from $|f_i(s)| = 0$,
and (B) follows from $|\assign_{F, i}(f_i(s))| = 1$ and Lemma \ref{lem:leaf} (iii).
This conflicts with the assumption.
Therefore, the case $|f_i(s)| = 0$ is impossible.

\item The case $|f_i(s)| = 1$: Then we have $f_i(s) = d_1$ by (\ref{eq:tvzly8tb85os}).
Also, we have $d_2 \in \{0, 1\} = \PREF^1_{F, \trans_i(s)}$ by $F \in \mathscr{F}_1$.
Thus, there exists $\pmb{x} \in \mathcal{S}^{+}$ such that $f^{\ast}_{\trans_i(s)}(\pmb{x}) \succeq d_2$.
Then we have $f^{\ast}_i(s\pmb{x}) = f_i(s)f^{\ast}_{\trans_i(s)}(\pmb{x}) \succeq d_1d_2$, which leads to (\ref{eq:2opxijz9ic97}).
\item The case $|f_i(s)| \geq 2$: Directly from $f_i(s) \succeq d_1d_2$ by (\ref{eq:tvzly8tb85os}).
\end{itemize}
Consequently, (\ref{eq:2opxijz9ic97}) holds.

By $i \in \kernel_{F}$ and (\ref{eq:2opxijz9ic97}), we obtain $\pmb{d} \in \PREF^{\ast}_{F, i}$ applying Lemma \ref{lem:complete}.
Hence, there exists $\pmb{y} = y_1y_2\ldots y_n \in \mathcal{S}^{+}$ such that
\begin{equation}
\label{eq:3y3n7k0vxrmk}
f^{\ast}_i(\pmb{y}) \succeq \pmb{d}.
\end{equation}
By (\ref{eq:tvzly8tb85os}) and (\ref{eq:3y3n7k0vxrmk}), exactly one of $f_i(y_1) \succ f_i(s)$ and $f_i(y_1) \preceq f_i(s)$ holds.
Now, the latter $f_i(y_1) \preceq f_i(s)$ must hold because the former $f_i(y_1) \succ f_i(s)$ conflicts with $\bar{\PREF}^0_{F, i}(f_i(s)) = \emptyset$ by Lemma \ref{lem:leaf} (i).
Therefore, there exists $\pmb{c} = c_1c_2\ldots c_r \in \mathcal{C}^{\ast}$ such that $f_i(y_1)\pmb{c} = f_i(s)$.
We divide into the following three cases by $|\pmb{c}|$.

\begin{itemize}
\item The case $|\pmb{c}| = 0$: We have $f_i(y_1) = f_i(s)$, which leads to $y_1 = s$ by $|\assign_{F, i}(f_i(s))| = 1$.
Hence, we have
\begin{equation}
f_i(s)f^{\ast}_{\trans_i(s)}(\suff(\pmb{y}))
=  f_i(y_1)f^{\ast}_{\trans_i(y_1)}(\suff(\pmb{y}))
= f^{\ast}_i(\pmb{y})
\overset{(\mathrm{A})}{\succeq} \pmb{d}
\overset{(\mathrm{B})}{=} f_i(s)\pmb{b},
\end{equation}
where
(A) follows from (\ref{eq:3y3n7k0vxrmk}),
and (B) follows from (\ref{eq:tvzly8tb85os}).
Comparing both sides, we obtain $\allowbreak f^{\ast}_{\trans_i(s)}(\suff(\pmb{y})) \succeq \pmb{b}$.
This leads to $\pmb{b} \in \PREF^2_{F, \trans_i(s)}$, which conflicts with (\ref{eq:jze6jqdy4507}).

\item The case $|\pmb{c}| = 1$: We have $f_i(y_1) = f_i(s)c_1$, which conflicts with Lemma \ref{lem:propertyT1} (v).

\item The case $|\pmb{c}| \geq 2$: We have
\begin{equation}
\label{eq:eap0acivmujc}
f_i(y_1)c_1c_2 \preceq f_i(s).
\end{equation}
which leads to
\begin{equation}
\label{eq:t9ucivzc29gh}
c_1c_2 \in \bar{\PREF}^2_{F, i}(f_i(y_1)).
\end{equation}
Also, we have
\begin{equation}
f_i(y_1)f^{\ast}_{\trans_i(y_1)}(\suff(\pmb{y}))
= f^{\ast}_i(\pmb{y})
\overset{(\mathrm{A})}{\succeq} \pmb{d}
\overset{(\mathrm{B})}{=} f_i(s)\pmb{b}
\succeq f_i(s)
\overset{(\mathrm{C})}{\succeq} f_i(y_1)c_1c_2,
\end{equation}
where
(A) follows from (\ref{eq:3y3n7k0vxrmk}),
(B) follows from (\ref{eq:tvzly8tb85os}),
and (C) follows from (\ref{eq:eap0acivmujc}).
Comparing both sides, we obtain $f^{\ast}_{\trans_i(y_1)}(\suff(\pmb{y})) \succeq c_1c_2$,
which leads to
\begin{equation}
\label{eq:xug0c0smuwng}
c_1c_2 \in \PREF^2_{F, \trans_i(y_1)}.
\end{equation}
By (\ref{eq:t9ucivzc29gh}) and (\ref{eq:xug0c0smuwng}), we obtain $\bar{\PREF}^2_{F, i}(f_i(y_1)) \cap \PREF^2_{F, \trans_i(y_1)} \neq \emptyset$, which conflicts with $F \in \mathscr{F}_{2\hdec}$.
\end{itemize}
\end{proof}

\subsection{Proof of Lemma \ref{lem:tilde-equal} (iii)}
\label{subsec:proof-tilde-equal}

To prove Lemma \ref{lem:tilde-equal} (iii), we prove the following Lemmas \ref{lem:tilde-equal0} and \ref{lem:tilde-equal02}.

\begin{lemma}
\label{lem:tilde-equal0}
Let $F \in \mathscr{F}_1$, $i \in [F]$, and $s, s' \in \mathcal{S}$, and let $\gamma(s_1)\gamma(s_2)\ldots \gamma(s_{\rho})$ (resp. $\gamma(s'_1)\gamma(s'_2)\ldots \gamma(s'_{\rho'})$) be the $\gamma$-decomposition of $f_i(s)$ (resp. $f_i(s')$).
For any $r = 1, 2, \ldots, m \coloneqq \min\{\rho, \rho'\}$, 
if one of the following conditions (a) and (b) holds, then $\gamma(s_r) = \gamma(s'_r) \iff \dot{\gamma}(s_{r}) = \dot{\gamma}(s'_r) $:
\begin{enumerate}[(a)]
\item $r = 1$.
\item $r \geq 2$ and $s_{r-1} = s'_{r-1}$.
\end{enumerate}
\end{lemma}

\begin{proof}[Proof of Lemma \ref{lem:tilde-equal0}]
Assume that the condition (a) or (b) holds.

($\implies$) Directly from (\ref{eq:6807mxs1xye4}).

($\impliedby$) We prove the contraposition.
Namely, we prove $\dot{\gamma}(s_r) \neq \dot{\gamma}(s'_r)$ assuming $\gamma(s_{r}) \neq \gamma(s'_r)$.
Put $\gamma(s_{r}) = g_1g_2\ldots g_l$ and $\gamma(s'_r) = g'_1g'_2\ldots g'_{l'}$.
We consider the following two cases separately: the case $|\gamma(s_r)| \neq |\gamma(s'_r)|$ and the case $|\gamma(s_r)| = |\gamma(s'_r)|$.
\begin{itemize}
\item The case $|\gamma(s_r)| \neq |\gamma(s'_r)|$: We have
\begin{equation}
|\dot{\gamma}(s_r)| \overset{(\mathrm{A})}{=} |\gamma(s_r)| \overset{(\mathrm{B})}{\neq} |\gamma(s'_r)| \overset{(\mathrm{C})}{=} |\dot{\gamma}(s'_r)|,
\end{equation}
where
(A) follows from Lemma \ref{lem:tilde-equal} (i),
(B) follows from the assumption,
and (C) follows from Lemma \ref{lem:tilde-equal} (i).
This implies $\dot{\gamma}(s_r) \neq \dot{\gamma}(s'_r)$ as desired.

\item The case $|\gamma(s_r)| = |\gamma(s'_r)|$:
If $|\gamma(s_r)| = |\gamma(s'_r)| \geq 3$ and $g_3g_4\ldots g_{l} \neq g'_3g'_4\ldots g'_{l'}$, then we obtain $\dot{\gamma}(s_r) \neq \dot{\gamma}(s'_r)$ directly from (\ref{eq:6807mxs1xye4}).
Thus, we assume 
\begin{equation}
\label{eq:mzoa7o8x11w8}
g_j \neq g'_j \,\,\text{for some}\,\, 1 \leq j \leq \min\{2, |\gamma(s_r)|\}.
\end{equation}

We divide into the following two cases by which of the conditions (a) and (b) holds: the case $r = 1$ and the case $r \geq 2, s_{r-1} = s'_{r-1}$.
\begin{itemize}
\item The case $r = 1$: We consider the following two cases separately: the case $|\PREF^2_{F, i}| = 2$ and the case $|\PREF^2_{F, i}| \geq 3$.

\begin{itemize}
\item The case $|\PREF^2_{F, i}| = 2$: 
By Lemma \ref{lem:propertyT1} (ii), we have $\PREF^2_{F, i} = \{0a, 1b\}$ for some $a, b \in \mathcal{C}$ and
we have $|\gamma(s_1)| = |\gamma(s'_1)| \geq 2$.
This shows $g_1g_2, g'_1g'_2 \in \{0a, 1b\}$.
Hence, since $g_1g_2 \neq g'_1g'_2$ by (\ref{eq:mzoa7o8x11w8}), we may assume
\begin{equation}
\label{eq:4yzbb62w8vq0}
g_1 \neq g'_1.
\end{equation}
Thus, we obtain
\begin{equation}
\dot{\gamma}(s_r) \overset{(\mathrm{A})}{=} a_{F, i} g_1 g_3 g_4 \ldots g_l  \overset{(\mathrm{B})}{\neq} a_{F, i} g'_1 g'_3 g'_4 \ldots g'_{l'}  \overset{(\mathrm{C})}{=}  \dot{\gamma}(s'_r)
\end{equation}
as desired, where
(A) follows from the first case of (\ref{eq:6807mxs1xye4}) since $r = 1$ and $|\PREF^2_{F, i}| = 2$,
(B) follows from (\ref{eq:4yzbb62w8vq0}),
and (C) follows from the first case of (\ref{eq:6807mxs1xye4}) since $r = 1$ and $|\PREF^2_{F, i}| = 2$.

\item The case $|\PREF^2_{F, i}| \geq 3$: 
We obtain
\begin{equation}
\dot{\gamma}(s_r) \overset{(\mathrm{A})}{=} \gamma(s_r)  \overset{(\mathrm{B})}{\neq} \gamma(s'_r) \overset{(\mathrm{C})}{=}  \dot{\gamma}(s'_r)
\end{equation}
as desired, where
(A) follows from the second case of (\ref{eq:6807mxs1xye4}) since $r = 1$ and $|\PREF^2_{F, i}| \geq 3$,
(B) follows from (\ref{eq:mzoa7o8x11w8}),
and (C) follows from the second case of (\ref{eq:6807mxs1xye4}) since $r = 1$ and $|\PREF^2_{F, i}| \geq 3$.
\end{itemize}

\item The case $r \geq 2, s_{r-1} = s'_{r-1}$:
By Lemma \ref{lem:dot-length} (iii), we have $g_1g_2 \in \bar{\PREF}^2_{F, i}(f_i(s_{r-1}))$ and $g'_1g'_2 \in \bar{\PREF}^2_{F, i}(f_i(s'_{r-1}))$.
Since $s_{r-1} = s'_{r-1}$, we have
\begin{equation}
\label{eq:zwl5e2nk9t36}
\{g_1g_2, g'_1g'_2\} \subseteq \bar{\PREF}^2_{F, i}(f_i(s_{r-1})).
\end{equation}

Now, we show
\begin{equation}
\label{eq:ja6s7uodbdon}
g_1 \neq g'_1
\end{equation}
by contradiction assuming the contrary $g_1 = g'_1$.
Then by (\ref{eq:mzoa7o8x11w8}), it must hold that $|\gamma(s_r)| = |\gamma(s'_r)| \geq 2$ and $g_2 \neq g'_2$.
Hence, we have
\begin{eqnarray}
\label{eq:s0xt1lbh5fvr}
g_1\PREF^1_{F, i}(f_i(s_{r-1})g_1) \cup \bar{g_1}\PREF^1_{F, i}(f_i(s_{r-1})\bar{g_1})
&\overset{(\mathrm{A})}{=}& \bar{\PREF}^2_{F, i}(f_i(s_{r-1}))\\
&\overset{(\mathrm{B})}{\supseteq}& \{g_1g_2, g'_1g'_2\}\\
&\overset{(\mathrm{C})}{=}& g_1\{g_2, g'_2\}\\
&\overset{(\mathrm{D})}{=}& g_1\{0, 1\},
\end{eqnarray}
where
(A) follows from Lemma \ref{lem:pref-sum} (iii), 
(B) follows from (\ref{eq:zwl5e2nk9t36}),
(C) follows from $g_1 = g'_1$ and (\ref{eq:ja6s7uodbdon}),
and (D) follows from $g_2 \neq g'_2$.
Comparing both sides of (\ref{eq:s0xt1lbh5fvr}), we obtain $\PREF^1_{F, i}(f_i(s_{r-1})g_1) = \{0, 1\}$, which conflicts with Lemma \ref{lem:propertyT1} (vi).
Hence, we conclude that (\ref{eq:ja6s7uodbdon}) holds.

We have
\begin{equation}
\label{eq:6mg0ji0y1mca}
|\bar{\PREF}^1_{F, i}(f_i(s_{r-1}))| \overset{(\mathrm{A})}{=} |\{g_1, g'_1\}| \overset{(\mathrm{B})}{=} |\{0, 1\}| = 2,
\end{equation}
where (A) follows from (\ref{eq:zwl5e2nk9t36}) and Lemma \ref{lem:pref-inc} (ii),
and (B) follows from (\ref{eq:ja6s7uodbdon}).
Therefore, we obtain
\begin{equation}
\dot{\gamma}(s_r) \overset{(\mathrm{A})}{=} \bar{a}_{F, \trans_i(s_{r-1})} g_1 g_3 g_4 \ldots g_l  \overset{(\mathrm{B})}{\neq} \bar{a}_{F, \trans_i(s'_{r-1})} g'_1 g'_3 g'_4 \ldots g'_{l'}  \overset{(\mathrm{C})}{=}  \dot{\gamma}(s'_r)
\end{equation}
as desired, where
(A) follows from the third case of (\ref{eq:6807mxs1xye4}) since $r \geq 2$ and (\ref{eq:6mg0ji0y1mca}) hold,
(B) follows from (\ref{eq:ja6s7uodbdon}),
and (C) follows from the third case of (\ref{eq:6807mxs1xye4}) since $r \geq 2$ and (\ref{eq:6mg0ji0y1mca}) hold.
\end{itemize}

\end{itemize}
\end{proof}

\begin{lemma}
\label{lem:tilde-equal02}
Let $F \in \mathscr{F}_1$, $i \in [F]$, and $s, s' \in \mathcal{S}$, and let $\gamma(s_1)\gamma(s_2)\ldots \gamma(s_{\rho})$ (resp. $\gamma(s'_1)\gamma(s'_2)\ldots \gamma(s'_{\rho'})$) be the $\gamma$-decomposition of $f_i(s)$ (resp. $f_i(s')$).
If $\dot{f}_i(s) \preceq \dot{f}_i(s')$, then 
for any $r = 1, 2, \ldots, m \coloneqq \min\{\rho, \rho'\}$, we have $\gamma(s_{r}) = \gamma(s'_{r})$.
\end{lemma}

\begin{proof}[Proof of Lemma \ref{lem:tilde-equal02}]
Assume 
\begin{equation}
\label{eq:w464gftu6uii}
\dot{f}_i(s) \preceq \dot{f}_i(s').
\end{equation}
It suffices to prove that the following conditions (a) and (b) hold for any $r = 1, 2, \ldots, m$ by induction for $r$.
\begin{enumerate}[(a)]
\item $\gamma(s_{r}) = \gamma(s'_{r})$.
\item If $r \neq m$, then $s_{r} = s'_{r}$.
\end{enumerate}
We fix $q \geq 1$ and show that (a) and (b) hold for $r = q$ under the assumption that (a) and (b) hold for any $r = 1, 2, \ldots, q-1$.

We first show that the condition (a) holds for $r = q$.
We have
\begin{eqnarray}
\dot{f}_i(s_{q-1})\dot{\gamma}(s_q) \dot{\gamma}(s_{q+1}) \ldots \dot{\gamma}(s_{\rho})
&=& \dot{f}_i(s)\\
&\overset{(\mathrm{A})}{\preceq}& \dot{f}_i(s')\\
&=& \dot{f}_i(s'_{q-1})\dot{\gamma}(s'_q) \dot{\gamma}(s'_{q+1}) \ldots \dot{\gamma}(s'_{\rho'})\\
&\overset{(\mathrm{B})}{=}& \dot{f}_i(s_{q-1})\dot{\gamma}(s'_q) \dot{\gamma}(s'_{q+1}) \ldots \dot{\gamma}(s'_{\rho'})
\end{eqnarray}
where
we suppose $\dot{f}_i(s_{q-1}) \coloneqq \lambda$ for the case $q = 1$, and
(A) follows from (\ref{eq:w464gftu6uii}),
and (B) follows from the induction hypothesis.
Comparing both sides, we have
\begin{equation}
\dot{\gamma}(s_q) \dot{\gamma}(s_{q+1}) \ldots \dot{\gamma}(s_{\rho}) \preceq \dot{\gamma}(s'_q) \dot{\gamma}(s'_{q+1}) \ldots \dot{\gamma}(s'_{\rho'}).
\end{equation}
Hence, at least one of $\dot{\gamma}(s_q) \preceq \dot{\gamma}(s'_q)$ and $\dot{\gamma}(s_q) \succeq \dot{\gamma}(s'_q)$ holds.
We show that both relations hold, that is,
\begin{equation}
\label{eq:q4w3fk0eznui}
\dot{\gamma}(s_q) = \dot{\gamma}(s'_q)
\end{equation}
by contradiction.
Assume that one does not hold, that is, $\gamma(s_q) \prec \gamma(s'_q)$ by symmetry.
Then we have
\begin{eqnarray}
f_i(s_q)
&=& \gamma(s_1)\gamma(s_2)\ldots \gamma(s_{q-1})\gamma(s_q)\\
&\overset{(\mathrm{A})}{=}& \gamma(s'_1)\gamma(s'_2)\ldots \gamma(s'_{q-1})\gamma(s_q)\\
&\prec& \gamma(s'_1)\gamma(s'_2)\ldots \gamma(s'_{q-1})\gamma(s'_q)\\
&=& f_i(s'_q), \label{eq:t9s1ps62bn2y}
\end{eqnarray}
where (A) follows from the induction hypothesis.
Hence, we obtain
\begin{equation}
s_q
\overset{(\mathrm{A})}{\in} \assign^{\prec}_{F, i}(f_i(s'_q))
=  \{s'_1, s'_2, \ldots s'_{q-1}\}
\overset{(\mathrm{B})}{=} \{s_1, s_2, \ldots s_{q-1}\},
\end{equation}
where (A) follows from (\ref{eq:t9s1ps62bn2y}),
and (B) follows from the induction hypothesis.
This conflicts with the definition of $\gamma$-decomposition of $f_i(s'_{\rho'})$.
Consequently, (\ref{eq:q4w3fk0eznui}) holds.

Since $q = 1$ or $s_{q-1} = s'_{q-1}$ hold by the induction hypothesis and (\ref{eq:q4w3fk0eznui}) holds, we obtain $\gamma(s_q) = \gamma(s'_q)$ by applying Lemma \ref{lem:tilde-equal0}.
Namely, the condition (a) holds for $r = q$.

Next, we show that the condition (b) holds for $r = q$.
We have
\begin{equation}
\label{eq:d8wtr5uk5a99}
f_i(s_q) = \gamma(s_1)\gamma(s_2)\ldots \gamma(s_q)
\overset{(\mathrm{A})}{=} \gamma(s'_1)\gamma(s'_2)\ldots \gamma(s'_q) = f_i(s'_q),
\end{equation}
where (A) follows from the induction hypothesis and $\gamma(s_q) = \gamma(s'_q)$ proven above.
Also, if $q \neq m$, then we have $\bar{\PREF}^0_{F, i}(f_i(s_q)) \neq \emptyset$ applying Lemma \ref{lem:leaf} (i) since $f_i(s_q) \prec f_i(s_m)$.
Hence, by Lemma \ref{lem:propertyT1} (iv), we have
\begin{equation}
\label{eq:pdrwu402wgpb}
|\assign_{F, i}(f_i(s_q))| = 1.
\end{equation}
By (\ref{eq:d8wtr5uk5a99}) and (\ref{eq:pdrwu402wgpb}),
it must hold that $s_q = s'_q$.
Namely, the condition (b) holds for $r = q$.
\end{proof}

\begin{proof}[Proof of Lemma \ref{lem:tilde-equal} (iii)]
Let $\gamma(s_1)\gamma(s_2)\ldots \gamma(s_{\rho})$ (resp. $\gamma(s'_1)\gamma(s'_2)\ldots \gamma(s'_{\rho'})$) be the $\gamma$-decomposition of $f_i(s)$ (resp. $f_i(s')$).

($\implies$): Assume $f_i(s) \preceq f_i(s')$. Then we have
\begin{equation}
f_i(s') = \gamma(s_1)\gamma(s_2)\ldots \gamma(s_{\rho})\gamma(s'_{\rho+1})\gamma(s'_{\rho+2})\ldots \gamma(s'_{\rho'}).
\end{equation}
Hence, we obtain
\begin{eqnarray}
\dot{f}_i(s) &=& \dot{\gamma}(s_1)\dot{\gamma}(s_2)\ldots \dot{\gamma}(s_{\rho})\\
&\preceq& \dot{\gamma}(s_1)\dot{\gamma}(s_2)\ldots \dot{\gamma}(s_{\rho})\dot{\gamma}(s'_{\rho+1})\dot{\gamma}(s'_{\rho+2})\ldots \dot{\gamma}(s'_{\rho'})\\
&=& \dot{f}_i(s')
\end{eqnarray}
as desired.

($\impliedby$): Assume
\begin{equation}
\label{eq:s1cjkvkxq3wk}
\dot{f}_i(s) \preceq \dot{f}_i(s').
\end{equation}
Then we have
\begin{equation}
\label{eq:wigp6x40codu}
f_i(s_m) = \gamma(s_1)\gamma(s_2)\ldots \gamma(s_m)
\overset{(\mathrm{A})}{=} \gamma(s'_1)\gamma(s'_2)\ldots \gamma(s'_m) = f_i(s'_m),
\end{equation}
where $m \coloneqq \{\rho, \rho'\}$ and (A) follows from Lemma \ref{lem:tilde-equal02}.
This implies 
\begin{equation}
\label{eq:e46zvzx9um3v}
\dot{f}_i(s_m) = \dot{f}_i(s'_m)
\end{equation}
by ($\implies$) of this lemma.
We consider the following two cases separately:
the case $m = \rho \leq \rho'$ and the case $m = \rho' < \rho$.
\begin{itemize}
\item The case $m = \rho \leq \rho'$:
We have
\begin{equation}
f_i(s) = f_i(s_m)
\overset{(\mathrm{A})}{=} f_i(s'_m)
\overset{(\mathrm{B})}{\preceq} f_i(s'_m)\gamma(s'_{m+1})\gamma(s'_{m+2})\ldots \gamma(s'_{\rho'}) = f_i(s')
\end{equation}
as desired, where
(A) follows from (\ref{eq:wigp6x40codu}),
and (B) follows from $m = \rho \leq \rho'$.

\item The case $m = \rho' < \rho$:
We show that this case is impossible.
We have
\begin{equation}
\dot{f}_i(s_m)\dot{\gamma}(s_{m+1})\dot{\gamma}(s_{m+2})\ldots \dot{\gamma}(s_{\rho})
= \dot{f}_i(s)
\overset{(\mathrm{A})}{\preceq} \dot{f}_i(s')
\overset{(\mathrm{B})}{=} \dot{f}_i(s'_m)
\overset{(\mathrm{C})}{=} \dot{f}_i(s_m),
\end{equation}
where
(A) follows from (\ref{eq:s1cjkvkxq3wk}),
(B) follows from $m = \rho'$,
and (C) follows from (\ref{eq:e46zvzx9um3v}).
Comparing both sides, we obtain $\dot{\gamma}(s_{m+1})\allowbreak\dot{\gamma}(s_{m+2})\ldots \dot{\gamma}(s_{\rho}) = \lambda$, which leads to $\gamma(s_{m+1})\gamma(s_{m+2})\ldots \gamma(s_{\rho}) = \lambda$ by Lemma \ref{lem:tilde-equal} (i).
In particular, we have $\gamma(s_{m+1}) = \lambda$ by $m < \rho$.
This conflicts with Lemma \ref{lem:dot-length} (ii).
\end{itemize}
\end{proof}

\subsection{Proof of Lemma \ref{lem:tilde-pref}}
\label{subsec:proof-tilde-pref}

\begin{proof}[Proof of Lemma \ref{lem:tilde-pref}]
(Proof of (i) (a)):
For any $\pmb{x} = x_1x_2\ldots x_n \in \mathcal{S}^{\ast}$, we have
\begin{equation}
\label{eq:snt68cyg69lt}
|\dot{\gamma}(s_1)| \overset{(\mathrm{A})}{=}  |\gamma(s_1)| \overset{(\mathrm{B})}{\geq} 2,
\end{equation}
where $\gamma(s_1)\gamma(s_2)\ldots\gamma(s_{\rho})$ is the $\gamma$-decomposition of $f_i(x_1)$, and
(A) follows from Lemma \ref{lem:tilde-equal} (i),
and (B) follows from $|\PREF^2_{F, i}| = 2$ and Lemma \ref{lem:dot-length} (ii).

For any $c \in \mathcal{C}$, we have 
\begin{eqnarray}
c \in \PREF^1_{F, i}
&\iff& {}^{\exists}\pmb{x} \in \mathcal{S}^{+}; f^{\ast}_i(\pmb{x}) \succeq c \\
&\overset{(\mathrm{A})}{\iff}& {}^{\exists}\pmb{x} \in \mathcal{S}^{+}; {}^{\exists}c' \in \mathcal{C}; \gamma(s_1) \succeq cc'\\
&\overset{(\mathrm{B})}{\iff}& {}^{\exists}\pmb{x} \in \mathcal{S}^{+}; \dot{\gamma}(s_1) \succeq a_{F, i}c\\
&\overset{(\mathrm{C})}{\iff}& {}^{\exists}\pmb{x} \in \mathcal{S}^{+}; \dot{f}^{\ast}_i(\pmb{x}) \succeq a_{F, i}c\\
&\iff& a_{F, i}c \in \PREF^2_{\dot{F}, i} \label{lem:k2eis0ln9p0n},
\end{eqnarray}
where $\pmb{x} = x_1x_2\ldots x_n$, and $\gamma(s_1)\gamma(s_2)\ldots \gamma(s_{\rho})$ is the $\gamma$-decomposition of $f_i(x_1)$, and
(A) follows from (\ref{eq:snt68cyg69lt}),
(B) follows from $|\PREF^2_{F, i}| = 2$ and the first case of (\ref{eq:6807mxs1xye4}),
and (C) follows from (\ref{eq:snt68cyg69lt}).
Since $\PREF^1_{F, i} = \{0, 1\}$ by $F \in \mathscr{F}_1$, we obtain $\PREF^2_{\dot{F}, i} = \{a_{F, i}0, a_{F, i}1\}$ by (\ref{lem:k2eis0ln9p0n}) as desired.

(Proof of (i) (b)):
Assume $|\PREF^2_{F, j}| \geq 3$.
We consider the three cases of the right hand side of (\ref{eq:ldq1qk35kx1g}) separately.
\begin{itemize}
\item The case $|\bar{\PREF}^1_{F, i}(f_i(s))| = 0$: Clearly, we have $\PREF^2_{\dot{F}, j} \subseteq\{00, 01, 10, 11\}$ as desired.

\item The case $|\bar{\PREF}^1_{F, i}(f_i(s))| = 1, |\bar{\PREF}^1_{F, j}| = 1$:
We have
\begin{equation}
\label{eq:4c3d7d353z5c}
1 \overset{(\mathrm{A})}{\geq} |\assign_{F, j}(\lambda)| = \frac{2|\assign_{F, j}(\lambda)|}{2}
\overset{(\mathrm{B})}{=} \frac{\sum_{s \in \assign_{F, j}(\lambda)}|\PREF^1_{F, \trans_j(s)}|}{2}
\overset{(\mathrm{C})}{\geq}  \frac{|\PREF^1_{F, j}| - |\bar{\PREF}^1_{F, j}|}{2}
\overset{(\mathrm{D})}{=}  \frac{2-1}{2} > 0,
\end{equation}
where
(A) follows from Lemma \ref{lem:propertyT1} (iv) because $\bar{\PREF}^0_{F, j} \neq \emptyset$ holds by $|\bar{\PREF}^1_{F, j}| = 1$ and Lemma \ref{lem:pref-inc2} (ii) (a),
(B) follows since $|\PREF^1_{F, \trans_j(s)}| = 2$ from $F \in \mathscr{F}_1$,
(C) follows from Lemma \ref{lem:pref-sum} (i),
and (D) follows from $F \in \mathscr{F}_1$ and $|\bar{\PREF}^1_{F, j}| = 1$.
Thus, we have $|\assign_{F, j}(\lambda)| = 1$, that is, there exists $s' \in \mathcal{S}$ such that
\begin{equation}
\label{eq:4c3d7d353z5c}
\assign_{F, j}(\lambda) = \{s'\}.
\end{equation}

Now, we have
\begin{equation}
\label{eq:da6ansmwz4bu}
|\PREF^2_{F, \trans_j(s')}| = 2
\end{equation}
because 
\begin{eqnarray}
2 \overset{(\mathrm{A})}{\leq} |\PREF^2_{F, \trans_j(s')}|
\overset{(\mathrm{B})}{=} |\PREF^2_{F, j}| - |\bar{\PREF}^2_{F, j}| \overset{(\mathrm{C})}{\leq} |\PREF^2_{F, j}| - |\bar{\PREF}^1_{F, j}| \overset{(\mathrm{D})}{=} |\PREF^2_{F, j}| - 1 \overset{(\mathrm{E})}{\leq} 3 - 1 = 2,
\end{eqnarray}
where
(A) follows from Lemma \ref{lem:propertyT1} (i),
(B) follows from Lemma \ref{lem:pref-sum} (ii),
(C) follows from Lemma \ref{lem:pref-inc2} (ii) (b),
(D) follows from $|\bar{\PREF}^1_{F, j}| = 1$,
and (E) follows from Lemma \ref{cor:pref-sum} and $|\bar{\PREF}^1_{F, i}(f_i(s))| = 1$.

Hence, applying the first case of (i) of this lemma, we obtain
\begin{equation}
\label{eq:xrwsnm4lxs1y}
\PREF^2_{\dot{F}, \dot{\trans}_j(s')} = \{a_{F, \trans_j(s')}0, a_{F, \trans_j(s')}1\}.
\end{equation}
Also, by (\ref{eq:da6ansmwz4bu}) and Lemma \ref{lem:propertyT1} (ii) (b), we have $\bar{\PREF}^2_{F, \trans_j(s')} = \{0a, 1b\}$ for some $a, b \in \mathcal{C}$. 
Hence, by Lemma \ref{lem:pref-inc} (ii), we obtain
\begin{equation}
\label{eq:kbl91g7jnwbk}
|\bar{\PREF}^1_{F, \trans_j(s')}| = |\{0, 1\}| = 2.
\end{equation}

Thus, for any $\pmb{x} = x_1x_2\ldots, x_n \in \mathcal{S}^{+}$, we have
\begin{eqnarray}
\dot{f}_j(x_1)
&=& \dot{\gamma}(s_1)\dot{\gamma}(s_2)\ldots \dot{\gamma}(s_{\rho-1})\dot{\gamma}(s_\rho)\\
&\succeq& \dot{\gamma}(s_1)\dot{\gamma}(s_2)\\
&\overset{(\mathrm{A})}{=}& \dot{\gamma}(s')\dot{\gamma}(s_2)\\
&\overset{(\mathrm{B})}{\succeq}& \dot{\gamma}(s')\bar{a}_{F, \trans_j(s')}1\\
&\overset{(\mathrm{C})}{=}& \bar{a}_{F, \trans_j(s')}1,
\end{eqnarray}
where
$\gamma(s_1)\gamma(s_2)\ldots \gamma(s_{\rho-1})\gamma(s_\rho)$ is the $\gamma$-decomposition of $f_j(x_1)$, and 
(A) follows from (\ref{eq:4c3d7d353z5c}) and Lemma \ref{lem:dot-length} (i),
(B) is obtained by applying the fifth case of (\ref{eq:6807mxs1xye4}) by $|\bar{\PREF}^1_{F, j}(f_j(s'))| = |\bar{\PREF}^1_{F, j}| = 1$, (\ref{eq:da6ansmwz4bu}) and (\ref{eq:kbl91g7jnwbk}),
and (C) follows from (\ref{eq:4c3d7d353z5c}) and Lemma \ref{lem:tilde-equal} (i).
This shows
\begin{equation}
\label{eq:aennbnjwrgv8}
\bar{\PREF}^2_{\dot{F}, j} \subseteq \{\bar{a}_{F, \trans_j(s')}1\}.
\end{equation}
Finally, we obtain
\begin{eqnarray}
\PREF^2_{\dot{F}, j}
&\overset{(\mathrm{A})}{=}& \bar{\PREF}^2_{\dot{F}, j} \cup \Big( \bigcup_{s'' \in \assign_{F, j}(\lambda)} \PREF^2_{\dot{F}, \dot{\trans}_j(s'')} \Big)\\
&\overset{(\mathrm{B})}{=}& \bar{\PREF}^2_{\dot{F}, j} \cup \PREF^2_{\dot{F}, \dot{\trans}_j(s')} \\
&\overset{(\mathrm{C})}{\subseteq}& \{a_{F, \trans_j(s')}0, a_{F, \trans_j(s')}1, \bar{a}_{F, \trans_j(s')}1\}\\
&\overset{(\mathrm{D})}{=}& \{a_{F, j}0, a_{F, j}1, \bar{a}_{F, j}1\}
\end{eqnarray}
as desired, where
(A) follows from Lemma \ref{lem:pref-sum} (i),
(B) follows from (\ref{eq:4c3d7d353z5c}),
(C) follows from (\ref{eq:xrwsnm4lxs1y}) and (\ref{eq:aennbnjwrgv8}),
and (D) follows since $a_{F, \trans_j(s')} = a_{F, j}$ by (\ref{eq:4c3d7d353z5c}) and the first case of (\ref{v47umw9fprm9}).

\item The case $|\bar{\PREF}^1_{F, i}(f_i(s))| = 1, |\bar{\PREF}^1_{F, j}| = 2$:
We show $\pmb{c} \in \PREF^2_{F, j}$ for an arbitrarily fixed $\pmb{c} = c_1c_2 \in \PREF^2_{\dot{F}, j}$.

We have
\begin{eqnarray}
|\assign_{F, j}(\lambda)|
&=& \frac{2|\assign_{F, j}(\lambda)|}{2}\\
&\overset{(\mathrm{A})}{\leq}& \frac{\sum_{s' \in \assign_{F, j}(\lambda)} |\PREF^2_{F, \trans_j(s')}|}{2} \\
&\overset{(\mathrm{B})}{=}&  \frac{|\PREF^2_{F, j}| - |\bar{\PREF}^2_{F, j}|}{2}\\
&\overset{(\mathrm{C})}{\leq}&  \frac{|\PREF^2_{F, j}| - |\bar{\PREF}^1_{F, j}|}{2}\\
&\overset{(\mathrm{D})}{\leq}& \frac{3 - |\bar{\PREF}^1_{F, j}|}{2}\\
&\overset{(\mathrm{E})}{=}& \frac{3 - 2}{2}\\
&<& 1,
\end{eqnarray}
where
(A) follows since $|\PREF^2_{F, \trans_j(s')}| \geq 2$ for any $s' \in \assign_{F, j}(\lambda)$ by Lemma \ref{lem:propertyT1} (i),
(B) follows from Lemma \ref{lem:pref-sum} (ii),
(C) follows from Lemma \ref{lem:pref-inc2} (ii) (b),
(D) follows from Lemma \ref{cor:pref-sum} and $|\bar{\PREF}^1_{F, i}(f_i(s))| = 1$,
and (E) follows from $|\bar{\PREF}^1_{F, j}| = 2$.
This shows
\begin{equation}
\label{eq:mh0pu0yfj9pd}
\assign_{F, j}(\lambda) = \emptyset.
\end{equation}

By $\pmb{c} \in \PREF^2_{\dot{F}, j}$, there exists $\pmb{x} = x_1x_2\ldots x_{n} \in \mathcal{S}^{+}$ such that
\begin{equation}
\label{eq:h85un1rwxt1s}
\dot{f}^{\ast}_j(\pmb{x}) \succeq \pmb{c}.
\end{equation}
Then we have 
\begin{equation}
\dot{f}_j(x_1) = \dot{\gamma}(s_1)\dot{\gamma}(s_2)\ldots \dot{\gamma}(s_{\rho-1})\dot{\gamma}(s_{\rho})
\succeq \dot{\gamma}(s_1)
\overset{(\mathrm{A})}{=} \gamma(s_1)\label{eq:6sl0027tsfnj},
\end{equation}
where $\gamma(s_1)\gamma(s_2)\ldots \gamma(s_{\rho-1})\gamma(s_{\rho})$ is the $\gamma$-decomposition $f_j(x_1)$ and (A) follows from $|\PREF^2_{F, j}| \geq 3$ and the second case of (\ref{eq:6807mxs1xye4}).

By (\ref{eq:mh0pu0yfj9pd}) and Lemma \ref{lem:dot-length} (i), it holds that $|\gamma(s_1)| \geq 1$.
We consider the following two cases separately: the case $|\gamma(s_1)| = 1$ and the case $|\gamma(s_1)| \geq 2$.
\begin{itemize}
\item The case $|\gamma(s_1)| = 1$:
By (\ref{eq:h85un1rwxt1s}) and (\ref{eq:6sl0027tsfnj}), we have
\begin{equation}
\label{eq:6q9sy5juzyvo}
f_j(s_1) = \gamma(s_1) = c_1.
\end{equation}
We obtain
\begin{equation}
\PREF^2_{F, j} \overset{(\mathrm{A})}{\supseteq} \bar{\PREF}^2_{F, j} \overset{(\mathrm{B})}{\supseteq} c_1\PREF^1_{F, j}(c_1) \overset{(\mathrm{C})}{=} c_1\PREF^1_{F, j}(f_j(s_1)) \overset{(\mathrm{D})}{\supseteq} c_1\PREF^1_{F, \trans_j(s_1)} \overset{(\mathrm{E})}{=} c_1\{0, 1\} \owns c_1c_2 = \pmb{c}
\end{equation}
as desired, where
(A) follows from Lemma \ref{lem:pref-sum} (i),
(B) follows from Lemma \ref{lem:pref-sum} (iii),
(C) follows from (\ref{eq:6q9sy5juzyvo}),
(D) follows from Lemma \ref{lem:pref-sum} (i),
and (E) follows from $F \in \mathscr{F}_1$.

\item The case $|\gamma(s_1)| \geq 2$: By (\ref{eq:h85un1rwxt1s}) and (\ref{eq:6sl0027tsfnj}), we have $f^{\ast}_j(\pmb{x}) \succeq \gamma(s_1) \succeq \pmb{c}$, which leads to $\pmb{c} \in \PREF^2_{F, j}$.
\end{itemize}
\end{itemize}

(Proof of (ii)):
We consider the following two cases separately: the case $|\bar{\PREF}^1_{F, i}(f_i(s))| = 0$ and the $|\bar{\PREF}^1_{F, i}(f_i(s))| \geq	 1$.
\begin{itemize}
\item The case $|\bar{\PREF}^1_{F, i}(f_i(s))| = 0$:
We have
\begin{eqnarray}
|\bar{\PREF}^1_{F, i}(f_i(s))|  = 0
&\overset{(\mathrm{A})}{\iff}& \bar{\PREF}^0_{F, i}(f_i(s)) = \emptyset\\
&\overset{(\mathrm{B})}{\iff}& {}^{\forall} s' \in \mathcal{S}; f_i(s) \not\prec f_i(s')\\
&\overset{(\mathrm{C})}{\iff}& {}^{\forall} s' \in \mathcal{S}; \dot{f}_i(s) \not\prec \dot{f}_i(s')\\
&\overset{(\mathrm{D})}{\iff}& \bar{\PREF}^0_{\dot{F}, i}(\dot{f}_i(s)) = \emptyset\\
 &\overset{(\mathrm{E})}{\iff}& \bar{\PREF}^2_{\dot{F}, i}(\dot{f}_i(s)) = \emptyset
\end{eqnarray}
as desired, 
where
(A) follows from Lemma \ref{lem:pref-inc2} (ii) (a),
(B) follows from Lemma \ref{lem:leaf} (i),
(C) follows from Lemma \ref{lem:tilde-equal} (iii),
(D) follows from Lemma \ref{lem:leaf} (i),
and (E) follows from Lemma \ref{lem:pref-inc2} (ii) (a).

\item The case $|\bar{\PREF}^1_{F, i}(f_i(s))| \geq 1$:
Choose $\pmb{x} = x_1x_2\ldots x_n \in \mathcal{S}^{+}$ such that $\dot{f}^{\ast}_i(\pmb{x}) \succeq \dot{f}_i(s)$,
and $\dot{f}_i(x_1) \succ \dot{f}_i(s)$ arbitrary and let $\gamma(s_1) \gamma(s_2)\ldots  \gamma(s_{\rho'})$ be the $\gamma$-decomposition of $f_i(x_1)$.
Then by $\dot{f}_i(x_1) \succ \dot{f}_i(s)$, there exists an integer $\rho$ such that $\rho < \rho'$ and $f_i(s) = \gamma(s_1) \gamma(s_2)\ldots  \gamma(s_{\rho})$.
We have
\begin{eqnarray}
\dot{f}^{\ast}_i(\pmb{x})
&\succeq& \dot{f}_i(x_1)\\
&=& \dot{\gamma}(s_1) \dot{\gamma}(s_2)\ldots  \dot{\gamma}(s_{\rho'})\\
&=& \dot{f}_i(s) \dot{\gamma}(s_{\rho+1}) \ldots \dot{\gamma}(s_{\rho'})\\
&\succeq& \dot{f}_i(s) \dot{\gamma}(s_{\rho+1})\\
&\overset{(\mathrm{A})}{\succeq}& \dot{f}_i(s) \dot{g_1}\dot{g_2} \label{eq:2g2hiku0bai7},
\end{eqnarray}
where $\dot{\gamma}(s_{\rho+1}) = \dot{g}_1\dot{g}_2\ldots \dot{g}_l$,
and (A) follows since $|\dot{\gamma}(s_{\rho+1})| = |\gamma(s_{\rho+1})| \geq 2$ by Lemma \ref{lem:dot-length} (ii) and Lemma \ref{lem:tilde-equal} (i).
Therefore, the set $\bar{\PREF}^2_{\dot{F}, i}(\dot{f}_i(s))$ is included in the set of all possible  sequences as $\dot{g}_1\dot{g}_2 \in \mathcal{C}^2$.
We consider what sequences are possible as $\dot{g}_1\dot{g}_2 \in \mathcal{C}^2$ for the following three cases: the case $|\PREF^2_{F, j}| = 2$, the case $|\PREF^2_{F, j}| \geq 3, |\bar{\PREF}^1_{F, j}| = 1$, and the case $|\PREF^2_{F, j}| \geq 3, |\bar{\PREF}^1_{F, j}| = 2$.
\begin{itemize}
\item The case $|\PREF^2_{F, j}| = 2$:
\begin{itemize}
\item The case $|\bar{\PREF}^1_{F, i}(f_i(s))| = 2$: We have $\dot{g}_1\dot{g}_2 \subseteq \{\bar{a}_{F, j}0, \bar{a}_{F, j}1\}$ applying the third case of (\ref{eq:6807mxs1xye4}).
\item The case $|\bar{\PREF}^1_{F, i}(f_i(s))| = 1$:
By $|\PREF^2_{F, j}| = 2$ and Lemma \ref{lem:propertyT1} (ii) (b), we have $|\PREF^2_{F, j}| = \{0a, 1b\}$ for some $a, b \in \mathcal{C}$.
Thus, we have $|\bar{\PREF}^1_{F, j}| = |\{0, 1\}| = 2$ applying Lemma \ref{lem:pref-inc} (ii).
Hence, we obtain $\dot{g}_1\dot{g}_2 = \bar{a}_{F, j}1$ applying the fifth case of (\ref{eq:6807mxs1xye4}).
\end{itemize}
These show $\bar{\PREF}^2_{\dot{F}, i}(\dot{f}_i(s)) \subseteq \{\bar{a}_{F, j}0, \bar{a}_{F, j}1\}$ as desired.

\item The case $|\PREF^2_{F, j}| \geq 3$:
Then we have $|\bar{\PREF}^1_{F, i}(f_i(s))| \leq 1$ by Lemma \ref{cor:pref-sum}.
Combining this with $|\bar{\PREF}^1_{F, i}(f_i(s))| \geq 1$, we obtain
\begin{equation}
\label{eq:vu1fko998ipv}
|\bar{\PREF}^1_{F, i}(f_i(s))| = 1.
\end{equation}

\begin{itemize}
\item The case $|\bar{\PREF}^1_{F, j}| = 1$: 
We obtain $\dot{g}_1\dot{g}_2 =  \bar{a}_{F, j}0$ applying the fourth case of (\ref{eq:6807mxs1xye4}) by (\ref{eq:vu1fko998ipv}) and $|\bar{\PREF}^1_{F, j}| = 1$. 
 This shows $\bar{\PREF}^2_{\dot{F}, i}(\dot{f}_i(s)) \subseteq \{\bar{a}_{F, j}0\}$ as desired.
 
\item The case $|\bar{\PREF}^1_{F, j}| = 2$: 
We obtain $\dot{g}_1\dot{g}_2 =  g_1g_2$ by the sixth case of (\ref{eq:6807mxs1xye4}) by (\ref{eq:vu1fko998ipv}), $|\bar{\PREF}^1_{F, j}| = 2$, and $|\PREF^2_{F, j}| \geq 3$.
This shows $\bar{\PREF}^2_{\dot{F}, i}(\dot{f}_i(s)) \subseteq \bar{\PREF}^2_{F, i}(f_i(s))$ as desired because $g_1g_2 \in \bar{\PREF}^2_{F, i}(f_i(s))$ by Lemma \ref{lem:dot-length} (iii).

\end{itemize}
\end{itemize}

\end{itemize}
\end{proof}

\subsection{Proof of Lemma \ref{lem:tilde-pref2}}
\label{subsec:proof-tilde-pref2}

\begin{proof}[Proof of Lemma \ref{lem:tilde-pref2}]
(Proof of (i)): Assume $|\PREF^2_{F, i}| = 2$. Then we have $\PREF^2_{\dot{F}, i} = \{a_{F, i}0, a_{F, i}1\}$ by Lemma \ref{lem:tilde-pref} (i) (a).
Hence, we have $\PREF^1_{\dot{F}, i} = \{a_{F, i}\}$ by Lemma \ref{lem:pref-inc} (i).
Therefore, by (\ref{eq:emksr78o7m54}), we obtain $d_{\dot{F}, i} = a_{F, i}$ as desired.

(Proof of (ii)):
Assume $s \neq s'$ and $\dot{f}_i(s) = \dot{f}_i(s')$.
Then since $f_i(s) = f_i(s')$ by Lemma \ref{lem:tilde-equal} (iii), 
we have
\begin{equation}
\label{eq:cfgmg3qbdv27}
|\PREF^2_{F, \trans_i(s)}| = |\PREF^2_{F, \trans_i(s')}| = 2
\end{equation}
applying Lemma \ref{lem:propertyT1} (iii).
Hence, by (i) of this lemma, we obtain
\begin{equation}
\label{eq:jegocq7j08yo}
d_{\dot{F}, \trans_i(s)} = a_{F, \trans_i(s)}, \quad d_{\dot{F}, \trans_i(s')} = a_{F, \trans_i(s')}.
\end{equation}
Also, by (\ref{eq:cfgmg3qbdv27}) and Lemma \ref{lem:propertyT1} (ii) (a), we have $\assign_{F, \trans_i(s)}(\lambda) = \assign_{F, \trans_i(s')}(\lambda) = \emptyset$,
in particular,
\begin{equation}
\label{eq:2pn1phdmp1tg}
|\assign_{F, \trans_i(s)}(\lambda)| \neq 1, \quad |\assign_{F, \trans_i(s')}(\lambda)| \neq 1.
\end{equation}

Now we show $\PREF^2_{F, \trans_i(s)} \owns 0a_{F, \trans_i(s)}$ considering the following two cases: the case $\PREF^2_{F, \trans_i(s)} \owns 00$ and the case $\PREF^2_{F, \trans_i(s)} \not\owns 00$.
\begin{itemize}
\item The case $\PREF^2_{F, \trans_i(s)} \owns 00$: 
By (\ref{eq:2pn1phdmp1tg}) and the second case of (\ref{v47umw9fprm9}), we have $a_{F, \trans_i(s)} = 0$ and thus $\PREF^2_{F, \trans_i(s)} \owns 00 = 0a_{F, \trans_i(s)}$.
\item The case $\PREF^2_{F, \trans_i(s)} \not\owns 00$: 
By Lemma \ref{lem:propertyT1} (ii) (b), there exists $b \in \mathcal{C}$ such that 
\begin{equation}
\PREF^2_{F, \trans_i(s)} \owns 0b \overset{(\mathrm{A})}{=} 01 \overset{(\mathrm{B})}{=} 0a_{F, \trans_i(s)},
\end{equation}
where
(A) follows from $\PREF^2_{F, \trans_i(s)} \not\owns 00$,
and (B) follows from (\ref{eq:2pn1phdmp1tg}), $\PREF^2_{F, \trans_i(s)} \not\owns 00$, and the third case of (\ref{v47umw9fprm9}).
\end{itemize}
Therefore, we conclude that $\PREF^2_{F, \trans_i(s)} \owns 0a_{F, \trans_i(s)}$.
By the same argument, we also have $\PREF^2_{F, \trans_i(s')} \owns 0a_{F, \trans_i(s')}$. 
Consequently, we have
\begin{equation}
\label{eq:0p54m80pqmse}
\{0a_{F, \trans_i(s)}\} \cap \{0a_{F, \trans_i(s')}\} \subseteq \PREF^2_{F, \trans_i(s)} \cap \PREF^2_{F, \trans_i(s')} \overset{(\mathrm{A})}{=} \emptyset,
\end{equation}
where (A) follows from $F \in \mathscr{F}_{2\hdec}$.
This shows
\begin{equation}
\label{eq:zrz0v2tzgm8q}
a_{F, \trans_i(s)} \neq a_{F, \trans_i(s')}.
\end{equation}
Combining (\ref{eq:jegocq7j08yo}) and (\ref{eq:zrz0v2tzgm8q}), we obtain the desired result.
\end{proof}

\subsection{Proof of Lemma \ref{lem:transA-arg4}}
\label{subsec:proof-transA-arg4}

 To prove Lemma \ref{lem:transA-arg4}, we prove Lemmas \ref{lem:transA-arg2} and \ref{lem:transA-arg1} as follows.
 
\begin{lemma}
\label{lem:transA-arg2}
For any $F \in \mathscr{F}_1$ and $i \in [F]$, the mapping $\widehat{\dot{f}_i}$ is injective.
\end{lemma}

\begin{proof}[Proof of Lemma \ref{lem:transA-arg2}]
Choose $s, s' \in \mathcal{S}$ such that $\widehat{\dot{f}_i}(s) = \widehat{\dot{f}_i}(s')$ arbitrarily. We show $s = s'$.

We have
\begin{equation}
\label{eq:gk92ksr1karf}
\dot{f}_i(s) d_{\dot{F}, \dot{\trans}_i(s)} \overset{(\mathrm{A})}{=} d_{\dot{F}, i} \widehat{\dot{f}_i}(s) \overset{(\mathrm{B})}{=} d_{\dot{F}, i} \widehat{\dot{f}_i}(s') \overset{(\mathrm{C})}{=} \dot{f}_i(s') d_{\dot{F}, \dot{\trans}_i(s')},
\end{equation}
where
(A) follows from Lemma \ref{lem:rotation} (i),
(B) follows directly from $\widehat{\dot{f}_i}(s) = \widehat{\dot{f}_i}(s')$,
and (C) follows from Lemma \ref{lem:rotation} (i).

Also, we have
\begin{equation}
\label{eq:k0r7bthzqnep}
|d_{\dot{F}, \dot{\trans}_i(s)}| = |d_{\dot{F}, \dot{\trans}_i(s')}|
\end{equation}
because if we assume the contrary, that is, $|d_{\dot{F}, \dot{\trans}_i(s)}| = 1$ and $|d_{\dot{F}, \dot{\trans}_i(s')}| = 0$ by symmetry, then by (\ref{eq:gk92ksr1karf}), we have $\dot{f}_i(s)d_{\dot{F}, \dot{\trans}_i(s)} = \dot{f}_i(s')$, which conflicts with Lemma \ref{lem:propertyT1} (v).

By (\ref{eq:gk92ksr1karf}) and (\ref{eq:k0r7bthzqnep}), we obtain
$\dot{f}_i(s) = \dot{f}_i(s')$ and $d_{\dot{F}, \dot{\trans}_i(s)} = d_{\dot{F}, \dot{\trans}_i(s')}$.
Hence, we obtain $s = s'$ as desired applying the contraposition of Lemma \ref{lem:tilde-pref2} (ii).
\end{proof}

\begin{lemma}
\label{lem:transA-arg1}
For any $F \in \mathscr{F}_1, i \in [F]$, and $s \in \mathcal{S}$,
if $\bar{\PREF}^0_{F, i}(f_i(s)) = \emptyset$ or $\trans_i(s) \in \mathcal{M}_F$, then $\bar{\PREF}^0_{\widehat{\dot{F}}, i}(\widehat{\dot{f}_i}(s)) = \emptyset$.
\end{lemma}

\begin{proof}[Proof of Lemma \ref{lem:transA-arg1}]
We assume that $\bar{\PREF}^0_{F, i}(f_i(s)) = \emptyset$ or $\trans_i(s) \in \mathcal{M}_F$ holds and prove by contradiction assuming $\bar{\PREF}^0_{\widehat{\dot{F}}, i}(\widehat{\dot{f}_i}(s)) \neq \emptyset$.
Then by Lemma \ref{lem:leaf} (i), there exist $s' \in \mathcal{S} \setminus \{s\}$ and $c \in \mathcal{C}$ such that
\begin{equation}
\label{eq:gk2aqce0gk4s}
\widehat{\dot{f}_i}(s) c \preceq \widehat{\dot{f}_i}(s').
\end{equation}
Thus, we have
\begin{equation}
\label{eq:d31iwliixt21}
\dot{f}_i(s) d_{\dot{F}, \dot{\trans}_i(s)} c 
\overset{(\mathrm{A})}{=} d_{\dot{F}, i} \widehat{\dot{f}_i}(s) c 
\overset{(\mathrm{B})}{\preceq} d_{\dot{F}, i} \widehat{\dot{f}_i}(s') 
\overset{(\mathrm{C})}{=}  \dot{f}_i(s') d_{\dot{F}, \dot{\trans}_i(s')},
\end{equation}
where
(A) follows from Lemma \ref{lem:rotation} (i),
(B) follows from (\ref{eq:gk2aqce0gk4s}),
and (C) follows from Lemma \ref{lem:rotation} (i).

We consider the following two cases separately: the case $\bar{\PREF}^0_{F, i}(f_i(s)) = \emptyset$ and the case $\trans_i(s) \in \mathcal{M}_F$.

\begin{itemize}
\item The case $\bar{\PREF}^0_{F, i}(f_i(s)) = \emptyset$:
We have
\begin{equation}
\label{eq:cuxzudo2f9uh}
|\bar{\PREF}^0_{\dot{F}, i}(\dot{f}_i(s))|
\overset{(\mathrm{A})}{\leq} |\bar{\PREF}^2_{\dot{F}, i}(\dot{f}_i(s))|
\overset{(\mathrm{B})}{=} 0,
\end{equation}
where
(A) follows from Lemma \ref{lem:pref-inc2} (ii) (b),
and (B) follows from the first case of (\ref{eq:v5z4ipil23dc}) because $\bar{\PREF}^1_{F, i}(f_i(s)) = \emptyset$ holds by $\bar{\PREF}^0_{F, i}(f_i(s)) = \emptyset$ and Lemma \ref{lem:pref-inc2} (ii) (a).

Also, we have
\begin{equation}
\label{eq:1z4yjbb9brsn}
|\dot{f}_i(s)| \overset{(\mathrm{A})}{\leq} |\dot{f}_i(s')| + |d_{\dot{F}, \dot{\trans}_i(s')}| - |d_{\dot{F}, \dot{\trans}_i(s)}| - |c| \overset{(\mathrm{B})}{\leq} |\dot{f}_i(s')|,
\end{equation}
where
(A) follows from (\ref{eq:d31iwliixt21}),
and (B) follows from $|d_{\dot{F}, \dot{\trans}_i(s')}| \leq 1$, $|d_{\dot{F}, \dot{\trans}_i(s)}| \geq 0$, and $|c| = 1$.

In fact, the equalities hold in (\ref{eq:1z4yjbb9brsn}), that is, we have
\begin{equation}
\label{eq:adqlqdbcmzu1}
|\dot{f}_i(s)| = |\dot{f}_i(s')|
\end{equation}
because if we assume $|\dot{f}_i(s)| < |\dot{f}_i(s')|$, then we have $\dot{f}_i(s) \prec \dot{f}_i(s')$ by (\ref{eq:d31iwliixt21}), which conflicts with (\ref{eq:cuxzudo2f9uh}) and Lemma \ref{lem:leaf} (i).

By (\ref{eq:d31iwliixt21}) and (\ref{eq:adqlqdbcmzu1}), we obtain
\begin{equation}
\label{eq:2d37xkllascw}
\dot{f}_i(s) = \dot{f}_i(s').
\end{equation}
Hence, applying Lemma \ref{lem:tilde-pref2} (ii), we have
$d_{\dot{F}, \dot{\trans}_i(s)} = a_{F, \trans_i(s)}$ and $d_{\dot{F}, \dot{\trans}_i(s')} = a_{F, \trans_i(s')}$.
In particular,
\begin{equation}
\label{eq:pi1nxev6u8ir}
|d_{\dot{F}, \dot{\trans}_i(s)}| = |d_{\dot{F}, \dot{\trans}_i(s')}| = 1.
\end{equation}
Thus, we obtain
\begin{equation}
|\dot{f}_i(s)| + 2 \overset{(\mathrm{A})}{=} |\dot{f}_i(s) d_{\dot{F}, \dot{\trans}_i(s)} c| \overset{(\mathrm{B})}{\leq} |\dot{f}_i(s') d_{\dot{F}, \dot{\trans}_i(s')}| \overset{(\mathrm{C})}{=}  |\dot{f}_i(s')| + 1 \overset{(\mathrm{D})}{=} |\dot{f}_i(s)| + 1,
\end{equation}
where
(A) follows from (\ref{eq:pi1nxev6u8ir}),
(B) follows from (\ref{eq:d31iwliixt21}),
(C) follows from (\ref{eq:pi1nxev6u8ir}),
and (D) follows from (\ref{eq:2d37xkllascw}).
This is a contradiction.

\item The case $\trans_i(s) \in \mathcal{M}_F$:
By Lemma \ref{lem:tilde-pref2} (i), we have
\begin{equation}
\label{eq:1vov21om8dxh}
d_{\dot{F}, \dot{\trans}_i(s)} = a_{F, \trans_i(s)}.
\end{equation}
Substituting (\ref{eq:1vov21om8dxh}) for (\ref{eq:d31iwliixt21}), we obtain
\begin{equation}
\label{eq:m8qc97am01if}
\dot{f}_i(s) a_{F, \trans_i(s)} c \preceq \dot{f}_i(s') d_{\dot{F}, \dot{\trans}_i(s')}.
\end{equation}

Also, we have
\begin{equation}
\label{eq:v8pjj512ywyx}
|\dot{f}_i(s)| + 1 = |\dot{f}_i(s)| + |a_{F, \trans_i(s)}|
\overset{(\mathrm{A})}{\leq} |\dot{f}_i(s')| + |d_{\dot{F}, \dot{\trans}_i(s')}| - |c|
\overset{(\mathrm{B})}{\leq} |\dot{f}_i(s')|,
\end{equation}
where
(A) follows from (\ref{eq:m8qc97am01if}),
and (B) follows from $|d_{\dot{F}, \dot{\trans}_i(s')}| \leq 1$ and $|c| = 1$.

By (\ref{eq:m8qc97am01if}) and (\ref{eq:v8pjj512ywyx}), we have $\dot{f}_i(s) a_{F, \trans_i(s)} \preceq \dot{f}_i(s')$, which leads to $\bar{\PREF}^1_{\dot{F}, i}(\dot{f}_i(s)) \owns a_{F, \trans_i(s)}$.
Hence, applying Lemma \ref{lem:pref-inc} (ii), we have
\begin{equation}
\label{eq:u8loznhgxuc0}
\bar{\PREF}^2_{\dot{F}, i}(\dot{f}_i(s)) \owns a_{F, \trans_i(s)} c'
\end{equation}
for some $c' \in \mathcal{C}$.
On the other hand, by $\trans_i(s) \in \mathcal{M}_F$ and Lemma \ref{lem:tilde-pref} (i) (a),
we have
\begin{equation}
\label{eq:m6qw4zvpbmvy}
\PREF^2_{\dot{F}, \dot{\trans}_i(s)} = \{a_{F, \trans_i(s)}0, a_{F, \trans_i(s)}1\}.
\end{equation}
By (\ref{eq:u8loznhgxuc0}) and (\ref{eq:m6qw4zvpbmvy}),
we obtain $\PREF^2_{\dot{F}, \dot{\trans}_i(s)} \cap \bar{\PREF}^2_{\dot{F}, i}(\dot{f}_i(s)) \neq \emptyset$.
Hence, we have $\dot{F} \not\in \mathscr{F}_{2\hdec}$, which conflicts with the proof of Lemma \ref{lem:transA-preserve}.
\end{itemize}
\end{proof}

\begin{proof}[Proof of Lemma \ref{lem:transA-arg4}]
Applying Lemma \ref{lem:transA-preserve} in a repetitive manner, we have
\begin{equation}
\label{eq:eobi3f9itqe1}
F^{(0)}, F^{(1)}, \ldots, F^{(t)}, F^{(t+1)}, \ldots, F^{(t')} \in \mathscr{F}_1
\end{equation}
and
\begin{equation}
\label{eq:33t3kbhtphj6}
L(F) = L(F^{(0)}) = L(F^{(1)}) = \cdots = L(F^{(t)}) = L(F^{(t+1)}) = \cdots = L(F^{(t')}).
\end{equation}

We prove Lemma \ref{lem:transA-arg4} by contradiction assuming that there exists $p \in \mathcal{M}_{F^{(t)}} \cap \mathcal{M}_{F^{(t')}}$.
By $\kernel_{F} = |F|$, there exist $i \in [F]$ and $s \in \mathcal{S}$ such that $\trans_i(s) = p$.
By (\ref{eq:g8ovoqynvwxb}) and (\ref{eq:bpsfrfe2v0zi}), we have $\trans_i^{(t)}(s) = \trans_i^{(t')}(s) = p$ and 
\begin{eqnarray}
\lefteqn{\trans_i^{(t)}(s) = p \in \mathcal{M}_{F^{(t)}}} \nonumber\\
&\overset{(\mathrm{A})}{\implies}& \bar{\PREF}^0_{F^{(t+1)}, i}(f_i^{(t+1)}(s)) = \emptyset.\\
&\overset{(\mathrm{A})}{\implies}& \bar{\PREF}^0_{F^{(t+2)}, i}(f_i^{(t+1)}(s)) = \emptyset.\\
&\overset{(\mathrm{A})}{\implies}& \cdots\\
&\overset{(\mathrm{A})}{\implies}& \bar{\PREF}^0_{F^{(t')}, i}(f_i^{(t')}(s)) = \emptyset, \label{eq:5h02t4n3x37f}
\end{eqnarray}
where (A)s follow from (\ref{eq:eobi3f9itqe1}) and Lemma \ref{lem:transA-arg1}.
Applying Lemma \ref{lem:transA-arg2} to $F^{(t'-1)}$, we see that $f_i^{(t')}(s)$ is injective, in particular,
\begin{equation}
\label{eq:q9fo54xy0ojp}
|\assign_{F^{(t')}, i}(f^{(t')}_i(s))| = 1.
\end{equation}
By (\ref{eq:5h02t4n3x37f}) and (\ref{eq:q9fo54xy0ojp}), we obtain $|\PREF^2_{F^{(t')}, p}| = |\PREF^2_{F^{(t')}, \trans^{(t')}_i(s)}| = 4$ applying Lemma \ref{lem:complete-leaf}, which conflicts with $p \in \mathcal{M}_{F^{(t')}}$.
\end{proof}

\subsection{Proof of Lemma \ref{lem:tilde2-equal} (iii)}
\label{subsec:proof-tilde2-equal}

We can prove Lemma \ref{lem:tilde2-equal} (iii) in a similar way to prove Lemma \ref{lem:tilde-equal} (iii) by using the following Lemma \ref{lem:tilde2-equal0} instead of Lemma \ref{lem:tilde-equal0}.

\begin{lemma}
\label{lem:tilde2-equal0}
Let $F \in \mathscr{F}_2$, $i \in [F]$, and $s, s' \in \mathcal{S}$, and let $\gamma(s_1)\gamma(s_2)\ldots \gamma(s_{\rho})$ (resp. $\gamma(s'_1)\gamma(s'_2)\ldots \gamma(s'_{\rho'})$) be the $\gamma$-decomposition of $f_i(s)$ (resp. $f_i(s')$).
For any $r = 1, 2, \ldots, m \coloneqq \min\{\rho, \rho'\}$, 
if one of the following conditions (a) and (b) holds, then $\gamma(s_r) = \gamma(s'_r) \iff \ddot{\gamma}(s_{r}) = \ddot{\gamma}(s'_r)$:
\begin{enumerate}[(a)]
\item $r = 1$.
\item $r \geq 2$ and $s_{r-1} = s'_{r-1}$.
\end{enumerate}
\end{lemma}

\begin{proof}[Proof of Lemma \ref{lem:tilde-equal0}]
Assume that (a) or (b) holds.

($\implies$) Directly from (\ref{eq:oezl1k2insm5}).

($\impliedby$) We prove the contraposition.
Namely, we prove $\ddot{\gamma}(s_r) \neq \ddot{\gamma}(s'_r)$ assuming $\gamma(s_{r}) \neq \gamma(s'_r)$.
Put $\gamma(s_{r}) = g_1g_2\ldots g_l$ and $\gamma(s'_r) = g'_1g'_2\ldots g'_{l'}$.
We consider the following two cases separately: the case $|\gamma(s_r)| \neq |\gamma(s'_r)|$ and the case $|\gamma(s_r)| = |\gamma(s'_r)|$.
\begin{itemize}
\item The case $|\gamma(s_r)| \neq |\gamma(s'_r)|$: We have
\begin{equation}
|\ddot{\gamma}(s_r)| \overset{(\mathrm{A})}{=} |\gamma(s_r)| \overset{(\mathrm{B})}{\neq} |\gamma(s'_r)| \overset{(\mathrm{C})}{=} |\ddot{\gamma}(s'_r)|,
\end{equation}
where
(A) follows from Lemma \ref{lem:tilde2-equal} (i),
(B) follows from the assumption,
and (C) follows from Lemma \ref{lem:tilde2-equal} (i).
This shows $\ddot{\gamma}(s_r) \neq \ddot{\gamma}(s'_r)$.

\item The case $|\gamma(s_r)| = |\gamma(s'_r)|$:
If $|\gamma(s_r)| = |\gamma(s'_r)| \geq 3$ and $g_3g_4\ldots g_{l} \neq g'_3g'_4\ldots g'_{l'}$, then we obtain $\ddot{\gamma}(s_r) \neq \ddot{\gamma}(s'_r)$ directly from (\ref{eq:oezl1k2insm5}).
Thus, we assume 
\begin{equation}
\label{eq:5258bln24k8f}
g_j \neq g'_j \,\,\text{for some}\,\, 1 \leq j \leq \min\{2, |\gamma(s_r)|\}. 
\end{equation}

Now we show that the condition (a) is necessarily holds by contradiction assuming that the condition (a) does not hold and the condition (b) holds.
Then we have $|\gamma(s_r)| = |\gamma(s'_r)| \geq 2$ by Lemma \ref{lem:dot-length} (ii) and 
we have $g_1g_2 \in \bar{\PREF}^2_{F, i}(f_i(s_{r-1}))$ and $g'_1g'_2 \in \bar{\PREF}^2_{F, i}(f_i(s'_{r-1}))$ by Lemma \ref{lem:dot-length} (iii).
Since $s_{r-1} = s'_{r-1}$ by the condition (b), we have
\begin{equation}
\label{eq:3ycapemtus74}
\{g_1g_2, g'_1g'_2\} \subseteq \bar{\PREF}^2_{F, i}(f_i(s_{r-1})).
\end{equation}
Therefore, we have
\begin{equation}
|\{g_1g_2, g'_1g'_2\}| \overset{(\mathrm{A})}{\leq} |\bar{\PREF}^2_{F, i}(f_i(s_{r-1}))| \overset{(\mathrm{B})}{\leq} |\PREF^2_{F, i}(f_i(s_{r-1}))| - |\PREF^2_{F, \trans_i(s_{r-1})}| \overset{(\mathrm{C})}{\leq} 4 - 3 = 1,
\end{equation}
where
(A) follows from (\ref{eq:3ycapemtus74}),
(B) follows from Lemma \ref{lem:pref-sum} (ii),
and (C) follows from $F \in \mathscr{F}_{2}$.
This leads to $g_1g_2 = g'_1g'_2$, which conflicts with (\ref{eq:5258bln24k8f}).
Therefore, the condition (a), that is, $r = 1$ holds.

We consider the following two cases separately: the case $|\PREF^2_{F, i}| = 4$ and the case $|\PREF^2_{F, i}| = 3$.
\begin{itemize}
\item The case $|\PREF^2_{F, i}| = 4$: We obtain
\begin{equation}
\ddot{\gamma}(s_1) \overset{(\mathrm{A})}{=} \gamma(s_1)  \overset{(\mathrm{B})}{\neq} \gamma(s'_1) \overset{(\mathrm{C})}{=}  \ddot{\gamma}(s'_1)
\end{equation}
as desired, where
(A) follows from $|\PREF^2_{F, i}| = 4$ and the first case of (\ref{eq:oezl1k2insm5}),
(B) follows from (\ref{eq:5258bln24k8f}),
and (C) follows from $|\PREF^2_{F, i}| = 4$ and the first case of (\ref{eq:oezl1k2insm5}).

\item The case $|\PREF^2_{F, i}| = 3$:
We first prove
\begin{equation}
\label{eq:h3kyk0pw4u1y}
|\gamma(s_1)| = |\gamma(s'_1)| \geq 2
\end{equation}
by assuming the contrary $|\gamma(s_1)| = |\gamma(s'_1)| = 1$.
Then by (\ref{eq:5258bln24k8f}), we may assume $\gamma(s_1) = 0$ and $\gamma(s'_1) = 1$ without loss of generality.
Hence, we have
\begin{equation}
\PREF^2_{F, i}
\overset{(\mathrm{A})}{\supseteq} \bar{\PREF}^2_{F, i}
\overset{(\mathrm{B})}{=} 0\PREF^1_{F, i}(0) \cup 1\PREF^1_{F, i}(1)
\overset{(\mathrm{C})}{\supseteq} 0\PREF^1_{F, \trans_i(s_1)} \cup 1\PREF^1_{F, \trans_i(s'_1)}
\overset{(\mathrm{D})}{=} 0\{0, 1\} \cup 1\{0, 1\}
= \{00, 01, 10, 11\},
\end{equation}
where
(A) follows from Lemma \ref{lem:pref-sum} (i),
(B) follows from Lemma \ref{lem:pref-sum} (iii),
(C) follows from Lemma \ref{lem:pref-sum} (i),
and (D) follows from $F \in \mathscr{F}_2 \subseteq \mathscr{F}_1$.
This conflicts with $|\PREF^2_{F, i}| = 3$.
Therefore, (\ref{eq:h3kyk0pw4u1y}) holds.

By $|\PREF^2_{F, i}| = 3$, we have $\PREF^2_{F, i} = \{h_1h_2, \bar{h}_10, \bar{h}_11\}$ for some $h_1h_2 \in \mathcal{C}^2$.
By (\ref{eq:h3kyk0pw4u1y}), we have $g_1g_2 \in \PREF^2_{F, i} = \{h_1h_2, \bar{h}_10, \bar{h}_11\}$.
\begin{itemize}
\item If $g_1g_2 = h_1h_2$, then $\ddot{\gamma}(s_1) = 01$ by the third case of (\ref{eq:oezl1k2insm5}).
\item If $g_1g_2 = \bar{h}_10$, then $\ddot{\gamma}(s_1) = 10$ by the fourth case of (\ref{eq:oezl1k2insm5}).
\item If $g_1g_2 = \bar{h}_11$, then $\ddot{\gamma}(s_1) = 11$ by the fourth case of (\ref{eq:oezl1k2insm5}).
\end{itemize}
By the same argument, we have $\ddot{\gamma}(s'_1) = 01$ (resp. $10, 11$) if $g'_1g'_2 = h_1h_2$ (resp. $\bar{h}_10, \bar{h}_11$).
In particular, $\ddot{\gamma}(s_1) = \ddot{\gamma}(s'_1)$ holds if and only if $g_1g_2 = g'_1g'_2$.
Therefore, $\ddot{\gamma}(s_1) \neq \ddot{\gamma}(s'_1)$ is implied by (\ref{eq:5258bln24k8f}) as desired.
\end{itemize}

\end{itemize}
\end{proof}

\subsection{Proof of Lemma \ref{lem:tilde2-pref}}
\label{subsec:proof-lem:tilde2-pref}

\begin{proof}[Proof of Lemma \ref{lem:tilde2-pref}]
(Proof of (i)):
We consider the following two cases separately: (I) the case $|\PREF^2_{F, i}| = 3$; (II) the case $|\PREF^2_{F, i}| = 4$.
\begin{itemize}
\item[(I)] The case $|\PREF^2_{F, i}| = 3$:
Choose $\pmb{x} = x_1x_2\ldots x_{n} \in \mathcal{S}^{\ast}$ arbitrarily, and let $\gamma(s_1)\gamma(s_2)\ldots \gamma(s_{\rho})$ be the $\gamma$-decomposition of $f_i(x_1)$.
By $|\PREF^2_{F, i}| = 3$,
applying second, third, and fourth cases of (\ref{eq:oezl1k2insm5}), we have either $\ddot{\gamma}(s_1) \succeq 1$ or $\ddot{\gamma}(s_1) \succeq 01$, in particular,
$\ddot{f_i}^{\ast}(\pmb{x}) \not \succeq 00$.
This implies 
\begin{equation}
\label{eq:7zma735nyikb}
\PREF^2_{\ddot{F}, i} \subseteq \{01, 10, 11\}.
\end{equation}
By $|\PREF^2_{F, i}| = 3$, there exists $\pmb{c} = c_1c_2 \in \mathcal{C}^2$ such that 
\begin{equation}
\label{eq:he2rzbncp66u}
\PREF^2_{F, i} = \{c_1c_2, \bar{c_1}0, \bar{c_1}1\}.
\end{equation}
Then there exists $\pmb{x}' = x'_1x'_2\ldots x'_{n'} \in \mathcal{S}^{+}$ such that
\begin{equation}
\label{eq:behdwirdloju}
f^{\ast}_i(\pmb{x}') \succeq \pmb{c}.
\end{equation}
Let $\gamma(s'_1)\gamma(s'_2)\ldots \gamma(s'_{\rho'})$ be the $\gamma$-decomposition of $f_{i}(x'_1)$.
Now we show $|\gamma(s'_1)| \geq 2$ by deriving a contradiction for the following two cases separately: the case $|\gamma(s'_1)| = 0$ and the case $|\gamma(s'_1)| = 1$.
\begin{itemize}
\item If we assume $|\gamma(s'_1)| = 0$: We have
\begin{equation}
|\PREF^2_{F, i}| \overset{(\mathrm{A})}{\geq} |\bar{\PREF}^2_{F, i}| + |\PREF^2_{F, \trans_i(s'_1)}|
\overset{(\mathrm{B})}{\geq} |\bar{\PREF}^0_{F, i}| + |\PREF^2_{F, \trans_i(s'_1)}|
\overset{(\mathrm{C})}{\geq} 1 + |\PREF^2_{F, \trans_i(s'_1)}|
\overset{(\mathrm{D})}{\geq} 1 + 3 = 4,
\end{equation}
where
(A) follows from Lemma \ref{lem:pref-sum} (ii) and $|\gamma(s'_1)| = 0$,
(B) follows from Lemma \ref{lem:pref-inc2} (ii) (b),
(C) follows from Lemma \ref{lem:leaf} (iii) because $f_i$ is injective by Lemma \ref{lem:propertyT2}, 
and (D) follows from $F \in \mathscr{F}_2$.
This conflicts with $|\PREF^2_{F, i}| = 3$.

\item If we assume $|\gamma(s'_1)| = 1$: We have
\begin{equation}
\PREF^2_{F, i} \overset{(\mathrm{A})}{\supseteq} \bar{\PREF}^2_{F, i} \overset{(\mathrm{B})}{\supseteq} c_1\PREF^1_{F, i}(c_1) \overset{(\mathrm{C})}{=} c_1\PREF^1_{F, i}(f_i(s'_1)) \overset{(\mathrm{D})}{\supseteq} c_1\PREF^1_{F, \trans_i(s'_1)} \overset{(\mathrm{E})}{=} c_1\{0, 1\} \owns c_1\bar{c_2},
\end{equation}
where
(A) follows from Lemma \ref{lem:pref-sum} (i),
(B) follows from Lemma \ref{lem:pref-sum} (iii),
(C) follows since $c_1 = f_i(s'_1)$ by (\ref{eq:behdwirdloju}) and $|\gamma(s'_1)| = 1$,
(D) follows from Lemma \ref{lem:pref-sum} (i),
and (E) follows from $F \in \mathscr{F}_2 \subseteq \mathscr{F}_1$.
This conflicts with (\ref{eq:he2rzbncp66u}).
\end{itemize}

Hence, we have $|\gamma(s'_1)| \geq 2$ and thus $\gamma(s'_1) \succeq c_1c_2$ by (\ref{eq:behdwirdloju}).
Therefore, by the third case of (\ref{eq:oezl1k2insm5}), we obtain $\ddot{f}^{\ast}_i(\pmb{x}') \succeq \ddot{f}^{\ast}_i(x'_1) \succeq \ddot{\gamma}(s'_1) \succeq 01$, which leads to
\begin{equation}
\label{eq:8ic1sqqabh1q}
01 \in \PREF^2_{\ddot{F}, i}.
\end{equation}

Next, we show that
\begin{equation}
\label{eq:2r20ae7sycq0}
10, 11 \in \PREF^2_{\ddot{F}, i}.
\end{equation}
To prove it, we choose $a \in \mathcal{C}$ arbitrarily and show that $1a \in \PREF^2_{\ddot{F}, i}$.
Since $\bar{c_1}a \in \PREF^2_{F, i}$ by (\ref{eq:he2rzbncp66u}), 
there exists $\pmb{x}'' = x''_1x''_2\ldots x''_{n''} \in \mathcal{S}^{+}$ such that
\begin{equation}
\label{eq:4imtv3zkl0vb}
f^{\ast}_i(\pmb{x}'') \succeq \bar{c_1}a.
\end{equation}
Let $\gamma(s''_1)\gamma(s''_2)\ldots \gamma(s''_{\rho''})$ be the $\gamma$-decomposition of $f_{i}(x''_1)$.
We consider the following two cases separately: the case $|\gamma(s''_1)| \geq 2$ and the case $|\gamma(s''_1)| = 1$.
\begin{itemize}
\item The case $|\gamma(s''_1)| \geq 2$: 
Then we have $\gamma(s''_1) \succeq \bar{c_1}a$ by (\ref{eq:4imtv3zkl0vb}).
Hence, by $|\PREF^2_{F, i}| = 3$, $|\gamma(s''_1)| \geq 2$, and (\ref{eq:he2rzbncp66u}), we have $\ddot{\gamma}(s''_1) \succeq 1a$ applying the fourth case of (\ref{eq:oezl1k2insm5}).
Thus, we obtain $\ddot{f}^{\ast}_i(\pmb{x}'') \succeq \ddot{\gamma}(s''_1) \succeq 1a$, which leads to $1a \in \PREF^2_{\ddot{F}, i}$ as desired.
\item The case $|\gamma(s''_1)| = 1$: We have
\begin{equation}
\PREF^2_{\ddot{F}, i}
\overset{(\mathrm{A})}{\supseteq} \bar{\PREF}^2_{\ddot{F}, i}
\overset{(\mathrm{B})}{\supseteq}  1\PREF^1_{\ddot{F}, i}(1)
\overset{(\mathrm{C})}{=}  1\PREF^1_{\ddot{F}, i}(\ddot{\gamma}(s''_1))
\overset{(\mathrm{D})}{\supseteq}  1\PREF^1_{\ddot{F}, \ddot{\trans}_i(s''_1)}
\overset{(\mathrm{E})}{=}  1\{0, 1\} \owns 1a,
\end{equation}
where
(A) follows from Lemma \ref{lem:pref-sum} (i),
(B) follows from Lemma \ref{lem:pref-sum} (iii),
(C) is obtained by applying the second case of (\ref{eq:oezl1k2insm5}) by $|\PREF^2_{F, i}| = 3$ and $|\gamma(s''_1)| = 1$,
(D) follows from Lemma \ref{lem:pref-sum} (i),
and (E) follows from $F \in \mathscr{F}_2 \subseteq \mathscr{F}_1$.

\end{itemize}
Therefore, we conclude that (\ref{eq:2r20ae7sycq0}) holds.
By (\ref{eq:7zma735nyikb}), (\ref{eq:8ic1sqqabh1q}), and (\ref{eq:2r20ae7sycq0}), we obtain $\PREF^2_{\ddot{F}, i} = \{01, 10, 11\}$ as desired.

\item[(II)] The case $|\PREF^2_{F, i}| = 4$:
We consider the following two cases separately: (II-A) the case $\assign_{F, i}(\lambda) \neq \emptyset$; (II-B) the case $\assign_{F, i}(\lambda) = \emptyset$.
\begin{itemize}
\item[(II-A)] The case $\assign_{F, i}(\lambda) \neq \emptyset$:
Since $f_i$ is injective by Lemma \ref{lem:propertyT2}, we can choose $s \in \mathcal{S}$ such that $\assign_{F, i}(\lambda) = \{s\}$.
Also, we have $\bar{\PREF}^0_{F, i} \neq \emptyset$ applying Lemma \ref{lem:leaf} (iii).
Hence, by Lemma \ref{cor:pref-sum}, we have $|\PREF^2_{F, \trans_i(s)}| \leq 3$.
In particular, it holds that $|\PREF^2_{F, \trans_i(s)}| = 3$ by $F \in \mathscr{F}_2$.
Therefore, by the result of the case (I), we obtain
\begin{equation}
\label{eq:5jqigczmqbuk}
\PREF^2_{\ddot{F}, \trans_i(s)} = \{01, 10, 11\}.
\end{equation}

Since $f_i$ is injective, we can choose $s' \in \mathcal{S}$ such that $s' \neq \lambda$.
Let $\gamma(s'_1)\gamma(s'_2)\ldots \gamma(s'_{\rho'})$ be the $\gamma$-decomposition of $f_i(s')$. 
By Lemma \ref{lem:dot-length} (i) and $\assign_{F, i}(\lambda) \neq \emptyset$, we have
\begin{equation}
\label{eq:ai6o2dog2hs4}
\gamma(s'_1) = \lambda.
\end{equation}
Note that $\rho' \geq 2$ holds by (\ref{eq:ai6o2dog2hs4}) and $s'_{\rho'} = s' \neq \lambda$.
We have
\begin{eqnarray}
\ddot{f}_i(s') &=& \ddot{\gamma}(s'_1)\ddot{\gamma}(s'_2)\ldots \ddot{\gamma}(s'_{\rho'})\\
&\succeq& \ddot{\gamma}(s'_1)\ddot{\gamma}(s'_2)\\
&\overset{(\mathrm{A})}{=}& \ddot{\gamma}(s'_2)\\
&\overset{(\mathrm{B})}{\succeq}& 00,
\end{eqnarray}
where
(A) follows from (\ref{eq:ai6o2dog2hs4}) and Lemma \ref{lem:tilde2-equal} (i),
and (B) follows from the fifth case of (\ref{eq:oezl1k2insm5}).

Hence, we have
\begin{equation}
\label{eq:y9joty2syliq}
00 \in \bar{\PREF}^2_{\ddot{F}, i}.
\end{equation}
We obtain
\begin{equation}
\PREF^2_{\ddot{F}, i} \overset{(\mathrm{A})}{\supseteq} \PREF^2_{\ddot{F}, \trans_i(s)}  \cup \bar{\PREF}^2_{\ddot{F}, i} \overset{(\mathrm{B})}{\supseteq} \{01, 10, 11\} \cup \{00\} = \{00, 01, 10, 11\}
\end{equation}
as desired, where
(A) follows from Lemma \ref{lem:pref-sum} (i),
and (B) follows from (\ref{eq:5jqigczmqbuk}) and (\ref{eq:y9joty2syliq}).

\item[(II-B)] The case $\assign_{F, i}(\lambda) = \emptyset$:
It suffices to show that $\PREF^2_{\ddot{F}, i} \supseteq \PREF^2_{F, i}$ since $|\PREF^2_{F, i}| = 4$.
Choose $\pmb{c} = c_1c_2 \in \PREF^2_{F, i} = \{00, 01, 10, 11\}$ arbitrarily.
Then there exists $\pmb{x} = x_1x_2\ldots x_{n} \in \mathcal{S}^{+}$ such that
\begin{equation}
\label{eq:szvhp47q8soi}
f^{\ast}_i(\pmb{x}) \succeq \pmb{c}.
\end{equation}
Let $\gamma(s_1)\gamma(s_2)\ldots \gamma(s_{\rho})$ be the $\gamma$-decomposition of $f_i(x_1)$.
We consider the following two cases separately: the case $|\gamma(s_1)| \geq 2$ and the case $|\gamma(s_1)| = 1$.
Note that we can exclude the case $|\gamma(s_1)| = 0$ since $\assign_{F, i}(\lambda) = \emptyset$.

\begin{itemize}
\item The case $|\gamma(s_1)| \geq 2$:
We have
\begin{equation}
\ddot{f}_{i}(x_1) \succeq \ddot{\gamma}(s_1) \overset{(\mathrm{A})}{=} \gamma(s_1) \overset{(\mathrm{B})}{\succeq} \pmb{c}, 
\end{equation}
where
(A) follows from $|\PREF^2_{F, i}| = 4$ and the first case of (\ref{eq:oezl1k2insm5}), 
and (B) follows from (\ref{eq:szvhp47q8soi}) and $|\gamma(s_1)| \geq 2$.
This implies $\pmb{c} \in \PREF^2_{\ddot{F}, i}$ as desired.

\item The case $|\gamma(s_1)| = 1$:
We have
\begin{equation}
\label{eq:r5zt0mjs9ah3}
\ddot{f}_i(s_1) = \ddot{\gamma}(s_1) \overset{(\mathrm{A})}{=} \gamma(s_1) \overset{(\mathrm{B})}{=} c_1,
\end{equation}
where
(A) follows from $|\PREF^2_{F, i}| = 4$ and the first case of (\ref{eq:oezl1k2insm5}),
and (B) follows from (\ref{eq:szvhp47q8soi}) and $|\gamma(s_1)| = 1$.

Put $j \coloneqq \trans_i(s_1)$. By Lemma \ref{lem:longest}, we can choose the longest sequence $\pmb{x}' = x'_1x'_2\ldots x'_{n'} \in \mathcal{S}^{+}$ such that $f^{\ast}_j(\pmb{x}') = \lambda$.
Then we have $\assign_{F, \trans^{\ast}_j(\pmb{x}')}(\lambda) = \emptyset$.
Also, we have $|\PREF^2_{F, \trans^{\ast}_j(\pmb{x}')}| \geq 3$ by $F \in \mathscr{F}_2$.
In particular, we have at one of the following conditions (a) and (b).
\begin{enumerate}[(a)]
\item $|\PREF^2_{F, \trans^{\ast}_j(\pmb{x}')}| = 3$.
\item $|\PREF^2_{F, \trans^{\ast}_j(\pmb{x}')}| = 4$ and $\assign_{F, \trans^{\ast}_j(\pmb{x}')}(\lambda) = \emptyset$.
\end{enumerate}
Therefore, from the cases (I) and (II-A) proven above,
we have $\PREF^2_{\ddot{F}, \ddot{\trans}^{\ast}_j(\pmb{x}')} \supseteq \{01, 10, 11\}$, which leads to
\begin{equation}
\label{eq:6omo9o8vukce}
\PREF^1_{\ddot{F}, \ddot{\trans}^{\ast}_j(\pmb{x}')} = \{0, 1\}
\end{equation}
by Lemma \ref{lem:pref-inc} (i).
Thus, we have
\begin{eqnarray}
\PREF^2_{\ddot{F}, i}
&\overset{(\mathrm{A})}{\supseteq}& \bar{\PREF}^2_{\ddot{F}, i} \overset{(\mathrm{B})}{\supseteq} c_1\PREF^1_{\ddot{F}, i}(c_1) \overset{(\mathrm{C})}{=} c_1\PREF^1_{\ddot{F}, i}(\ddot{f}_i(s_1)) \nonumber\\
&\overset{(\mathrm{D})}{\supseteq}& c_1\PREF^1_{\ddot{F}, j}
\overset{(\mathrm{D})}{\supseteq} c_1\PREF^1_{\ddot{F}, \ddot{\trans}^{\ast}_j(x'_1)} 
\overset{(\mathrm{D})}{\supseteq} c_1\PREF^1_{\ddot{F}, \ddot{\trans}^{\ast}_j(x'_1x'_2)}
\overset{(\mathrm{D})}{\supseteq} \cdots
\overset{(\mathrm{D})}{\supseteq} c_1\PREF^1_{\ddot{F}, \ddot{\trans}^{\ast}_j(\pmb{x}')} \nonumber\\
&\overset{(\mathrm{E})}{=}& c_1\{0, 1\} \owns c_1c_2 = \pmb{c},
\end{eqnarray}
where
(A) follows from Lemma \ref{lem:pref-sum} (i),
(B) follows from Lemma \ref{lem:pref-sum} (iii),
(C) follows from (\ref{eq:r5zt0mjs9ah3}),
(D)s follow from Lemma \ref{lem:pref-sum} (i),
and (E) follows from (\ref{eq:6omo9o8vukce}).
Therefore, we conclude that $\PREF^2_{\ddot{F}, i} \supseteq \PREF^2_{F, i} = \{00, 01, 10, 11\}$ as desired.
\end{itemize}

\end{itemize}
\end{itemize}

(Proof of (ii)):
We have
\begin{eqnarray}
\bar{\PREF}^0_{F, i}(f_i(s)) \neq \emptyset
&\overset{(\mathrm{A})}{\iff}& \bar{\PREF}^2_{F, i}(f_i(s)) \neq \emptyset\\
&\iff& {}^{\exists}\pmb{x} \in \mathcal{S}^{+}; {}^{\exists}\pmb{c} \in \mathcal{C}^2;  (f^{\ast}_i(\pmb{x}) \succeq f_i(s)\pmb{c}, f_i(x_1) \succ f_i(s))\\
&\overset{(\mathrm{B})}{\iff}& {}^{\exists}\pmb{x} \in \mathcal{S}^{+};  {}^{\exists}\pmb{c} \in \mathcal{C}^2; (\ddot{f}^{\ast}_i(\pmb{x}) \succeq \ddot{f}_i(s)\pmb{c}, \ddot{f}_i(x_1) \succ \ddot{f}_i(s))\\
&\iff& \bar{\PREF}^2_{\ddot{F}, i}(\ddot{f}_i(s)) \neq \emptyset, \label{eq:gj94kfo0jtgl}
\end{eqnarray}
where
(A) follows from Lemma \ref{lem:pref-inc2} (ii) (a),
and (B) follows from Lemma \ref{lem:tilde2-equal} (iii).

We consider the following two cases separately: the case $\bar{\PREF}^0_{F, i}(f_i(s)) = \emptyset$ and the case $\bar{\PREF}^0_{F, i}(f_i(s)) \neq \emptyset$.
\begin{itemize}
\item The case $\bar{\PREF}^0_{F, i}(f_i(s)) = \emptyset$:
By (\ref{eq:gj94kfo0jtgl}), the condition $\bar{\PREF}^0_{F, i}(f_i(s)) = \emptyset$ is equivalent to $\bar{\PREF}^2_{\ddot{F}, i}(\ddot{f}_i(s)) = \emptyset$ as desired.

\item The case $\bar{\PREF}^0_{F, i}(f_i(s)) \neq \emptyset$:
Then since $\bar{\PREF}^2_{\ddot{F}, i}(\ddot{f}_i(s)) \neq \emptyset$ holds by (\ref{eq:gj94kfo0jtgl}),
 it suffices to show that $\bar{\PREF}^2_{\ddot{F}, i}(\ddot{f}_i(s)) \subseteq \{00\}$.
Moreover, to prove this, it suffices to show that for any $\pmb{x} = x_1x_2\ldots x_{n} \in \mathcal{S}^{+}$ such that $\ddot{f}_i(x_1) \succ \ddot{f}_i(s)$, we have $\ddot{f}^{\ast}_i(\pmb{x}) \succeq \ddot{f}_i(s)00$.

Choose $\pmb{x} = x_1x_2\ldots x_{n} \in \mathcal{S}^{+}$ such that
\begin{equation}
\label{eq:alh8bogogq0g}
\ddot{f}_i(x_1) \succ \ddot{f}_i(s).
\end{equation}
Let $\gamma(s_1)\gamma(s_2)\ldots \gamma(s_{\rho})$ be the $\gamma$-decomposition of $f_i(x_1)$.
Because $f_i(x_1) \succ f_i(s)$ holds by (\ref{eq:alh8bogogq0g}) and Lemma \ref{lem:tilde2-equal} (iii), we have $s = s_r$ and $\ddot{f}_i(s) = \ddot{\gamma}(s_1)\ddot{\gamma}(s_2)\ldots \ddot{\gamma}(s_r)$ for some $r = 1, 2, \ldots, \rho-1$.
For such $r$, we have
\begin{eqnarray}
\ddot{f}^{\ast}_i(\pmb{x})
&\succeq& \ddot{f}_i(x_1)\\
&=& \ddot{\gamma}(s_1)\ddot{\gamma}(s_2)\ldots \ddot{\gamma}(s_r)\ddot{\gamma}(s_{r+1}) \ldots\ddot{\gamma}(s_{\rho})\\
&\succeq& \ddot{f}_i(s)\ddot{\gamma}(s_{r+1})\\
&\overset{(\mathrm{A})}{\succeq}& \ddot{f}_i(s)00 \label{eq:tw0s28r5grc3}
\end{eqnarray}
as desired, where (A) follows from the fifth case of (\ref{eq:oezl1k2insm5}).
\end{itemize}
\end{proof}

\subsection{List of Notations}
\label{subsec:notation}

\begin{longtable}{cp{14.5cm}}
  $|\mathcal{A}|$ & the cardinality of a set $\mathcal{A}$, defined at the beginning of Section \ref{sec:preliminary}. \\
  $\mathcal{A}^k$ & the set of all sequences of length $k$ over a set $\mathcal{A}$, defined at the beginning of Section \ref{sec:preliminary}. \\
  $\mathcal{A}^{\ast}$ & the set of all sequences of finite length over a set $\mathcal{A}$, defined at the beginning of Section \ref{sec:preliminary}.\\
  $\mathcal{A}^{+}$ & the set of all sequences of finite positive length over a set $\mathcal{A}$, defined at the beginning of Section \ref{sec:preliminary}.\\
  $a_{F, i}$ & defined in Definition \ref{def:tilde}.\\
  $\mathcal{C}$ & the coding alphabet $\mathcal{C} = \{0, 1\}$, defined at the beginning of Section \ref{sec:preliminary}. \\
  $d_{F, i}$ & defined in (\ref{eq:emksr78o7m54}). \\
  $f^{\ast}_i$ & defined in Definition \ref{def:f_T}. \\
  $F$ & shorthand for a code-tuple $F(f_0, f_1, \ldots, f_{m-1}, \trans_0, \trans_1, \ldots, \trans_{m-1})$, also written as $F(f, \trans)$, defined after Definition \ref{def:treepair}.\\
  $|F|$ & the number of code tables of $F$, defined after Definition \ref{def:treepair}. \\
  $[F]$ & shorthand for $[|F|] = \{0, 1, 2, \ldots, |F|-1\}$, defined below Definition \ref{def:treepair}.\\
  $\widehat{F}$ & defined in Definition \ref{def:rotation}. \\
  $\dot{F}$ & defined in Definition \ref{def:tilde}.\\
  $\ddot{F}$ & defined in Definition \ref{def:tilde2}.\\
  $\mathscr{F}^{(m)}$ & the set of all $m$-code-tuples, defined after Definition \ref{def:treepair}.\\
  $\mathscr{F}$ & the set of all code-tuples, defined after Definition \ref{def:treepair}.\\
  $\mathscr{F}_{\AIFV}$ & the set of all AIFV codes, defined in Definition \ref{def:subsetTaifv}. \\
  $\mathscr{F}_{\ext}$ & the set of all extendable code-tuples, defined in Definition \ref{def:F_ext}. \\
  $\mathscr{F}_{k\hdec}$ & the set of all $k$-bit delay decodable code-tuples, defined in Definition \ref{def:k-bitdelay}. \\
  $\mathscr{F}_{\opt}$ & the set of all optimal code-tuples, defined in Definition \ref{def:optimalset}. \\
  $\mathscr{F}_{\reg}$ & the set of all regular code-tuples, defined in Definition \ref{def:regular}. \\
  $\mathscr{F}_0$ & $\{F \in \mathscr{F}_{\reg} \cap \mathscr{F}_{2\hdec}: {}^{\forall}i \in [F]; \PREF^1_{F, i} \neq \emptyset \} = \mathscr{F}_{\reg} \cap \mathscr{F}_{\ext} \cap \mathscr{F}_{2\hdec}$, defined in Definition \ref{def:subsetT0}.\\
  $\mathscr{F}_1$ & $\{F \in \mathscr{F}_{\reg} \cap \mathscr{F}_{2\hdec}: {}^{\forall}i \in [F]; \PREF^1_{F, i} = \{0, 1\} \}$, defined in Definition \ref{def:classes}.\\
  $\mathscr{F}_2$ & $\{F \in \mathscr{F}_{\reg} \cap \mathscr{F}_{2\hdec}: {}^{\forall}i \in [F]; |\PREF^2_{F, i}| \geq 3 \}$, defined in Definition \ref{def:classes}.\\
  $\mathscr{F}_3$ & $\{F \in \mathscr{F}_{\reg} \cap \mathscr{F}_{2\hdec}: {}^{\forall}i \in [F]; \PREF^2_{F, i} \supseteq \{01, 10, 11\} \}$, defined in Definition \ref{def:classes}.\\
  $\mathscr{F}_4$ & $\{F \in \mathscr{F}_{\reg} \cap \mathscr{F}_{2\hdec} \cap \mathscr{F}^{(2)}: \PREF^2_{F, 0} = \{00, 01, 10, 11\}, \PREF^2_{F, 1} = \{01, 10, 11\}\}$, defined in Definition \ref{def:classes}.\\
  $L(F)$ & the average codeword length of a code-tuple $F$, defined in Definition \ref{def:evaluation}. \\
  $L_i(F)$ & the average codeword length of the $i$-th code table of $F$, defined in Definition \ref{def:evaluation}. \\
  $[m]$ & $\{0, 1, 2, \ldots, m-1\}$, defined at the beginning of Section \ref{sec:introduction}. \\
  $\mathcal{M}_F$ & $\{i \in [F] : |\PREF^2_{F, i}| = 2\}$, defined in Lemma \ref{lem:transA-arg4}.\\
  $\mathcal{P}^k_{F, i}$ & $\{\pmb{c} \in \mathcal{C}^k : \pmb{x} = x_1x_2\ldots x_{n} \in \mathcal{S}^{+}, f^{\ast}_i(\pmb{x}) \succeq \pmb{b}\pmb{c}, f_i(x_1) \succeq \pmb{b} \}$, defined in Definition \ref{def:pref}.\\
  $\bar{\mathcal{P}}^k_{F, i}$ & $\{\pmb{c} \in \mathcal{C}^k : \pmb{x} = x_1x_2\ldots x_{n} \in \mathcal{S}^{+}, f^{\ast}_i(\pmb{x}) \succeq \pmb{b}\pmb{c}, f_i(x_1) \succ \pmb{b} \}$, defined in Definition \ref{def:pref}.\\
  $\pref(\pmb{x})$ & the sequence obtained by deleting the last letter of $\pmb{x}$, defined at the beginning of Section \ref{sec:preliminary}. \\
  $Q(F)$ & the transition probability matrix, defined in Definition \ref{def:transprobability}.\\
  $Q_{i, j}(F)$ & the transition probability, defined in Definition \ref{def:transprobability}.\\
  $\mathcal{S}$ & the source alphabet, defined at the beginning of Section \ref{sec:preliminary}.\\
  $\assign_{F, i}$ & $\assign_{F, i}(\pmb{b}) \coloneqq \{s \in \mathcal{S} : f_i(s) = \pmb{b}\}$, defined in Definition \ref{def:assign}. \\
  $\pmb{x} \preceq \pmb{y}$ & \pmb{x} is a prefix of \pmb{y}, defined at the beginning of Section \ref{sec:preliminary}.\\
    $\pmb{x} \prec \pmb{y}$ & $\pmb{x} \preceq \pmb{y}$ and $\pmb{x} \neq \pmb{y}$, defined at the beginning of Section \ref{sec:preliminary}. \\
  $\suff(\pmb{x})$ & the sequence obtained by deleting the first letter of $\pmb{x}$, defined at the beginning of Section \ref{sec:preliminary}. \\
  $|\pmb{x}|$ & the length of a sequence $\pmb{x}$, defined at the beginning of Section \ref{sec:preliminary}.\\
  $\gamma(s_r)$ & defined in Definition \ref{def:gamma}.\\
  $\lambda$ & the empty sequence, defined at the beginning of Section \ref{sec:preliminary}.\\
  $\mu(s)$ & the probability of occurrence of symbol $s$, defined at the beginning of Subsection \ref{sec:preliminary}. \\
  $\pmb{\pi}(F)$ & defined in Definition \ref{def:regular}.\\
  $\sigma$ & the alphabet size $|\mathcal{S}|$, defined at the beginning of Section \ref{sec:preliminary}.\\
  $\trans^{\ast}_i$ & defined in Definition \ref{def:f_T}. \\
\end{longtable}


\nocite{*}
\bibliographystyle{IEEE}

%

\end{document}